\pdfoutput=1

\documentclass[a4paper]{report} 

\usepackage[utf8]{inputenc}
\usepackage[ngerman,american]{babel}
\usepackage{amssymb}
\usepackage{amsthm}
\usepackage{nicefrac} 
\usepackage{enumerate}
\usepackage{graphicx}
\usepackage{tabularx}
\usepackage{booktabs}
\usepackage[numbers]{natbib} 
\usepackage[tworuled,vlined]{algorithm2e}
\usepackage{subfig}
\usepackage{color}
\usepackage{microtype} 
\usepackage[pagebackref]{hyperref}
\usepackage{breakurl}
\usepackage{paralist}
\usepackage{mathtools}

\addto\extrasamerican{%
}

\hypersetup{breaklinks=true}

\makeatletter
\renewcommand*{\NAT@spacechar}{~}
\makeatother

\newcolumntype{L}{>{\raggedright\arraybackslash}X}%
\newcolumntype{C}{>{\centering\arraybackslash}X}%
\newcolumntype{R}{>{\raggedleft\arraybackslash}X}%

\usepackage{ifpdf}
\ifpdf 
\else 
\fi 

\renewcommand{\cite}{\citep}
\newcommand{\Eu}{Eulerian }
\newcommand{\EG}{\Eu graph}
\newcommand{\EE}{\Eu extension}

\newcommand{\claP}{\text{P}}
\newcommand{\claNP}{\text{NP}}
\newcommand{\clacoNP}{\text{coNP}}
\newcommand{\claPH}{\text{PH}}
\newcommand{\claFPT}{\text{FPT}}
\newcommand{\claXP}{\text{XP}}
\newcommand{\claW}[1]{\text{W[\ensuremath{#1}]}}
\newcommand{\claM}[1]{\text{M[\ensuremath{#1}]}}
\newcommand{\claS}[1]{$\Sigma^p_{#1}$}
\newcommand{\NPh}{\claNP-hard}
\newcommand{\NPhs}{\claNP-hardness}
\newcommand{\NPHs}{\claNP-Hardness}
\newcommand{\NPc}{\claNP-complete}

\newcommand{\slashpoly}{\text{$/$poly}}

\renewcommand{\emptyset}{\varnothing}
\newcommand{\quer}[1]{\overline{#1}}

\newcommand{\ptSATs}{\textsc{3SAT}}
\newcommand{\pEE}{\textsc{\Eu Extension}}
\newcommand{\pEEs}{\textsc{EE}}

\newcommand{\pdDEE}{\textsc{$d$-Dimensional \pEE}}
\newcommand{\pxDEE}[1]{\textsc{$#1$-Dimensional \pEE}}
\newcommand{\pxDEEs}[1]{\textsc{$#1$D\pEEs}}
\newcommand{\pWMEE}{\textsc{Eulerian Extension}}
\newcommand{\pWMEEs}{\textsc{EE}}
\newcommand{\pWMEEA}{\textsc{\pWMEE{} with Advice}}
\newcommand{\pWMEEAs}{\textsc{\pWMEEs A}}
\newcommand{\pWMEECLA}{\textsc{\pWMEE{} with Cycle-free Advice}}
\newcommand{\pWMEECLAs}{\textsc{\pWMEEs$\emptyset$A}}
\newcommand{\pWMEECA}{\textsc{\pWMEE{} with Minimal Connecting Advice}}
\newcommand{\pWMEECAs}{\textsc{\pWMEEs CA}}
\newcommand{\pWMEECCLA}{\textsc{\pWMEE{} with Cycle-free Minimal Connecting Advice}}
\newcommand{\pWMEECCLAs}{\textsc{\pWMEEs $\emptyset$CA}}
\newcommand{\WMEE}{\pWMEE{}}
\newcommand{\WMEEs}{\pWMEEs{}}

\newcommand{\CBMs}{\textsc{CBM}}
\newcommand{\pCBM}{\textsc{Conjoining Bipartite Matching}}
\newcommand{\pCBMs}{\textsc{CBM}}
\newcommand{\pRP}{\textsc{Rural Postman}}
\newcommand{\pRPs}{\textsc{RP}}
\newcommand{\pHC}{\textsc{Hamiltonian Cycle}}
\newcommand{\pHCs}{\textsc{HC}}
\newcommand{\pKML}{\textsc{Switch Set Cover}}
\newcommand{\pKMLs}{\textsc{SSC}}
\newcommand{\kiste}{switch}
\newcommand{\kisten}{switches}
\newcommand{\liste}{position}
\newcommand{\listen}{positions}

\newcommand{\meta}[2]{\mathbb{C}_{#1}(#2)}
\newcommand{\comp}[1]{\mathbb{C}_{#1}}

\newcommand{\SPP}{shortest-path preprocessing}
\newcommand{\wf}{\omega}

\DeclareMathOperator{\balance}{balance}
\DeclareMathOperator{\indeg}{indeg}
\DeclareMathOperator{\outdeg}{outdeg}
\DeclareMathOperator{\length}{length}
\DeclareMathOperator{\shortcut}{shortcut}

\DeclareMathOperator{\conn}{connect}

\DeclareMathOperator{\minpath}{minpath}
\DeclareMathOperator{\solveconnected}{solve\_connected}

\DeclareMathOperator{\bigO}{O}

\DeclareMathOperator{\binary}{binary}
\DeclareMathOperator{\true}{true}
\DeclareMathOperator{\false}{false}

\newtheorem{theorem}{Theorem}[section]
\newtheorem{lemma}{Lemma}[section]
\newtheorem{corollary}{Corollary}[section]
\newtheorem{rrule}{Reduction Rule}[section]
\newtheorem{observation}{Observation}[section]

\theoremstyle{definition}
\newtheorem{definition}{Definition}[section]
\newtheorem{construction}{Construction}[section]
\newtheorem{transformation}{Transformation}[section]
\newtheorem{example}{Example}[section]

\newlength{\probwidth}
\newcommand{\prob}[5]{%
  \begin{center}
    \begin{minipage}{\probwidth}
      #1
      \begin{compactdesc}
      \item[#2]#3
      \item[#4]#5
      \end{compactdesc}
    \end{minipage}
  \end{center}
}
  
\newcommand{\decprob}[3]{\prob{#1}{Input:}{#2}{Question:}{#3}}

\setdefaultenum{(1)}{(a)}{(i)}{(A)}
\newenvironment{lemenum}{\begin{asparaenum}[(i)]}{\end{asparaenum}}
\newcommand{\enuref}[1]{(\ref{#1})}

\begin{document}

\selectlanguage{ngerman}
\begin{titlepage}
\centering
\noindent \rule{\textwidth}{0.5pt}

\vspace{1em}
\huge On Making Directed Graphs Eulerian

\normalsize\vspace{1.12\topsep}\large

\normalsize\vspace{\topsep}\Large
\textsc{Manuel Sorge}\\
\rule{\textwidth}{0.5pt}

\vfill
\Large Diplomarbeit\\
\normalsize zur Erlangung des akademischen Grades\\ Diplom-Informatiker\\
\vspace{\topsep}
\includegraphics[width=5cm,trim=0cm 4cm 0cm 0cm,clip=true]{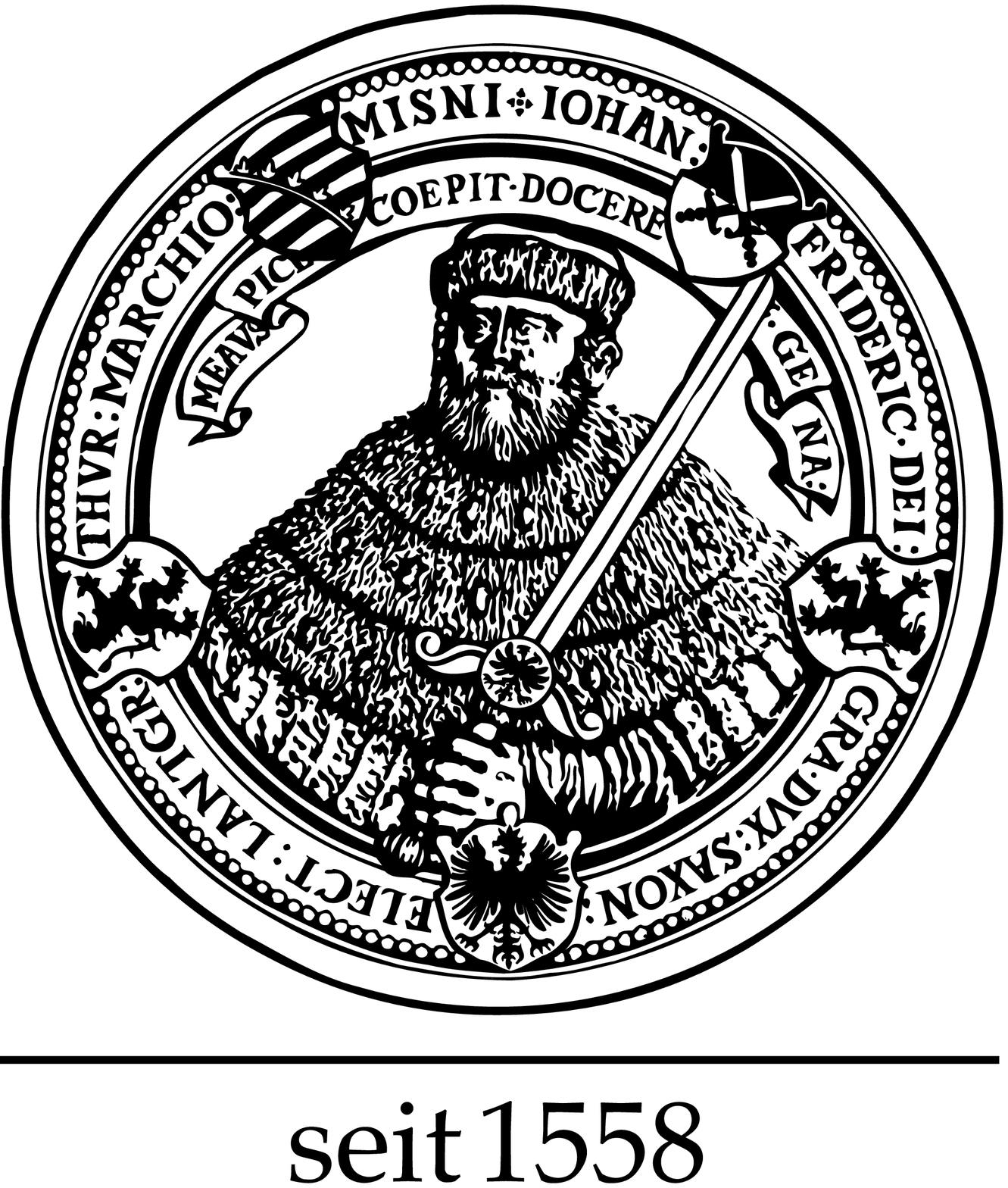}\\
\vspace{\topsep}
Friedrich-Schiller-Universität Jena\\
  Institut für Informatik\\
  Theoretische Informatik I / Komplexitätstheorie

\vfill
\begin{tabularx}{\textwidth}{lXr}
  Betreuung: & & \"Uberarbeitete Version\\
  Dipl.-Inf. René van Bevern & & Eingereicht von Manuel Sorge, \\
  Prof. Dr. Rolf Niedermeier & & geb. am 28.04.1986 in Zittau. \\
  Dipl.-Inf. Mathias Weller  & & Jena, den \today
\end{tabularx}

\end{titlepage}
\newpage
\mbox{}
\vfill
\paragraph{Zusammenfassung.}
Einen gerichteten Graphen nennt man \emph{Eulersch}, wenn er eine Tour enthält, die jede gerichtete Kante genau einmal besucht. Wir untersuchen das Problem \pWMEE{}~(\pWMEEs{}) in dem ein gerichteter Multigraph~$G$ und eine Gewichtsfunktion gegeben ist, und gefragt wird, ob~$G$ durch Hinzufügen gerichteter Kanten, deren Gesamtgewicht einen Grenzwert nicht überschreitet, Eulersch gemacht werden kann. Dieses Problem ist motiviert durch Anwendungen im Erstellen von Fahrplänen für Fahrzeuge und Abfolgeplänen für Fließbandarbeit. Allerdings ist das Problem \pWMEEs{} NP-schwer; deshalb analysieren wir es mit Hilfe von Parametrisierter Komplexität. Die Parametrisierte Komplexität eines Problems hängt nicht nur von der Eingabelänge, sondern auch von anderen Eigenschaften der Eingabe ab. Diese Eigenschaften nennt man ``Parameter''. \citet{DMNW10} zeigten, dass \pWMEEs{} in~$\bigO(4^kn^4)$~Zeit gelöst werden kann. Hier bezeichnet~$k$ den Parameter ``Anzahl der gerichteten Kanten, die hinzugefügt werden müssen''. In dieser Arbeit analysieren wir \pWMEEs{} mit den (kleineren) Parametern ``Anzahl~$c$ der verbunden Komponenten im Eingabegraph'' und ``Summe~$b$ aller $\indeg(v) - \outdeg(v)$ über alle Knoten~$v$ für die dieser Wert positiv ist''. Wir zeigen dass es einen Lösungsalgorithmus für \pWMEEs{} gibt, dessen Laufzeit den Term~$4^{c\log(bc^2)}$ als einzigen superpolynomiellen Term beinhaltet. Um diesen Algorithmus zu erhalten, machen wir mehrere Beobachtungen über die Mengen gerichteter Kanten, die dem Eingabegraph hinzugefügt werden müssen, um ihn Eulersch zu machen. Aufbauend auf diesen Beobachtungen geben wir außerdem eine Reformulierung von \pWMEEs{} in einem Matchingkontext. Diese Matchingformulierung könnte ein bedeutendes Werkzeug sein, um zu klären, ob \pWMEEs{} in Laufzeit gelöst werden kann, deren superpolynomieller Anteil nur von~$c$ abhängt. Außerdem betrachten wir Vorverarbeitungsalgorithmen polynomieller Laufzeit für \pWMEEs{}, und zeigen, dass diese keine Instanzen erzeugen können, deren Größe polynomiell nur von einem der Parameter~$b, c, k$ abhängt, es sei denn~$\clacoNP \subseteq \claNP \slashpoly$.
\vfill
\mbox{}
\newpage
\mbox{}
\vfill
\selectlanguage{american}
\paragraph{Abstract.}
A directed graph is called \emph{Eulerian}, if it contains a tour that traverses every arc in the graph exactly once. We study the problem of \pWMEE{}~(\pWMEEs{}) where a directed multigraph~$G$ and a weight function is given and it is asked whether~$G$ can be made \Eu by adding arcs whose total weight does not exceed a given threshold. This problem is motivated through applications in vehicle routing and flowshop scheduling. However, \pWMEEs{} is \NPh{} and thus we use the parameterized complexity framework to analyze it. In parameterized complexity, the running time of algorithms is considered not only with respect to input length, but also with respect to other properties of the input---called ``parameters''. 
\citet{DMNW10} proved that \pWMEEs{} can be solved in~$\bigO(4^kn^4)$~time, where~$k$ denotes the parameter ``number of arcs that have to be added''. In this thesis, we analyze \pWMEEs{} with respect to the (smaller) parameters ``number~$c$ of connected components in the input graph'' and ``sum~$b$ over $\indeg(v) - \outdeg(v)$ for all vertices~$v$ in the input graph where this value is positive''. We prove that there is an algorithm for \pWMEEs{} whose running time is polynomial except for the term~$4^{c\log(bc^2)}$. To obtain this result, we make several observations about the sets of arcs that have to be added to the input graph in order to make it Eulerian. We build upon these observations to restate \EE{} in a matching context. This matching formulation of \pWMEEs{} might be an important tool to solve the question of whether \pWMEEs{} can be solved within running time whose superpolynomial part depends only on~$c$. We also consider polynomial time preprocessing routines for \pWMEEs{} and show that these routines cannot yield instances whose size depends polynomially only on either of the parameters~$b, c, k$ unless~$\clacoNP \subseteq \claNP \slashpoly$.
\vfill
\mbox{}
\pagebreak

\tableofcontents

\chapter{Introduction}\label{sec:intro}
The notion of \Eu graphs dates back to Leonhard Euler. In 1735, he solved the question of whether there is a tour through the city of K\"onigsberg such that every bridge is crossed once and only once~\cite{Eul36}. The term ``\Eu tour'' has been coined for such a tour. At Euler's time, the city of Königsberg had seven bridges across the river Pregel and it turned out that an \Eu tour did not exist. Later, in the nineteenth century, a railway bridge has been built, making such a tour feasible~\cite{Wil86}.

In this thesis, we study the problem where, given a city, it is asked what is a minimum-cardinality set of bridges that have to be built such that the city allows for an \Eu tour? Typically, one aims for minimizing the costs for bridge-building. Thus, we mainly focus on a weighted version of this problem. Also, instead of ordinary bridges, we consider bridges that only allow traffic in one way. We call the problem of making a city allow for an \Eu tour by building one-way bridges, weighted according to a cost function, \pWMEE{}.

The problem of \pWMEE{} is also well-motivated through wholly different approaches. Recently, it has been shown by \citet{HJM09} that some sequencing problems can be solved with \pWMEE{}. An example for such a sequencing problem is ``no-wait flowshop'', where a schedule of jobs is sought, each processing on a fixed succession of machines, such that no waiting time occurs for any job between the processing on two subsequent machines. Such problems arise, for instance, in steel production~\cite{HS96}.

Another problem that has strong ties to \pWMEE{} is \pRP{}. There, a postman's tour in a city is sought such that a given subset of streets in the city is serviced. \citet{DMNW10} have proven that \pRP{} is equivalent to \pWMEE{}. \pRP{} can be applied, for example, in routing of snow plowing vehicles~\cite{EGL95}.

Unfortunately, \pWMEE{} is \NPh , and thus it is likely not to be solvable within time polynomial in the input size. Thus, we analyze \pWMEE{} using the parameterized-complexity framework. That is, we consider running times not only depending on the input size, but also depending on other properties of the input---called parameters. A problem is called fixed-parameter tractable for a specific parameter if there is an algorithm whose exponential running-time portion depends only on this parameter. For instance, \citet{DMNW10} have proven that \pWMEE{} is fixed-parameter tractable with respect to the parameter~$k = \text{``number bridges that have to be built''}$. They have shown that \pWMEE{} can be solved with an algorithm that has running time~$\bigO(4^kn^4)$.

In this work, we consider smaller parameters for \pWMEE{}. In particular, we look at the number of ``islands'' the input city consists of and a more technical parameter which we introduce later. By islands we mean districts in the city that are completely cut off from other districts and cannot be reached via bridges. On the positive side, despite \pWMEE{} being \NPh{}, we are able to derive algorithms for this problem whose exponential running time portion depends only on these two, presumably small, parameters. On the negative side, we observe that preprocessing \pWMEE{} such that the size of the resulting instances is bounded by polynomials in these parameters likely is not possible.

Our work is organized as follows: In \autoref{sec:prelim}, we gather a common knowledge base by recapitulating basic notions of graph theory and parameterized algorithmics. \autoref{sec:probzoo} treats previous work on \pWMEE{} and the related problem of \pRP{}. There we also give an \NPhs{} proof for \pWMEE{} for illustrative purposes. In \autoref{sec:conncomp}, constituting the main part of this work, we consider structural properties of \pWMEE{} and derive an efficient algorithm. We also restate the problem in a matching context. This matching formulation might become a stepping stone to solve the question of whether \pWMEE{} is fixed-parameter tractable with respect to the parameter ``number of islands in the city'' we have sketched above. \autoref{sec:incompress} contains our considerations with regard to preprocessing routines for \pWMEE{}. There, it is shown that no polynomial-time data reduction rules exist that reduce an instance to a polynomial-size problem kernel unless~$\clacoNP \subseteq \claNP \slashpoly $. Finally in \autoref{sec:conclusion}, we give a brief summary of our results and give directions for further research.

\section{Preliminaries}
\label{sec:prelim}
We assume the reader to be familiar with the basics of set theory, logic, algorithm analysis, and computational complexity. They are covered, for example, by \citet{CLRS01} and \citet{AB09}.


\subsection{Parameterized Algorithmics}
Many problems of practical importance are \NPh{}. Thus, it is widely believed that for these problems there is no algorithm whose running time is bounded by a polynomial in the input size. However, in practice sometimes the phenomenon can be observed, that instances of such problems are indeed solvable within reasonable time. The reason for this is that some algorithms can be analyzed such that their super-polynomial running time portion depends only on some property of the input instances. If such a property, called ``parameter'', is not dependent on the input size and if it is ``small'' in practical instances, then it becomes feasible to solve even big instances of \NPh{} problems. The notion of parameters also gives rise to a method measuring the effectiveness of data reduction rules: If the rules reduce an instance such that its size is bounded by polynomials depending only on the parameter, we may assume that the reduction rules perform well in practice.

The design of algorithms exploiting small parameters is treated in \citet{Nie06}. However, it is likely that not every problem admits such an algorithm for every possible parameter. In this regard, complexity-theoretic approaches are presented in \citet{DF99} and \citet{FG06}. We use the parameterized complexity framework in our analysis of \pWMEE{} and give some basic definitions of parameterized algorithms and complexity here.

\paragraph{Problems and Parameterizations.}
Let~$\Sigma$ be an alphabet. A \emph{parameterization} is a polynomial-time computable function~$\kappa : \Sigma^* \rightarrow \mathbb{N}$. A \emph{parameterized problem} over~$\Sigma$ is a tuple~$(Q, \kappa)$ where~$Q \subseteq \Sigma^*$ and~$\kappa$ is a parameterization. For an instance~$I \in \Sigma^*$ of a parameterized problem~$(Q, \kappa)$ we also call~$\kappa(I)$ the \emph{parameter}. A parameterized problem~$(Q, \kappa)$ is called \emph{fixed-parameter tractable} with respect to~$\kappa$, if there is an algorithm that, given an instance~$I \in \Sigma^*$, decides whether~$I \in Q$ in time at most~$f(\kappa(I))\cdot p(\length(I))$. Here,~$f: \mathbb{N} \rightarrow \mathbb{N}$ is a computable function and~$p$ is a polynomial.

\paragraph{Search Tree Algorithms.}
A straightforward way to prove a problem fixed-parameter tractable is to employ search tree algorithms. A \emph{search tree algorithm} recursively divides an instance of a problem into a number of new instances. It divides the instances, until they are solvable within polynomial time. If both the recursion depth of the algorithm and the number of new instances generated from each instance is bounded by the parameter, then we obtain an algorithm whose superpolynomial time-portion is bounded by some function depending only on the parameter. Thus, the algorithm is a witness to the fixed-parameter tractability of the problem. Since search tree algorithms often terminate early, are easily parallelized and often allow for many performance tweaks, they often form relevant practical algorithms for \NPh{} problems.

\paragraph{Problem Kernels.}
Problem kernels are often used to give a guarantee of efficiency for data reduction rules. Let~$(Q,\kappa)$ be a parameterized problem. A \emph{reduction rule} for~$(Q, \kappa)$ is a mapping~$r : \Sigma^* \rightarrow \Sigma^*$ such that  
\begin{asparaenum}
\item $r$ is polynomial-time computable,
\item for every~$I \in \Sigma^*$ it holds that~$I \in Q$ if and only if~$r(I) \in Q$, and \label{enu:rr11}
\item $\kappa(r(I)) \leq \kappa(I)$.
\end{asparaenum}
Statement~\enuref{enu:rr11} is called the \emph{correctness} of the rule. A reduction to a \emph{problem kernel} for~$(Q, \kappa)$ is a reduction rule~$r : \Sigma^* \rightarrow \Sigma^*$ such that for every~$I \in \Sigma^*$ it holds that~$\length(r(I)) \leq h(\kappa(I))$, where~$h$ is a computable function. If~$h$ is a polynomial, then we call~$r$ a reduction to a \emph{polynomial-size problem kernel}.

\paragraph{Parameterized Reductions.}
Let~$(Q, \kappa)$, and~$(Q', \kappa')$ be two parameterized problems. A \emph{parameterized many-one reduction} from~$(Q, \kappa)$ to~$(Q', \kappa')$ is a mapping~$r : \Sigma^* \rightarrow \Sigma^*$ such that for every~$I \in \Sigma^*$
\begin{asparaenum}
\item $r(I)$ is computable in time at most~$f(\kappa(I))\cdot p(\length(I))$, where~$f : \mathbb{N} \rightarrow \mathbb{N}$ is a computable function and~$p$ a polynomial, \label{enu:parared12}
\item $r(I) \in Q'$ if and only if~$I \in Q$, and \label{enu:parared11}
\item there is a computable function~$g : \mathbb{N} \rightarrow \mathbb{N}$ such that~$\kappa'(r(I)) \leq g(\kappa(I))$.\label{enu:parared13}
\end{asparaenum}
Statement~\enuref{enu:parared11} is called the \emph{correctness} of the reduction. If the function~$f$ is a polynomial in statement~\enuref{enu:parared12} we call~$r$ a \emph{parameterized polynomial-time many-one reduction}. If the function~$g$ in statement~\enuref{enu:parared13} is a polynomial, we call~$r$ a \emph{polynomial-parameter many-one reduction}.
A \emph{parameterized Turing reduction} from~$(Q, \kappa)$ to~$(Q', \kappa')$ is an algorithm with an oracle to~$Q'$ that, given an instance~$I \in \Sigma^*$,
\begin{asparaenum}
\item decides whether~$I \in Q$,
\item has running time at most~$f(\kappa(I))\cdot p(\length(I))$, where~$f: \mathbb{N} \rightarrow \mathbb{N}$ is a computable function and~$p$ a polynomial, and
\item there is a computable function~$g :\mathbb{N} \rightarrow \mathbb{N}$ such that for every oracle query~$I' \in \Sigma^*$ posed by the algorithm it holds that~$\kappa(I') \leq g(\kappa(I))$.
\end{asparaenum}


\paragraph{Parameterized Complexity.}
It is assumed that not all parameterized problems are fixed-parameter tractable. To distinguish the various degrees of (in\nobreakdash-)tractability, there is a multitude of complexity classes. We briefly introduce the classes~\[\claFPT{} \subseteq \claW{1} \subseteq \claW{2} \subseteq \ldots \subseteq \claW{P} \subseteq \claXP{}\text{.}\]
For a brief introduction to parameterized intractability and the W-hierarchy see \citet{CM08}, for comprehensive works see \citet{DF99} or \citet{FG06}.

The class~\claFPT{} contains all parameterized problems that are fixed-parameter tractable. The class~\claW{t},~$t \in \mathbb{N}$ contains all problems that are parameterized many-one reducible to the satisfiability problem for circuits with depth at most~$t$ and an AND output gate, parameterized by the weight of the sought truth assignment---that is the number of variables that are assigned~$\true$. The class~\claW{P} contains all parameterized problems~$(Q, \kappa)$ such that it can be decided whether a word~$I \in \Sigma^*$ is contained in~$Q$ in at most~$f(\kappa(I))\cdot p(\length(I))$~time, using a Turing machine that makes at most~$h(\kappa(I))\cdot \log(\length(I))$ nondeterministic steps. Here,~$f, h : \mathbb{N} \rightarrow \mathbb{N}$ are computable functions and~$p$ is a polynomial. The class~\claXP{} contains all parameterized problems~$(Q, \kappa)$ for which there is a computable function~$f$ such that it can be decided whether~$I \in \Sigma^*$ is contained in~$Q$ in time at most~$\length(I)^{f(\kappa(I))} + f(\kappa(I))$.

A parameterized problem~$(Q, \kappa)$ is assumed to be fixed-parameter intractable, if it is hard for the class of problems~\claW{1}. That is, all parameterized problems in~\claW{1} are parameterized many-one reducible to~$(Q, \kappa)$. Hardness for~\claW{1} can be shown via a parameterized many-one reduction from a~\claW{1}-hard parameterized problem. Such a problem is, for instance, \textsc{Independent Set} parameterized by the size of the sought independent set~$k$:
\decprob{\textsc{Independent Set}}{A graph $G=(V,E)$ and an integer~$k$.}{Is there a vertex subset~$S \subseteq V$ such that~$k \leq |S|$ and~$G[S]$ contains no edges?}

\subsection{Graphs}
We now recapitulate some basic notions of graph theory. We oriented our definitions towards the ones given by \citet{BG08}. Other books on graph theory include \citet{Die05} and \citet{Wes01}. We also give some non-canonical definitions for notions that we frequently use; especially in the ``Connectivity'' and ``Degree and Balance'' paragraphs below.

A \emph{directed multigraph}~$G$ is a tuple~$(V, A)$, where~$V$ is a set,~$A$ is a multiset and for every~$a \in A: a = (u, v) \in V \times V \wedge u \neq v$.\footnote{We exclude ``self-loops''~$(v,v) \in A$ here, because they are not meaningful in the context of \EE s.} We sometimes denote~$V$ by~$V(G)$ and~$A$ by~$A(G)$. Elements in~$V$ are called the \emph{vertices} of~$G$ and elements in~$A$ are called the \emph{arcs} of~$G$. We denote~$|V|$ by~$n$ and~$|A|$ by~$m$ where it is appropriate. A vertex~$v \in V$ and an arc~$(u, w) \in A$ are called \emph{incident} if~$u = v$ or~$w = v$. Two vertices~$u,v \in V$ are called \emph{adjacent} or \emph{neighbors} if there is an arc~$a \in A$ such that~$u, v$ and~$a$ are incident. Let~$B$ be an arc set. We define the directed multigraph~$G + B$ as~$(V, A \cup B)$.

\paragraph{Subgraphs.}
Any directed multigraph~$(V', A')$ such that~$V' \subseteq V$ and~$A' \subseteq A$ is called a \emph{subgraph} of~$G$. Let~$V' \subseteq V$ be a vertex set.  The graph~$G[V'] := (V', B)$ where~$B = A \cap V' \times V'$ is called the \emph{vertex-induced subgraph} of~$G$ with respect to~$V'$. Let~$A' \subseteq A$ be an arc-set. The graph~$G\langle A' \rangle := (W, A')$ where~$W = \{v \in V: \exists a \in A': v\text{ and } a \text{ are incident}\}$ is called the \emph{arc-induced subgraph} of~$G$ with respect to~$A'$.

\paragraph{Walks, Trails, and Paths.} A \emph{walk} is an alternating sequence~\[v_1, a_1, v_2, \ldots, v_{k-1}, a_{k-1}, v_k\] of vertices~$v_i \in V$ and arcs $a_j \in A$ such that~$a_j = (v_j, v_{j+1})$ for all~$1 \leq j \leq k-1$. A \emph{subwalk} of a walk~$w$ is a consecutive subsequence of~$w$ beginning and ending with a vertex. We say that a walk \emph{traverses} a vertex~$v$ (an arc~$a$) if the vertex~$v$ (the arc~$a$) is contained in the corresponding sequence. We say that a walk~$w$ traverses a vertex set (arc set), if all vertices (arcs) in the set are traversed by~$w$. The \emph{length} of a walk is the number of arcs it traverses. The first vertex of a walk~$w$ is called the \emph{initial} vertex and the last vertex is called the \emph{terminal} vertex. The initial and terminal vertices are also called the \emph{endpoints} of~$w$. We say that a walk is \emph{closed} if its initial and terminal vertices are equal. A~\emph{trail} in the graph~$G$ is a walk that traverses every arc of~$G$ at most once. A~\emph{path} in the graph~$G$ is a trail that traverses every vertex of~$G$ at most once. A closed trail that traverses every vertex of~$G$ at most once except for its initial and terminal vertices is called a~\emph{cycle}. We sometimes abuse notation and identify walks with their corresponding arc sets or their arc-induced subgraphs.

\paragraph{Graphs and Orientations.}
A \emph{multigraph}~$G$ is a tuple~$(V, E)$, where~$V$ is a set,~$E$ is a multiset and for every~$e \in E: e \subseteq V \wedge |e| = 2$. We sometimes denote~$V$ by~$V(G)$ and~$E$ by~$E(G)$. Vertices, edges, incidence, (vertex- or edge-induced) subgraphs, walks, trails and paths are defined analogously to the definitions for multigraphs. \emph{Directed graphs} (digraphs for short) and \emph{graphs} are the special cases of multigraphs that comprise arc or edge sets instead of multisets, respectively. A (directed) graph is called \emph{complete} if it contains all possible edges (arcs). Let~$G = (V, E)$ be a (directed) graph and let~$G' = (V, E')$ be a complete (directed) graph. The \emph{complement graph} of~$G$ is the graph~$G = (V, \quer{E})$, where~$\quer{E} = E' \setminus E$.

A directed multigraph is said to be an \emph{orientation} of a multigraph~$G$ if it can be obtained from~$G$ by substituting every edge~$\{u, v\}$ by either the arc~$(u, v)$ or the arc~$(v, u)$. The \emph{underlying multigraph} of a directed multigraph~$G$ is the uniquely determined multigraph~$G'$ such that~$G$ is an orientation of~$G'$.

\paragraph{Drawing Graphs.}
We draw a directed multigraph by drawing circles for vertices, sometimes drawing their names inside the circle, and by drawing arrows with the head at the circle corresponding to the vertex~$u$ for arcs~$(v, u)$. Multigraphs are drawn by drawing circles for vertices and by drawing lines between the corresponding circles for edges.

\paragraph{Connectivity.}\label{def:meta}
A (directed) multigraph~$G$ is said to be \emph{connected} if for every pair of vertices~$u,v \in V(G)$ there is a path with the endpoints~$u,v$ in (the underlying multigraph of)~$G$. A maximal vertex set~$C \subseteq V(G)$ such that~$G[C]$ is connected is called a \emph{connected component} of~$G$. We sometimes abuse notation and identify connected components~$C$ with their vertex-induced subgraphs~$G[C]$. When it is clear from the context we denote connected components simply by \emph{components}. By the \emph{component graph}~$\comp{G}$ of~$G$ we denote the complete graph that has a vertex for every connected component of~$G$. Consider a trail~$t$ that traverses only vertices of~$G$ and the trail~$s$ in~$\comp{G}$ that is obtained from~$t$ as follows: for every connected component~$C$ of~$G$, substitute every maximum length subtrail~$t'$ of~$t$ such that~$V(t') \subseteq C$ by the vertex in~$\comp{G}$ that corresponds to~$C$. 
We denote the underlying trail of~$s$ by~$\meta{G}{t}$. For an example on connected components and the~$\meta{G}$ mapping, see \autoref{fig:metagraph-example}.
\begin{figure}%
  \begin{center}%
    \includegraphics{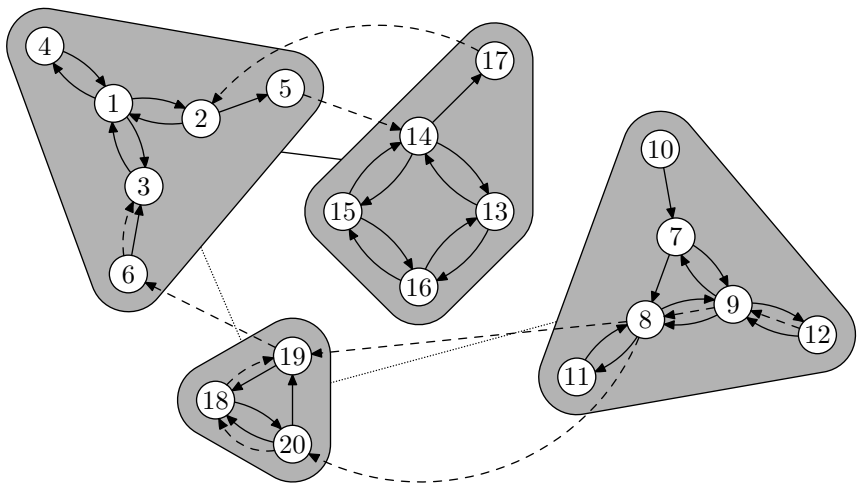}
    \caption{%
      A directed graph~$G$ (vertices~$1$ through~$20$, solid arcs), and its components~(encircled and shaded in gray). Furthermore, a number of trails is shown that traverse vertices of~$G$ (dashed arcs). The mapping~$\meta{G}$ maps both the trails traversing the vertices~$5,14$ and~$17,2$, respectively, to trail in~$\comp{G}$ that is represented by the solid line. The trails traversing the vertices~$12, 9, 8, 19, 6$ and~$8, 20, 18, 19, 6, 3$, respectively, both map to the trail represented by the dotted lines.
    }%
    \label{fig:metagraph-example}
  \end{center}
\end{figure}%
\paragraph{Degree and Balance.}
Let~$G = (V, A)$ be a directed multigraph. The \emph{indegree} (\emph{outdegree}) of a vertex~$v \in V$ denoted by~$\indeg(v)$~($\outdeg(v)$)  is~$|\{(u, v) \in A\}|$~($|\{(v, u) \in A\}|$). The \emph{balance} of~$v$, denoted by~$\balance(v)$, is~$\indeg(v) - \outdeg(v)$. In directed multigraphs, a vertex~$v$ is called \emph{balanced} if~$\balance(v) = 0$. 

Let~$G = (V, E)$ be a multigraph. The \emph{degree} of a vertex~$v \in V$ denoted by~$\deg(v)$ is~$|\{\{u, v\} \in E\}|$. In multigraphs, a vertex~$v$ is called \emph{balanced} if~$\deg(v)$ is even.

Let~$G = (V, E)$ be a (directed) multigraph. We denote the set of unbalanced vertices by~$I_G$. If~$G$ is directed, we denote~$\{v \in V: \balance(v) > 0\}$ by~$I_G^+$ and~$\{v \in V: \balance(v) < 0\}$ by~$I_G^-$.

\paragraph{Eulerian Graphs and Extensions.}  A closed trial in a (directed) multigraph~$G$ is said to be \emph{Eulerian}, if it traverses every edge in~$E(G)$ (arc in~$A(G)$) exactly once and every vertex in~$V(G)$ at least once.\footnote{Note that there seem to be two equally well-accepted definitions of \Eu{}trails: The definitions with and without the additional vertex condition. We chose the one with the vertex condition here, because it makes it easier to deal with connected components that consist only of one vertex. Algorithmically, problems according to both formulations are easily inter-transformable.} A (directed) multigraph is called Eulerian, if it contains an \Eu trail. The following theorem holds.
\begin{theorem}\label{the:eulerian}
  A (directed) multigraph is \Eu if and only if it is connected and every vertex is balanced.
\end{theorem}
A version of \autoref{the:eulerian} that is restricted to graphs is due to Euler
, a proof for the generalized version above can be found in~\citet{BG08}.
We call an edge multiset (arc multiset)~$E$ such that~$G + E$ is \Eu an \emph{\EE{}} for~$G$. Edges (arcs) contained in~$E$ are called \emph{extension edges}~(\emph{extension arcs}).

\paragraph{Vertex Partitions and Bipartite Graphs.}
Let~$G= (V, A)$ be a (directed) multigraph. A family of sets~$P = \{C_1, \ldots, C_k\}$ is called a \emph{vertex partition} of~$G$, if~$V = \bigcup_{i = 1}^k C_i$ and~$C_i \cap C_j = \emptyset$ for all~$1 \leq i < j \leq k$. The sets~$C_i$,~$1 \leq i \leq k$ are called \emph{cells} of the partition~$P$.

A graph~$G = (V_1 \uplus V_2, E)$ is called a \emph{bipartite graph}, if~$\{V_1, V_2\}$ is a vertex partition of~$G$ and for every~$e = \{u, v\} \in E: u \in V_1 \wedge v \in V_2$.\footnote{In this regard we use the symbol~$\uplus$ to indicate a disjoint union of the vertex sets.}

\paragraph{Matchings.}
Let~$G = (V, E)$ be a graph. A set~$M \subseteq E$ is called a \emph{matching} in~$G$ or of the vertices in~$G$, if for every~$e, f \in E: e \cap f = \emptyset$. A matching~$M$ is called \emph{perfect} if for every vertex~$v \in V$ there is an edge in~$M$ that is incident to~$v$. The following theorem holds:
\begin{theorem}[Hall's condition]\label{the:hallscondition}
  A bipartite graph~$G = (V_1 \uplus V_2, M)$ has a perfect matching, if and only if~$U \leq N(U)$ for every~$U \subseteq V_1$. Here,~$N(U)$ denotes the set of all neighbors of~$U$.
\end{theorem}
A proof for \autoref{the:hallscondition} can be found in \citet{BG08}.

\section{Problems, Variants, Relationships}\label{sec:probzoo}
\EG s are interesting by themselves from a graph-theoretic point of view. However, they also bear intuitive and practical applications. In this section we introduce various problems regarding \EG s, their complexity if it is known, and point out relations to other problems.

As we will see later in this section, some natural problems translate into the problem of making a given graph \Eu by adding edges or arcs, respectively. In these problems it is beneficial to add as few edges as possible, or to add edges such that their total weight is as low as possible. This translates into the following problem formulation:
\decprob{\pWMEE{} (\pWMEEs{})}{A directed multigraph~$G=(V,E)$ and a weight function~$\wf: V \times V \rightarrow [0, \wf_{max}]\cup \{\infty\} $.}{Is there an \EE{} for~$G$ of weight at most~$\wf_{max}$?}
The problem of \textsc{Unweighted Eulerian Extension} is \pWMEEs{} where every arc in~$V \times V$ has weight~$1$. Natural variants of these problems can be derived by substituting undirected multigraphs, directed graphs or graphs for multigraphs in the problem description. As we will also see, the complexity of problems regarding weighted \EE s depends heavily on the connectedness of the input graph. So, connectedness makes for another intuitive distinction in these problems.

\paragraph{Polynomial-time Solvable Variants.}

\begin{table}
  \centering
  \begin{tabularx}{\textwidth}{X XlX XcX XcX X} 
    \toprule
    & \multicolumn{9}{c}{\EE{} on unweighted graphs} & \\
    & &&& \multicolumn{3}{c}{Connected}
    & \multicolumn{3}{c}{Disconnected} & \\
    \midrule
    && Undirected &&& $\quer{m}\sqrt{n}$ &&& $\quer{m}\sqrt{n}$ & \\
    && Directed &&& $n\quer{m}\log(n)$ &&& $\quer{m}\log(n)(\quer{m}+n\log(n))$ & \\ 
    \bottomrule
  \end{tabularx}
  \caption{Complexity results regarding unweighted \EE{} problems. The number of edges in complement graphs of graphs with~$m$ edges is denoted by~$\quer{m}$. Running times in big-O notation. The result for undirected and directed graphs have been obtained by \citet{BST77} and \citet{DMNW10}, respectively.}
  \label{tab:complexityunweightedgraphs}
\end{table}

\autoref{tab:complexityunweightedgraphs} shows polynomial running time results for unweighted \EE{} problems on graphs. For unweighted multigraphs, \citet{DMNW10} obtained linear-time algorithms for both the directed and undirected case. These algorithms work regardless of whether the input multigraph is connected or disconnected. Polynomial-time solvability has also been proven for the unweighted and connected variants shown in \autoref{tab:complexityweighted}. 

\begin{table}
  \centering
  \begin{tabularx}{\textwidth}{X XlX XcX XcX X} 
    \toprule
     &\multicolumn{9}{c}{Weighted, connected \EE{}} &\\
     &&&& \multicolumn{3}{c}{Graphs} &
     \multicolumn{3}{c}{Multigraphs} &\\
    \midrule
    && Undirected &&& $|I_G|^3\log(|I_G|)$ &&& $|I_G|^3\log(|I_G|)$ && \\
    && Directed &&& $\quer{m}\log(n)(\quer{m} + n \log(n))$ &&& $n^3 \log(n)$ && \\ 
    \bottomrule
  \end{tabularx}
  \caption{Complexity results regarding weighted \EE{} problems on connected graphs. The number of edges in complement graphs of graphs with~$m$ edges is denoted by~$\quer{m}$. Recall that~$I_G$ denotes the set of not balanced vertices in the input graph. Running times in big-O notation. These results have been obtained by \citet{DMNW10}.}
  \label{tab:complexityweighted}
\end{table}

\paragraph{Fixed-Parameter Tractability.}
In general, \pWMEEs{} is \NPh{}. We recapitulate two \NPhs{} proofs in the following subsections. However, \citet{DMNW10} have proven \pWMEEs{} to be fixed-parameter tractable with respect to a slightly complicated parameterization: Let~$\mathbb{E}(G, \wf)$ be the set of all \EE s~$E$ for the directed multigraph~$G$ with weight~$\wf(E) \leq \wf_{max}$ according to the weight function~$\wf$.\label{def:paraeulerpapier}
\begin{theorem}
  \pWMEE{} parameterized by~$k = \max \{ |E| : E \in \mathbb{E}(G, \wf)\}$ is solvable in~$\bigO(4^kn^4)$~time.
\end{theorem}
Note that the according parameterization is likely not polynomial-time computable. This calls for the trick to encode the parameter in the corresponding language~$Q$ of the parameterized problem. The parameter then has to be checked for correctness by any algorithm that decides~$Q$.
\subsection{Relations to the Rural Postman Problem}\label{sec:relrpwmee}
In this section, we briefly review the many-one reductions from \pWMEE{}~(\pWMEEs{}) to the \pRP{} problem and back, given by \citet{DMNW10}. From these reductions we get parameterized equivalence with respect to parameters that motivate our choice of parameters for~\pWMEEs{}.
The \pRP{} problem is defined as follows.
\decprob{\pRP{} (\pRPs{})}{A directed graph~$G = (V, A)$, a set~$R \subseteq A$ of required arcs and a weight function~$\wf : A \rightarrow [0, \wf_{max}]\cup \{\infty\}$.}{Is there a walk~$W$ in~$G$ such that~$W$ traverses all arcs in~$R$ and~$\wf(W) \leq \wf_{max}$?}
\citet{DMNW10} observed that \pRPs{} parameterized by the ``number of arcs in the sought walk'' and \pWMEEs{} parameterized by ``number of arcs in the sought \EE'' are equivalent.\footnote{The actual parameters are slightly more complicated, but this intuition suffices here.}
We take a brief look at their construction here and observe a further parameterized equivalence. The main idea in both reductions is to exploit the following observation.
\begin{observation}\label{obs:walkvseuleriangraph}
  Let~$G$ be a directed graph and let~$W$ be a multiset of arcs in~$G$. There is a closed walk in~$G$ that uses exactly the arcs~$W$ if and only if the directed multigraph~$(V(G), W)$ is Eulerian.
\end{observation}
With \autoref{obs:walkvseuleriangraph} it is easy to see that the following two constructions are polynomial-time many-one reductions from \pRPs{} to \pWMEEs{} and from \pWMEEs{} to \pRPs{}, respectively.
\begin{construction}\label{cons:redrptowmee}
  Let the directed graph~$G = (V, A)$, the required arc set~$R$ and the weight function~$\wf : V \times V \rightarrow [0, \wf_{max}]\cup \{\infty\}$ constitute an instance of \pRPs{}. Construct an instance of \pWMEEs{} by defining the directed multigraph~$G' := (V, R)$ and a weight function~$\wf': V \times V \rightarrow [0, \wf_{max} - \wf(R)]\cup \{\infty\}$ by~
\[\wf' :=
\begin{cases}
  \wf(a), & a \in A \wedge \wf(a) \leq \wf_{max} - \wf(R)\text{,} \\
  \infty, & \text{otherwise.}
\end{cases}
\]
\end{construction}
\begin{construction}\label{cons:redwmeetorp}
  Let the directed multigraph~$G = (V, A)$ and the weight function~$\wf : V \times V \rightarrow [0, \wf_{max}]\cup \{\infty\}$ constitute an instance of \pWMEEs{}. Construct an instance of \pRPs{} by defining the directed graph~$G := (V, V \times V)$, the required arc set~$R := A$, the weight function~$\wf' = \wf$ and the maximum weight~$\wf'_{max} := \wf_{max} + \wf(R)$.
\end{construction}
In the search for suitable parameters for \pWMEEs{}, we observed the following. Intuitively, we expect the number of connected components in~$G\langle R \rangle$ to be small in practical instances. For instance, consider a postman's tour in a city that comprises a number of suburbs. The number of streets that have to be serviced in each of the suburbs is expected to be much higher than the streets in-between, thus forming connected components in each suburb. We also expect the sum of positive balances of all vertices in~$G\langle R \rangle$ to be small: This sum is at most proportional to the number of required arcs, and we assume this number to be small compared to~$n$ in practice. 
With regard to \pWMEEs{}, the following observation is of much interest.
\begin{observation}\label{obs:rpwmeeparaequiv}
  Let~$G$ be the input digraph and~$R$ the required arcs in an instance of \pRPs{}. Let~$G'$ be the input graph in an instance of \pWMEEs{}. \autoref{cons:redrptowmee} and \autoref{cons:redwmeetorp} are polynomial-time polynomial-parameter many-one reductions with respect to the parameters
\begin{lemenum}
\item number of connected components in~$G\langle R \rangle$ and number of connected components in~$G'$, and/or
\item sum of all positive balances in~$G\langle R \rangle$ and sum of all positive balances in~$G'$.
\end{lemenum}
\end{observation}
This motivates the analysis of \pWMEEs{} with respect to these two parameters. In this regard, \citet{Fre77} has proven the following theorem.
\begin{theorem}
  \pRP{} can be solved in~$\bigO(n^3n^{2c - 2}/c!)$ time, where~$c$ is the number of connected components in~$G\langle R \rangle$---the graph~$G$ being the input graph and~$R$ the set of required arcs.
\end{theorem}
From this theorem it immediately follows that \pRPs{} parameterized by the number of components in~$G\langle R \rangle$ is in~\claXP{} and thus, by \autoref{obs:rpwmeeparaequiv}, \pWMEEs{} parameterized by the number of components in the input graph also is in~\claXP .

\subsection{Relations to the Hamiltonian Cycle Problem}\label{sec:relhcwmee}
In this section we observe that the difficulty of solving \pWMEE{}~(\pWMEEs{}) depends on the number of components in the input graph. This is done using a reduction from the \pHC{} problem. A natural question is, whether the difficulty of solving \pWMEEs{} depends \emph{only} on the number~$c$ of components, that is, whether \pWMEEs{} is fixed-parameter tractable with respect to the parameter~$c$. We attack this question in \autoref{sec:conncomp}, especially \autoref{sec:matching}.

This section also shows that the parameter ``sum of all positive balances of vertices in the input graph'' for \pWMEEs{} will likely not yield fixed-parameter tractability.

The \pHC{} problem is defined as follows.
\begin{definition}
  Let~$G$ be a directed graph. A cycle in~$G$ is called \emph{Hamiltonian} if it traverses every vertex in~$G$ exactly once.
\end{definition}
\decprob{\pHC{} (\pHCs{})}{A directed graph~$G$.}{Is there a Hamiltonian cycle in~$G$?}
\citet{Orl76} notes that the complexity of \pRPs{} seems to depend on the number of connected components in~$G\langle R \rangle$, where~$G$ is the input graph and~$R$ is the set of required arcs. In a way, \citet{LR76} proved this statement by giving a reduction from the \NPh{}~\cite{Kar72} \pHCs{} problem such that the number of components in~$G\langle R \rangle$ in the \pRPs{} instance is exactly the number of vertices in the \pHCs{} instance. In this section, we give a reduction from \pHCs{} to \pWMEEs{} illustrating that the same is true for \pWMEEs{}.

The main idea of the reduction is that any \EE{} for \pWMEEs{} has to connect all connected components in the input graph. Thus, we model every vertex by a connected component consisting of two vertices that are connected by two arcs: One arc in either direction. To model edges in the instance of \pRPs{}, we utilize the weight function and choose~$\wf_{max}$ accordingly to ensure that every feasible \EE{} is a cycle.
\begin{construction}\label{cons:redhctowmee}
  Let the directed graph~$G' = (V', A')$ constitute an instance of \pRPs{}. Construct an instance of \pWMEEs{} as follows: 

Define the directed multigraph~$G = (V, A)$ by~$V := V' \times \{0, 1\}$ and~\[A := \{((v, 1), (v, 0)), ((v, 0), (v, 1)) : v \in V'\}\text{.}\] Set the maximum weight~$\wf_{max} := |V'|$ and define the weight function~$\wf$ by~\[
\wf(a) := 
\begin{cases}
  1,& a = ((u, 0), (v, 0)) \wedge (u, v) \in A', \\
  \infty,& \text{otherwise.}
\end{cases}
\]
\end{construction}
It is easy to see that this construction is correct using the following observation:
\begin{observation}
  Any \EE{}~$E$ for~$G$ with~$\wf(E) \leq \wf_{max}$ is a cycle.
\end{observation}
\begin{proof}
  Since~$E$ has to connect~$|V'|$ connected components in~$G$, it contains at least~$|V'|-1$ arcs. The \EE{}~$E$ cannot contain a maximum-length trail that is open, since there are no unbalanced vertices in~$G$. For sake of contradiction assume that~$E$ contains three arcs that are incident to one vertex~$v$ in~$G$. Then, to connect the remaining connected components in~$G$ via a closed trail,~$E$ has to contain at least~$|V'| - 3$ arcs. However, then~$v$ is still not balanced and~$E$ has to contain at least one additional arc, totalling in~$|V'| + 1$ arcs. Thus, by contradiction, every vertex in~$G$ has at most two incident arcs in~$E$ and thus~$E$ is a cycle.
\end{proof}
Thus, \autoref{cons:redhctowmee} is correct and we have that the difficulty in solving \pWMEEs{} depends on the number of components in the input graph. But the reduction given by \autoref{cons:redhctowmee} also gives the following observation.
\begin{observation}
  \pWMEEs{} parameterized by the sum of all positive balances of vertices in the input graph is not contained in \claXP, unless~$\claP = \claNP$.
\end{observation}
\begin{proof}
  Observe that all vertices in the graph~$G$ produced by \autoref{cons:redhctowmee} are balanced. If \pWMEEs{} parameterized by the sum~$b$ of all positive balances of vertices in the input graph was in \claXP , in particular all instances with~$b=0$ were solvable within polynomial time. Thus, \pHCs{} would be solvable within polynomial time.
\end{proof}%

\subsection{Constrained Eulerian Extensions}\label{sec:constrainedEE}
A natural modification of \EE{} problems is to give constraints on the set of edges or arcs that can be added to the input graph in order to make it Eulerian. 
Note for example that in \pEE{} on graphs we can regard the condition that the input graph has to remain a graph with the added edges as a constraint on the allowed edges (that is multiedges are forbidden). Thusly constrained problems might also be interesting in practice. For instance, \citet{HJM09} observed that the following class of constrained \EE{} problems has applications to sequencing problems: 
\decprob{\pdDEE}{A directed graph~$G=(V,A)$, where $V \subset \mathbb{Q}^d$.}{Is there an \EE{}~$E$ for~$G$ such that for every~$(u, v) \in E$ it holds that~$u \geq v$ component-wise?}
However, \citet{HJM09} also have proven that \pdDEE{} is \NPc . We model such constraints on the extension edges in such problems as instances of \pWMEE{} by simply defining the weight function accordingly---assigning forbidden arcs or edges the weight~$\infty$, and setting the maximum weight to a large enough value.

We use \pdDEE{} as a helper problem in \autoref{sec:incompress}. In order to deal conveniently with the arc constraints, we introduce some notation at this point.
\begin{definition}
  Let~$\wf$ be a weight function assigning weights~$[0 , \wf_{max}] \cup \{\infty\}$ to arcs. An arc~$a$ is called \emph{allowed} with respect to~$\wf$ if~$\wf(a) < \infty$. If the weight function is clear from the context, then we simply say that the arc is allowed.
\end{definition}







\section{Our Work}
In recent research by \citet{DMNW10} the problem \pWMEE{}~(\pWMEEs{}) has been shown to be fixed-parameter tractable with respect to the parameter~$k=\text{``number of arcs in the sought \EE{}''}$.\footnote{The actual parameter is slightly more complicated---see page~\pageref{def:paraeulerpapier}---but the intuition of the number of needed extension arcs suffices here.} In this work we initiate a more fine-grained analysis of the \pWMEEs{} problem by considering parameters that are upper bounded by~$k$. In particular, we study the parameterizations ``number~$c$ of components in the input graph'' and ``sum~$b$ of all positive balances of vertices in the input graph''. Since any \EE{}~$E$ for a multigraph has to produce a connected graph, it holds that~$|E| \geq c - 1$ and thus~$k \geq c - 1$ . Also, any \EE{}~$E$ has to balance all vertices in the given multigraph, that is, for every vertex~$v$ with balance~$d > 0$, it has to contain at least~$d$ outgoing arcs. Hence it holds that~$|E| \geq b$ and thus~$k \geq b$. \autoref{tab:paracomplexitywmee} gives a compact overview over the new and known results regarding \pWMEEs{}.
\begin{table}
  \centering
  \begin{tabularx}{\textwidth}{X XlX XlX XlX X} 
    \toprule
     &\multicolumn{9}{c}{Parameterized complexity results for \pWMEE{}} &\\
     & \multicolumn{3}{c}{Parameter} & 
       \multicolumn{3}{c}{Known} &
       \multicolumn{3}{c}{New} &\\
    \midrule
    && $k$ &&& $\in \claFPT{} : 4^k$ &&& no polykernel && \\
    && $c$ &&& $\in \claXP$ &&& $\in \claW{P}$, no polykernel && \\ 
    && $b, c$ &&& --- &&& $\in \claFPT{} : 4^{c \log(bc^2)}$, no polykernel && \\ 
    \bottomrule
  \end{tabularx}
  \caption{Overview on parameterized complexity results for \pWMEEs{} regarding various parameters. Fixed-parameter tractability results include the superpolynomial term of the corresponding algorithm. Known results: The fixed-parameter tractability result for parameter~$k$ is due \citet{DMNW10}. The \claXP-result for parameter~$c$ is due \citet{Fre77} (see \autoref{sec:relrpwmee}). New results: The fixed-parameter tractability result for the combined parameter~$b, c$ is shown in \autoref{the:kcalgwmeea} and \autoref{cor:wmeefptbc}. The \claW{P}-result for parameter~$c$ follows from \autoref{obs:cbminnpwp} and \autoref{the:cbmwmeeequiv}. The non-existence of polynomial-size problem kernels is shown in \autoref{the:nonpoly} and its corollaries.}
  \label{tab:paracomplexitywmee}
\end{table}

\pWMEEs{} parameterized only with~$b$ is already \NPh{} when~$b = 0$: Consider the reduction we give in \autoref{sec:relhcwmee} to prove \NPhs{} for \pWMEEs{}. This reduction produces instances with~$b = 0$. Also, the question whether \pWMEEs{} is fixed-parameter tractable when parameterized by~$c$ is a long-standing open question which arose implicitly in research by \citet{Fre77, Fre79}. His work implies that \pWMEEs{} is polynomial-time solvable for every constant value of~$c$ (see \autoref{sec:relrpwmee}). However, his algorithm does not imply fixed-parameter tractability and this question seems to be hard to answer. Nevertheless, in \autoref{sec:conncomp} we show that when parameterizing with both~$b$ and~$c$ \pWMEEs{} becomes fixed-parameter tractable.

Pursuing the question whether \pWMEEs{} is fixed-parameter tractable with only the parameter~$c$, we restate \pWMEEs{} in the context of matchings in \autoref{sec:conncomp} and show that the problem \pCBM{} is parameterized equivalent to \pWMEEs{}. Using the matching formulation we obtain a fixed-parameter tractability result for a restricted class of \pWMEEs{} when parameterized by~$c$.

We also consider preprocessing routines for \pWMEEs{} in \autoref{sec:incompress}. In this regard, we show that \pdDEE{} does not admit a polynomial problem kernel with respect to the parameter~$k$. The result also transfers to the parameters~$b, c$ and the more general problem \pWMEEs{}.

\chapter{Connected Components}\label{sec:conncomp}
The main results given in this chapter are an efficient algorithm for \pWMEE{}~(\pWMEEs{}) with running time in~$\bigO(4^{c\log(bc^2)}n^2(b^2 + n\log(n)) + n^2m)$ and the parameterized equivalence of \pWMEEs{} parameterized by~$c$ and \pCBM{}. Here,~$c$ is the number of components and~$b$ is the sum of all positive balances in the input graph, that is for the input graph~$G$, it is~$b = \sum_{v \in I_G^+}\balance(v)\text{.}$ The equivalence to the matching problem also yields an algorithm for a restricted form of \pWMEEs{} with~$\bigO(2^{c(c + \log(2c^4))}(n^4 + m))$~running time. The latter result represents some partial progress to answer the question of whether \pWMEEs{} is fixed-parameter tractable with respect to the parameter~$c$.

We first make some observations about \EE s in \autoref{sec:trails} which expose that every \EE{} corresponds to a specific structure that has an intimate relationship to the connected components of the input graph. This then leads to a modified problem derived from \pWMEEs{} in \autoref{sec:advice}. There we consider the problems \pWMEEA{} (\pWMEEAs{}) and \pWMEECA{} (\pWMEECAs) where the structure of the sought \EE s is made explicit in the input. These restricted problems seem to be easier to tackle and we derive an algorithm with~$\bigO(4^{c\log(b)}n^2(b^2+ n\log(n)) + n^2m)$~running time for \pWMEECAs{}. Using observations about the relationship between \pWMEEs{} and \pWMEECAs{} we derive an algorithm for \pWMEEs{} running in~$\bigO(4^{c\log(bc^2)}n^2(b^2 + n\log(n)) + n^2m)$~time.

In \autoref{sec:matching} we introduce \pCBM{}~(\pCBMs{}) and show that it is tractable on some restricted graph classes. We give parameterized reductions from \pWMEEs{} to \pCBMs{} and from \pCBMs{} to \pWMEEs{} using some intermediary problems that we introduce in \autoref{sec:advice}. This then yields the parameterized equivalence of \pCBMs{} and \pWMEEs{}. As simple corollaries, we derive fixed-parameter tractability of \pWMEEs{} with respect to parameter~$c$ on some restricted input instances. The reductions also yield some results for intermediary problems, for example a problem kernel for \pWMEECAs{} that has size polynomial in~$b$ and~$c$.

Consult \autoref{fig:reduction-schematic} and \autoref{tab:tractresults} for an overview on the reductions given in this chapter and the tractability results obtained.
\begin{figure}
  \begin{center}
    \includegraphics{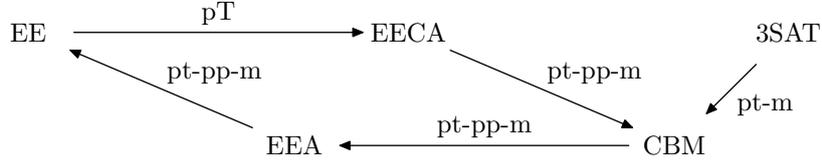}
    \caption{Schematic overview on the reductions given in this chapter. The label~``pT'' indicates a parameterized Turing reduction, the label~``pt-pp-m'' indicates a polynomial time polynomial parameter many-one reduction, and the label~``pt-m'' indicates a classical polynomial time many-one reduction. The reductions from and to \pWMEEs{} are covered in \autoref{sec:advice}. The reductions from and to \pCBMs{} are given in \autoref{sec:matching}.}
    \label{fig:reduction-schematic}
  \end{center}
\end{figure}%
\begin{table}
  \begin{minipage}{\linewidth}
  \centering
  \begin{tabularx}{\textwidth}{X l X l X l X} 
    \toprule
    & \multicolumn{5}{c}{Tractability results} & \\
    & \multicolumn{1}{c}{Problem} && \multicolumn{1}{c}{Result} && \multicolumn{1}{c}{Proposition} & \\
    \midrule
    & \pCBMs{}\footnote{When the input graph is a forest.}  && $n+m$ && \autoref{cor:cbmlintimeonforests} \\
    & \pCBMs{}\footnote{When the bipartite input graph has maximum degree two in one of its cells.} && $2^{j(j+1)}n + n^3$ && \autoref{lem:cbmmaxdeg2tractable}  \\
    & \pWMEECAs{} && $4^{c\log(b)}n^2(b^2+ n\log(n)) + n^2m$ && \autoref{the:kcalgwmeea} \\
    & \pWMEEAs{} && $b^2c$ vertex kernel && \autoref{cor:eecaprobkernel} \\
    & \pWMEEs{}\footnote{When the allowed arcs ``resemble'' a forest.} && $16^{c\log(c)}(cn^4 + m)$ && \autoref{cor:wmeetractableforests}  \\
    & \pWMEEs{}\footnote{When the allowed arcs ``resemble'' a vertex-disjoint union of cycles.} && $2^{c(c + \log(2c^4))}(n^4 + m)$ && \autoref{cor:wmeetractablecycles}  \\
    & \pWMEEs{} && $4^{c\log(bc^2)}n^2(b^2 + n\log(n)) + n^2m$ && \autoref{cor:wmeefptbc}  \\
    \bottomrule
  \end{tabularx}
  \caption{Overview on tractability results given in this chapter. All values in big-O notation. Here,~$j$ denotes the parameter ``join set size'' in \pCBMs{} instances. This parameter corresponds to the parameter ``number of components'' in \pWMEEs{} instances in reductions we give in this chapter.}
  \label{tab:tractresults}
\end{minipage}
\end{table}

\section{Structure of Eulerian Extensions}
\label{sec:trails}

In this section, we show that we can assemble a minimum-weight \EE{} for a graph~$G$ using trails that are of restricted structure, and bound the length and number of ``long'' trails by polynomials in the number of components in~$G$. To this end, we consider trails in \EE s. 
\begin{figure}
  \begin{center}
    \includegraphics{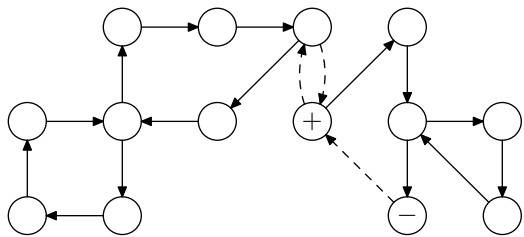}
    \caption{Examples of a closed maximum length trail (left, solid arcs) and an open maximum length trail (right, solid arcs) in an \EE{} (solid arcs). Arcs belonging to the input graph are dashed. Observe that the vertex~$+$ is the only vertex in~$I_G^+$ and the vertex~$-$ is the only vertex in~$I_G^-$.}
    \label{fig:pmetrails}
  \end{center}
\end{figure}

We investigate preprocessing routines for instances of \pWMEEs{}---namely, we split vertices (\autoref{trans:sb}) and use shortest-path preprocessing (\autoref{trans:spp})---that allow us to modify any valid \EE s such that we can make assumptions about their trails without increasing the weight of the extensions. In this section, we frequently use the component graph~$\comp{G}$ of a graph~$G$ and the mapping~$\meta{G}{t}$ of trails~$t$ in~$G$ to trails in~$\comp{G}$. These are defined on page~\pageref{def:meta} in \autoref{sec:prelim}. The main result of this section is as follows.
\newcommand{\eestructuretheorem}{%
  Let~$G$ be a directed multigraph with~$c$ connected components. Let~$G$ and the weight function~$\wf: V \times V \rightarrow [0, \wf_{max}]\cup \{\infty\}$ constitute an instance of \pWMEE{} that is preprocessed using \autoref{trans:sb} and \autoref{trans:spp}. Then, there is a set~$S := \{t_1, \ldots, t_k\}$ of pairwise edge-disjoint paths and cycles each in the graph~$(V, V \times V)$ such that
  \begin{lemenum}
  \item $\bigcup_{i=1}^kA(t_i)$ is an \EE{} of minimum weight for~$G$,\label{enu:ees5}
  \item each~$t_i \in S$ contains at most~$c+1$ vertices,\label{enu:ees2}
  \item in~$S$ there are at most~$c(c-1)/4$ paths and cycles containing more than one arc,\label{enu:ees3}
  \item the number of paths in~$S$ is at most~$|I_G^+| = |I_G^-|$,\label{enu:ees4}
  \item for~$t_i \neq t_j \in S$ of length at least two~$\meta{G}{t_i}$, and~$\meta{G}{t_j}$ are edge-disjoint,\label{enu:ees7}
  \item the graph defined by the union of all trails~$\meta{G}{t_1}, \ldots, \meta{G}{t_n}$ without their initial vertices does not contain a cycle. \label{enu:ees6}
  \end{lemenum}
}
\begin{theorem}\label{the:eestructure}
  \eestructuretheorem
\end{theorem}
In this section, let $G=(V,A)$ be a directed multigraph, let $E$ be an \EE{} for $G$---that is $G + E :=(V, A \cup E)$ is Eulerian---and let the function~$\wf: V \times V \rightarrow [0, \wf_{max}]\cup \{\infty\} $ be a weight function.
\begin{observation}
  \label{obs:walkdichotomy}
  A maximum-length trail in an \EE{} for a graph~$G$ either is closed or starts in~$I^+_G$ and ends in~$I^-_G$.
\end{observation}
\begin{proof}
  Consider the initial vertex~$v_A$ and terminal vertex~$v_\Omega$ of a trail~$t$ in the \EE~$E$. The vertices~$v_A$ and~$v_\Omega$ are balanced in~$G + E$. 

  Assume that~$v_\Omega$ is not balanced in~$G$. Every time~$t$ traverses~$v_\Omega$, it uses one arc in~$E$ that enters~$v_\Omega$ and one that leaves it. This implies that~$v_\Omega \neq v_A$ because~$v_\Omega$ is balanced in~$G+ E$ and thus there is an odd number of arcs in~$E$ incident to~$v_\Omega$ (recall that~$t$ is of maximum length). Since~$t$ ends in~$v_\Omega$, this also implies that~$v_\Omega \in I^-_G$. Analogously we get that~$v_A \in I^+_G$.
  
  Now assume that~$v_\Omega$ is balanced in~$G$. Since~$t$ cannot be extended, it already uses every arc incident to~$v_A$ and~$v_\Omega$. However, if~$v_\Omega$ is not equal to~$v_A$, there are more arcs entering~$v_\Omega$ than leaving~$v_\Omega$ in~$E$. This means that~$v_\Omega$ is not balanced in~$G + E$ which is a contradiction.
\end{proof}
\autoref{fig:pmetrails} illustrates \autoref{obs:walkdichotomy}.

\paragraph{Preprocessing Routines.}
There is a preprocessing routine introduced by~\citet{DMNW10} that ensures that every vertex has balance between~$-1$ and~$1$. This later helps to give a bound on very short trails in \EE s.
\begin{transformation}[Splitting Vertices]
  \label{trans:sb}
  Let the graph $(G=(V,A)$, the weight function~$\wf$ and the maximum weight~$\wf_{max})$ constitute an instance of \pWMEEs{}. Compute a new instance as follows: Search for a vertex~$v$ with~$|\balance(v)| > 1$, introduce a new vertex~$u$. If~$\balance(v) > 0$, choose an arbitrary arc~$(w, v)$, delete it and add the arc~$(w, u)$. Proceed analogously, if~$\balance(v) < 0$. Add the arcs~$(u, v), (v, u)$. Finally, define a new weight function~$\wf'$ for each pair of vertices~$x, y \in V$ as follows.
\[\wf'(x, y) = 
\begin{cases}
  \infty, & x = u, y = v \vee x = v, y = u \\
  \wf(v, y), & x = u \\
  \wf(x, v), & y = u \\
  \wf(x, y), & \text{otherwise}
\end{cases}
\]
\end{transformation}
\begin{lemma}
  \label{lem:smallbalancepp}
  \autoref{trans:sb} is correct, that is, it maps yes-instances and only yes-instances to yes-instances. Also, \autoref{trans:sb} can be applied exhaustively in~$\bigO(n^2m)$~time. When applied exhaustively, the resulting instance contains only vertices~$v$ with~$|\balance(v)| \leq 1$.
\end{lemma}
\begin{proof}
  The last statement of the lemma is clear. Concerning the running time, we can iterate over every vertex~$v \in V$~($\bigO(n)$~time), check if it has high absolute balance~($\bigO(m)$~time) and, if so, perform the weight function update~($\bigO(n)$~time) and perform the local modifications~($\bigO(1)$~time) for every ``excess arc'' incident to~$v$~(there are at most~$m$ many). In total, this is~$\bigO(n^2m)$~time.

  To prove the correctness, we only have to examine one application of \autoref{trans:sb}: Let~$(G'=(V',A'),\wf', \wf_{max})$ be an instance of \pWMEEs{} where \autoref{trans:sb} has been applied once at vertex~$v$ yielding the new vertex~$u$. Given an \EE{} for the input graph~$G$, we can obtain an \EE{} for~$G'$ of the same weight by modifying an arc~$a \in E$ incident to~$v$ appropriately such that it starts or ends in~$u$. If we are given an \EE{} for~$G'$, at least one arc in it has to be incident to~$u$ and thus we can obtain an \EE{} for~$G$ by modifying it to start or end in~$v$.
\end{proof}

We can apply a further preprocessing routine to make some further observations about trails in \EE s:
\begin{transformation}[Shortest-Path Preprocessing]
  \label{trans:spp}
  For an input instance of \pWMEEs{} consisting of the graph~$G = (V, A)$, the weight function~$\wf$ and the maximum weight~$\wf_{max}$, derive a new instance by computing a new weight function~$\wf'$ as follows:
  $$\wf'(u,v):=\text{weight of a shortest path from $u$ to $v$ in the graph~$(V, V \times V)$}\text{.}$$
\end{transformation}
\begin{lemma}\label{lem:spp}
  \autoref{trans:spp} is correct---that is, it maps yes-instances and only yes-instances to yes-instances---and can be applied in~$\bigO(n^3)$~time.
\end{lemma}
\begin{proof}
  It is clear that for any \EE{}~$E$ of~$G$ it holds that~$\wf'(E) \leq \wf(E)$, making any feasible \EE{} in the original instance also one for the modified instance. Now let~$E$ be an \EE{} for~$G$ with~$\wf'(E) \leq \wf_{max}$. We get an \EE{}~$E'$ for~$G$ with~$\wf(E') \leq \wf_{max}$ by exchanging every arc~$a =(u, v) \in E$ with~$\wf'(a) < \wf(a)$ by the set of arcs of a shortest path from~$u$ to~$v$ in the graph~$(V, V \times V)$ with respect to the weight function~$\wf$.

Using Dijkstra's algorithm 
we can compute in~$\bigO(n^2)$~time the weights of the shortest paths between one vertex~$v$ and any other in~$G$ and update the weight function accordingly. Doing this for every vertex in~$G$ takes~$\bigO(n^3)$~time.
\end{proof}
Shortest-path preprocessing and splitting vertices enables us to make a range of useful observations regarding trails in \EE s. In the following we assume any instance of \WMEE{} to be preprocessed using \autoref{trans:sb} and \autoref{trans:spp}. In the subsequent sections, we use this preprocessing in parameterized algorithms and reductions. Thus, we need to know whether it is parameter-preserving. This is the case, as the following observation shows. 
\begin{observation}\label{obs:sbsppinvariants}
  The number of components and the sum of all positive balances of vertices in an instance of \pWMEEs{} are invariant under \autoref{trans:sb} and \autoref{trans:spp}.
\end{observation}

\paragraph{Shortcutting Trails in \Eu Extensions.}%
Using \autoref{trans:spp}, we can define the following transformation that operates on trails of an \EE .
\begin{transformation}\label{trans:shortcut}
  Let~$E$ be an \EE{} of $G$, let~$t$ be a trail in the graph~$(V(G), E)$ and let~$s$ be a subtrail of~$t$ where~$s$ has the initial vertex~$v_A$ and the terminal vertex~$v_\Omega$. Obtain a new trail~$t'$ by substituting the edge~$(v_A, v_\Omega)$ for~$s$ in~$t$ and derive a new arc set~$E'$ by substituting~$A(t')$ for~$A(t)$ in~$E$. Define~$\shortcut(E, t, s) := (E', t')$.
\end{transformation}
\autoref{fig:shortcut} illustrates \autoref{trans:shortcut}.
\begin{figure}
  \begin{center}
    \includegraphics{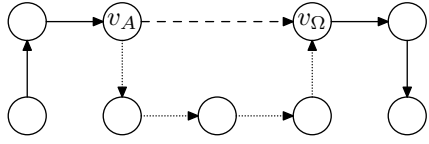}
    \caption{Example of an application of \autoref{trans:shortcut}. Solid arcs and dotted arcs belong to~$t$, dotted arcs to~$s$ and the dashed arc is substituted for the dotted arcs in~$t'$.}
    \label{fig:shortcut}
  \end{center}
\end{figure}
\begin{lemma}\label{lem:shortcut}
  Let $\shortcut(E, t, s) = (E', t')$ where the trail~$s$ has initial vertex~$v_A$ and terminal vertex~$v_\Omega$. The following statements hold:
  \begin{lemenum}
  \item $\wf(E') \leq \wf(E)$. \label{enu:ls1}
  \item Every vertex in~$V(s)$ is balanced in~$G + E'$. \label{enu:ls4}
  \item If every vertex of $s$ except $v_A$ and $v_\Omega$ is contained in a connected component of $G$ that also contains a vertex of $t'$, then the arc set $E'$ is an \EE{} for~$G$. \label{enu:ls2}
  \end{lemenum}
\end{lemma}
\begin{proof}
  Statement \enuref{enu:ls1} is trivial because of the implicitly transformed weight function (\autoref{trans:spp}).
  
  By substituting $(v_A, v_\Omega)$ for~$s$, every vertex on~$s$ except~$v_A$ and~$v_\Omega$ looses one indegree and one outdegree. Hence, augmenting~$G$ with~$E'$ results in a graph without unbalanced vertices (statement~\enuref{enu:ls4}).

  For statement~\enuref{enu:ls2} it remains to show that the graph $(V(G),A \cup E')$ is connected: If every vertex of~$s$ except~$v_A$ and~$v_\Omega$ is contained in a connected component of~$G$ that also contains another vertex of $t'$, then augmenting~$G$ with~$E'$ results in a connected graph, making $E'$ an \EE{} for~$G$ (\autoref{the:eulerian}).
\end{proof}
\begin{observation}
  \label{obs:simpletrails}
  For any \EE{} $E$ of~$G=(V, A)$ there is an \EE{}~$E'$ of at most the same weight such that any trail with arcs in~$E'$ visits every vertex at most once.
\end{observation}
\begin{proof}
  Assume that in the \EE{}~$E$ there is a trail~$t$ that visits~$v \in V$ more than once. Then there is a subtrail~\[s = (u, (u,v), v, (v, w), w)\] of~$t$ with~$u, w \in V$. Let~$(\hat{E},t') = \shortcut(E, t, s)$. By \autoref{lem:shortcut},~$\hat{E}$ is an \EE{} for~$G$ because $t'$~still visits~$v$ (one time less than~$t$). If we recursively $\shortcut$ edges in trails in~$E$ until every such trail visits any vertex at most once, we obtain an \EE{}~$E'$. By \autoref{lem:shortcut},~$\wf(E') \leq \wf(E)$.
\end{proof}
\autoref{obs:simpletrails} allows us to assume trials in \EE s to be cycles when they are closed and paths otherwise.

\begin{observation}
  \label{obs:shorttrails}
  For any \EE{} $E$ of $G$, there is an \EE~$E'$ of at most the same weight such that for any path~$p$ and any cycle~$c$ in~$E'$ such that~$p$ and~$c$ are edge-disjoint and have length at least two the following statements hold: 
  \begin{lemenum}
  \item $p$ and $c$ do not successively visit two vertices contained in exactly one connected component of~$G$. \label{enu:osht1}
  \item $p$ and $c$ do not visit one connected component of~$G$ twice except for the initial and terminal vertex. \label{enu:osht2}
  \item $p$ and $c$ have length at most the number of connected components of~$G$. \label{enu:osht3}
  \end{lemenum}
\end{observation}
\begin{proof}
The proof for \enuref{enu:osht1} and~\enuref{enu:osht2} is similar to the proof of the observation above.  Again we can $\shortcut$ edges and obtain an \EE{} of at most the same weight. Statement~\enuref{enu:osht3} directly follows from~\enuref{enu:osht1} and~\enuref{enu:osht2}.
\end{proof}
\paragraph{Shortcutting and Component Graphs.}We can further extend our observations by looking at component graphs~$\comp{G}$ and the mapping of trails~$t$ in~$G$ to trails~$\meta{G}{t}$ in~$\comp{G}$. Recall these definitions stated on page~\pageref{def:meta} in \autoref{sec:prelim}. The following lemma is a generalization of statement~\enuref{enu:ls2} in \autoref{lem:shortcut}.
\begin{lemma}
  \label{lem:metashortcut}
  Let~$E$ be an \EE{} of~$G$, let~$t, r$ be trails in the directed multigraph~$(V(G), E)$ such that the trails~$\meta{G}{r}$ and~$\meta{G}{t}$ are not vertex-disjoint. Furthermore, let~$s$ be a subtrail of~$t$ in the directed multigraph~$(V(G), E)$ such that~$\meta{G}{s}$ is a subtrail of~$\meta{G}{r}$. Let~$s'$ be a subtrail of~$t$ such that~$s$ is a subtrail of~$s'$ and~$s$ traverses exactly one vertex less than~$s'$. Set~$(E', t') = \shortcut(E, t, s')$. Then~$E'$ is an \EE{} for~$G$.
\end{lemma}
\begin{proof}
  \autoref{lem:shortcut} shows that the vertices in~$G + E'$ are balanced. It remains to show that the resulting graph is connected: Any connected component that is traversed by~$s$ is also traversed by~$u$. The trails~$\meta{G}{u}$ and $\meta{G}{t'}$ still share a vertex and thus~$G + E'$ is connected.
\end{proof}
\autoref{lem:metashortcut} leads to the following \autoref{obs:metaedgedisjoint}, which is illustrated in \autoref{fig:shortcutcomponents}.%
\begin{figure}
  \begin{center}
    \includegraphics{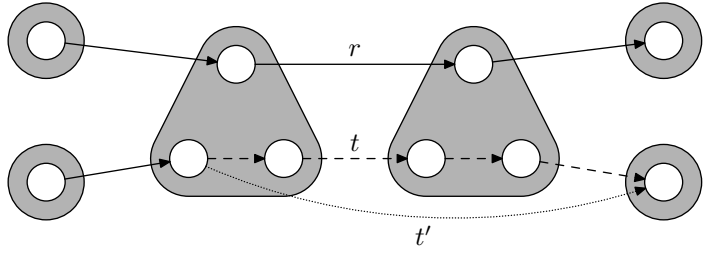}
    \caption{Gray objects represent components of~$G$. Shown are two trails~$r$ (top) and~$t$ (bottom, solid and dashed arcs) in an \EE . The trails~$\meta{G}{r}, \meta{G}{t}$ share two vertices. The dashed arcs represent a subtrail~$s'$ of~$t$ as in \autoref{lem:metashortcut} and thus we can obtain a path~$t'$ (bottom, solid and dotted arcs) replacing~$t$, while maintaining connectedness and balance of all vertices.}
    \label{fig:shortcutcomponents}
  \end{center}
\end{figure}%
\begin{observation}
  \label{obs:metaedgedisjoint}
  For any \EE{}~$E$ of~$G$ there is an \EE{}~$E'$ of at most the same weight such that for any two edge-disjoint trails~$t_1, t_2$ in~$E'$ it holds that~$\meta{G}{t_1}, \meta{G}{t_2}$ either are vertex-disjoint, share at most one vertex, or share only their initial and terminal vertices.
\end{observation}
\begin{proof}
This follows directly from \autoref{lem:metashortcut} by shortcutting subtrails that are shared by two such trails in~$\comp{G}$.
\end{proof}
We can improve this even to the following.
\begin{observation}\label{obs:metanocycle}
  For any \EE{}~$E$ of~$G$ there is an \EE{}~$E'$ of at most the same weight such that for any set of edge-disjoint trails~$\{t_1, \ldots, t_n\}$ in~$E'$ it holds that the edge-induced graph~$\comp{G}\langle\bigcup_{i=1}^n\meta{G}{t_i}'\rangle$ does not contain a cycle as subgraph, where~$\meta{G}{t_i}'$ is~$\meta{G}{t_i}$ without the initial vertex.
\end{observation}
\begin{proof}
  By \autoref{obs:simpletrails} we may assume that~$S := \{t_1, \ldots, t_n\}$ are paths or cycles. Assume that~$\meta{G}{\bigcup_{i=1}^n\meta{G}{t_i}'}$ contains a cycle~$c$ and that~$S$ is minimal with respect to this property. Let~$e \in t_i$ be an arbitrary edge on~$c$. There is a subtrail~$s$ of~$t_i$ such that~$\meta{G}{s}$ traverses~$e$ and at least one edge not belonging to~$c$---recall that~$\meta{G}{t_i}'$ is~$\meta{G}{t_i}$ without the initial vertex. Shortcutting~$s$ maintains balance of every vertex (statement~\enuref{enu:ls4}, \autoref{lem:shortcut}) and connectedness, because afterwards~$\meta{G}{t_i}$ is not vertex-disjoint from~$c$. Since an edge is removed from~$c$, it is a path after shortcutting~$s$.

  Iterating the shortcutting for every cycle in the graph~$\meta{G}{\bigcup_{i=1}^n\meta{G}{t_i}'}$ eventually removes every cycle after a finite amount of steps, because obviously the statement of the lemma holds true, if~$t_1, \ldots, t_n$ have length one, and because in every step the number of arcs in~$E$ decreases by at least one.
\end{proof}
We use \autoref{obs:metanocycle} in forthcoming \autoref{sec:wmeeslasha} to efficiently derive the structure of a suitable \EE{} for a given graph. We are now ready to prove \autoref{the:eestructure}.
\newtheorem*{rep}{\autoref{the:eestructure}}
\begin{rep}
  \eestructuretheorem
\end{rep}
\begin{proof}
  We simply take an \EE~$E$ of minimum weight for the directed multigraph~$G$ and successively remove maximum-length paths from~$E$ to obtain a set of trails~$S = \{t_1, \ldots, t_k\}$. The sought properties of the trails follow from the observations we made in this section: Statement~\enuref{enu:ees5} is trivial. From \autoref{obs:simpletrails} we can assume that each~$t_i$ either is a path or a cycle. The maximum-length~$c + 1$ of maximum-length cycles and paths~(statement~\enuref{enu:ees2}) can be assumed because, by \autoref{obs:shorttrails}, we can assume that each trail traverses at most one vertex in each component except the terminal vertex. Statement~\enuref{enu:ees7} follows directly from \autoref{obs:metaedgedisjoint}. The maximum number of maximum-length paths~$p$ and cycles~$d$ of length at least two~(statement~\enuref{enu:ees3}) can be assumed because we can assume that~$\meta{G}{p}, \meta{G}{d}$ use two edges~(\autoref{obs:shorttrails}), they are edge-disjoint~(\autoref{obs:metaedgedisjoint}) and there are at most~$c(c-1)/2$ edges in~$\comp{G}$. The upper bound~$|I^+_G|$ on the number of maximum-length paths~(statement~\enuref{enu:ees4}) can be assumed because every vertex~$v$ has~$|\balance(v)| \leq 1$ (\autoref{lem:smallbalancepp}) and each such path starts and ends in an unbalanced vertex (\autoref{obs:walkdichotomy}). Finally, statement~\enuref{enu:ees6} follows directly from \autoref{obs:metanocycle}.
\end{proof}


\section{Simplification through Advice}\label{sec:advice}

In \autoref{sec:trails} we observed that any \EE{} can be modified to conform to a restricted structure with respect to the connected components in the input graph. We will observe in \autoref{sec:incompress}, that this structure cannot be determined within polynomial time---unless~$\clacoNP \subseteq \claNP \slashpoly$, which seems unlikely. There, we implicitly use that fact, that it is not clear how components are connected through an \EE{} in order to obtain lower bounds for problem kernels. An obvious question is, whether the structure of an \EE{} can be determined using fixed-parameter algorithms whose super-polynomial-time portion depends only on the connected components of the input graph. This question is considered in the following sections.

We consider the general problem \pWMEE{} (\pWMEEs{}), and investigate its connection to the problem \pWMEEA{} (\pWMEEAs{}) in which the structure of allowed \EE s may be given by the input. In order to get a grasp at the structure of \EE s, we introduce the notion of hints and advice:
\begin{definition}
  Let~$G=(V, A)$ be a directed multigraph. A \emph{hint} for~$G$ is an undirected path or cycle~$t$ of length at least one in the component graph~$\comp{G}$ together with the information that~$t$ shall form a cycle of a path in an \EE{} of~$G$.\footnote{The extra information is necessary because a hint to a path may be a cycle in~$\comp{G}$.} We call the corresponding hints \emph{cycle hints} and \emph{path hints}, respectively. We say a set of hints~$P$ is an \emph{advice} to the graph~$G$ if the hints are edge-disjoint.
\footnote{Note that there is a difference between advice in our sense and the notion of advice in computational complexity theory. There a piece of advice applies to every instance of a specific length.} We say that a path~$p$ in the graph~$(V, V \times V)$ \emph{realizes} a path hint~$h$ if~$\meta{G}{p} = h$ and the initial vertex of~$p$ has positive balance and the terminal vertex has negative balance in~$G$. We say that a cycle~$c$ in the graph~$(V, V \times V)$ realizes a cycle hint~$h$ if~$\meta{G}{c} = h$. We say that an \EE~$E$ \emph{heeds the advice~$P$} if it contains paths and cycles that realize all hints in~$P$.
\end{definition}
Now consider the following restricted version of \pWMEEs{}:
\decprob{\pWMEEA~(\pWMEEAs{})}{A directed multigraph~$G = (V, A)$ with a weight function~$\wf: V \times V \rightarrow [0, \wf_{max}]\cup \{\infty\}$ and advice~$P$.}{Is there an \EE{}~$E$ of~$G$ that is of weight at most~$\wf_{max}$ and heeds the advice~$P$?}
\begin{figure}
  \begin{center}
    \includegraphics{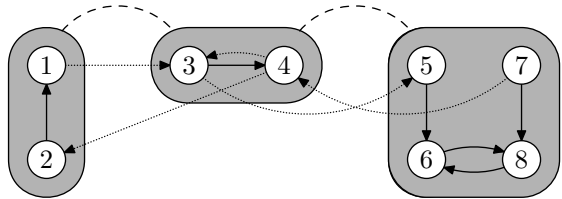}
    \caption{An instance of \pWMEEAs{} comprising the vertices~$1$ through~$8$ and the solid arcs. Gray objects represent components of the input graph~$G$ and the the dashed lines are a hint~$h$ that forms a piece of advice~$P = \{h\}$ for~$G$. The dotted arcs form an \EE{}~$E$ of~$G$. Both the paths traversing the vertices~$1, 3, 5$ and~$7, 4, 2$ realize~$h$. Thus,~$E$ heeds~$P$.}
    \label{fig:wmeeaexample}
  \end{center}
\end{figure}%
For an example of an instance of \pWMEEAs{}, see \autoref{fig:wmeeaexample}. The \pWMEEAs{} problem may be interesting in practical applications where the structure of a sought \EE{} is partly known. However, our intent is to use this problem to make the complete structure of the \EE{} explicit. We derive efficient algorithms that guess the structure as advice and then realize each hint.

In \autoref{sec:remcycle}, we simplify \pWMEEAs{} and gather a useful tool for its analysis. Then, in \autoref{sec:wmeeslasha}, we look at the relationship of \pWMEEs{} and \pWMEEAs{}. We introduce a variant of \pWMEEAs{} that seems to be easier to tackle than \pWMEEs{}. In \autoref{sec:multivariatealg}, we give an efficient algorithm for this variant that also transfers over to \pWMEEs{}. 

In the following sections, we assume all instances of \pWMEEs{} and \pWMEEAs{} to be preprocessed using \autoref{trans:sb} (``splitting vertices'') and \autoref{trans:spp} (``\SPP{}'') as introduced in \autoref{sec:trails}. We give parameterized reductions that use the parameters number of components and sum of all positive balances of vertices in the input graph. For these one can assume without loss of generality that the instances are preprocessed using the two transformations, because of \autoref{obs:sbsppinvariants}.

\subsection{Computing Realizations of Hints}\label{sec:remcycle}
In this subsection, we introduce the $\minpath$~function, which calculates minimum-weight paths that consist of allowed arcs and traverse connected components in a specific order. Using this function, we show that \pWMEEAs{} and the problem \pWMEECLA{} (\pWMEECLAs{}) are equivalent under polynomial-time many-one reductions. That is, a minimum-weight realization for any hint to a cycle can be found in polynomial time. We use this equivalence in the forthcoming sections to derive algorithms more conveniently, and to simplify reductions from and to \pWMEEAs{}.

\subsubsection{The $\minpath$ Function}
On many occasions we need to find a minimum-weight realization of a path-hint in an advice that starts and terminates in some specified vertices. Hence we need to compute a minimum-weight path that traverses vertices of components in the order given by the hint. The~$\minpath$ function defined below finds such paths.
\begin{definition}\label{def:minpath}
  Let the directed multigraph~$G=(V,A)$ and the weight function~$\wf: V \times V \rightarrow [0, \wf_{max}]\cup \{\infty\}$ constitute an instance of \pWMEEs{}. Let~$p$ be a path in~$\comp{G}$ and let~$u$ be a vertex in the component of~$G$ that corresponds to the initial vertex of~$p$ and~$v$ a vertex in the component that corresponds to the terminal vertex of~$p$. Define~$\minpath(G,\wf, p, u, v)$ as the shortest path~$s$ from~$u$ to~$v$ in the complete graph~$(V, V \times V)$ such that~$\meta{G}{s} = p$.
\end{definition} 
Recall that we have made \SPP{} (\autoref{trans:spp}) implicit at the start of this section. Thus, by \autoref{obs:shorttrails}, we may assume that any shortest path in~$(V, V \times V)$ with respect to the weight function~$\wf$ does not successively visit two vertices of one connected component of~$G$. This gives the following strategy to compute $\minpath(G, \wf, p, u, v)$:

Orient the path~$p$ to obtain a directed path~$p'$. 
Initialize a new weight function~$\wf'$ that assigns every arc in~$V \times V$ the weight~$\infty$. Iterate over the arcs of $p'$. For any such arc~$(c_1, c_2)$ let~$C_1, C_2$ be the corresponding components. For every arc~$(w, x)\in C_1 \times C_2$ set~$\wf'(w,x) := \wf(w,x)$. Now, using the weight function $\wf'$, compute a shortest path~$s$ from~$u$ to~$v$ in the graph~$(V, V \times V)$. Return~$s$. See also the pseudocode in \autoref{alg:minpath}.
\begin{algorithm}
  \LinesNumbered
  
  \KwIn{A directed multigraph~$G=(V, A)$, a weight function~$\wf: V \times V \rightarrow [0, \wf_{max}]\cup \{\infty\}$, a path~$p$ in~$\comp{G}$, and vertices~$u, v$ in the components~$C^{u}, C^{v}$ corresponding to the initial and terminal vertices of~$p$, respectively.}
  \KwOut{A minimum-weight path~$s$ from~$u$ to~$v$ in~$(V, V \times V)$ such that~$\meta{G}{s} = p$.}

  \SetKwData{npath}{Path}

  \BlankLine
  \tcc{Orient the path $p$.}
  $p' \leftarrow$ a path that is an orientation of~$p$ and starts in the vertex corresponding to~$C^{u}$ and terminates in the vertex corresponding to~$C^{v}$\;

  \tcc{Initialize a modified weight function $\wf'$.}
  \lFor{$w, x \in V$}{$\wf'(w,x) \leftarrow \infty$\;}
  \For{$(c_1,c_2) \in p'$}{
    $C_1 \leftarrow$ connected component of $G$ corresponding to $c_1$\;
    $C_2 \leftarrow$ connected component of $G$ corresponding to $c_2$\;
    \lFor{$w \in C_1, x \in C_2$}{$\wf'(w, x) \leftarrow \wf(w, x)$\;}
  }
  $\npath \leftarrow $ a shortest path from $u$ to $v$ in the complete directed graph with the vertices of $G$ and with weight function $\wf'$\;
  
  \Return{\npath\;}
    
  \SetAlgoRefName{Min\-Path}
  \caption{Finding minimum-weight paths that traverse components in a specified order.}
  \label{alg:minpath}
\end{algorithm}
\begin{lemma}
  \label{lem:minpath}
  \autoref{alg:minpath} computes the function~$\minpath(G,\wf,p,u,v)$ in~$\bigO(n^2)$~time.
\end{lemma}
\begin{proof}
  Consider $p_{min} = \minpath(G,\wf,p,u,v)$. This path retains its weight under the weight function~$\wf'$. It follows that the output~$s$ of \autoref{alg:minpath} has at most the weight of~$p_{min}$. However, since in any vertex of a component of~$G$ only arcs that lead to the next component according to~$p'$ may have weight~$\leq \infty$, we may assume that~$\meta{G}{s} = p$ and thus~$\wf(s) \geq \wf(p_{min})$.

  The dominating running time portion is in the computation of a shortest path in line 7, which is possible in~$\bigO(n^2)$~time using Dijkstra's algorithm (there are no negative weights in~$\wf'$).
\end{proof}
Using the~$\minpath$ function, we can formulate a fact about \EE s that we use in reductions involving \pWMEEAs{}.
\begin{observation}\label{obs:minpathinee}
  Let $E$ be an \EE{} for the multigraph~$G$ that heeds the advice~$P$, let~$P$ contain a path-hint~$h$ and let~$\wf$ be a weight function~$V \times V \rightarrow [0, \wf_{max}]\cup \{\infty\}$. There is an \EE{}~$E'$ such that the following statements hold:
  \begin{lemenum}
  \item $E'$~heeds the advice~$P$,
  \item $\wf(E') \leq \wf(E)$, and 
  \item $A(\minpath(G, \wf, h, u, v)) \subseteq E'$.
  \end{lemenum}
  Here,~$u, v$ are vertices contained in the connected components of~$G$ that correspond to the initial and terminal vertices of~$h$, respectively.
\end{observation}
\begin{proof}
  \autoref{obs:minpathinee} is easy to prove: Simply remove the realization~$p$ of~$h$ from~$E$ and add the edges of~$\minpath(G, \wf, h, u, v)$ where~$u, v$ are the initial and terminal vertices of~$p$, respectively.
\end{proof}

\subsubsection{Removing Cycles from an Advice}
Now regarding hints to cycles, we may proceed as in \autoref{alg:findcycle2} (see page~\pageref{alg:findcycle2}): First we introduce a new component~$K'$ that is a copy of an arbitrary component~$K$ visited by the given cycle hint~$c$ (lines~1 and~2). Then we extend the weight-function~$\wf$ such that any arc in~$V \times V$ that contains a vertex~$v$ of~$K'$ is assigned the same weight as the arc that contains the original vertex in~$K$ (lines~3 to~5). We then split the cycle~$c$ to a path~$p$ that goes from~$K$ to~$K'$ (lines~6 to~9). Then for every vertex~$v \in K$ we compute~$\minpath(G,\wf,p, v, v')$ and $\minpath(G,\wf,p, v', v)$ where~$v'$ is the copy of~$v$ in~$K'$. This is done in lines~11 to~18. The shortest path found in this procedure is modified such that the vertex it contains in~$K'$ is replaced by its original in~$K$. This modified path is returned.
\begin{algorithm}
  \LinesNumbered
  
  \KwIn{A directed multigraph~$G=(V, A)$, a weight function~$\wf: V \times V \rightarrow [0, \wf_{max}]\cup \{\infty\}$ and a cycle~$c$ in~$\comp{G}$.}
  \KwOut{A minimum-weight cycle in $G$ that occurs in an \EE{} of $G$ that heeds an advice containing~$c$.}

  \SetKwData{spath}{CurrentShortestPath}
  \SetKwData{npath}{Path}
  \SetKwData{npath'}{Path'}

  \BlankLine
  \tcc{Introduce a new component to split the cycle.}
  $K \leftarrow$ an arbitrary component of $G$ that is visited by $c$\;
  $G \leftarrow$ $G$ with an additional copy $K'$ of $K$\;
  \For{$(v, w) \in K \times V$}{
    $v' \leftarrow$ the copy of $v$ in $K'$\;
    $\wf(v', w) \leftarrow \wf(v, w)$\;
  }
  $k \leftarrow$ the vertex in~$\comp{G}$ that corresponds to $K$\;
  $k' \leftarrow$ the vertex in~$\comp{G}$ that corresponds to $K'$\;
  $\{k, v\} \leftarrow$ an edge in $c$ that is incident to $k$\;
  $p \leftarrow$ $c \setminus (\{\{k, v\}\} \cup \{\{k', v\}\})$\;
  \tcc{Probe vertices for shortest cycles.}
  \spath $\leftarrow$ empty list\;
  \For{$v \in K$}{
    $v' \leftarrow$ the copy of $v$ in $K'$\;
    $\npath \leftarrow \minpath(G,\wf,p, v, v')$\;
    $\npath' \leftarrow \minpath(G,\wf,p, v', v)$\;
    \If{$\wf'(\npath) < \wf'(\spath)$}{$\spath \leftarrow \npath$\;}
    \If{$\wf'(\npath') < \wf'(\spath)$}{$\spath \leftarrow \npath'$\;}
  }
  
  \Return{\spath with every vertex in $\spath \cap K'$ replaced by its original in~$K$\;}
  
  \SetAlgoRefName{Determine\-Cycle}
  \caption{Finding minimum-weight cycles with advice.}
  \label{alg:findcycle2}
\end{algorithm}

\begin{lemma}\label{lem:findcycle}
  The output returned by \autoref{alg:findcycle2} is a cycle that is contained in a minimum-weight \EE{}~$E$ for~$G$ that heeds an advice~$P$ such that $P$ contains the input cycle~$c$. The algorithm runs in~$\bigO(n^3)$~time.
\end{lemma}
\begin{proof}
  It is easy to see that the output is a cycle: The algorithm computes a path from~$v \in K$ to its copy $v' \in K'$. However, $v'$~is replaced by~$v$ in the final step in line~19.

  Since the \EE{}~$E$ heeds some advice that contains the cycle-hint~$c$, it contains a number of closed trails that all visit the components whose corresponding vertices in~$\comp{G}$ are contained in~$c$. Let~$c^G_{min}$ be a trail that is of minimum-weight among those trails. Because of \SPP{} and \autoref{obs:shorttrails} we may assume that~$c^G_{min}$ is a cycle that contains exactly one vertex of every component it visits. By copying an arbitrary component~$K$ this cycle visits and modifying the cycle so that it starts in one vertex~$v$ of~$K$ and ends in the copy of~$v$, we obtain a path of the same weight. That is, the path found by \autoref{alg:findcycle2} has at most the weight of~$c^G_{min}$. However, it may not find a cycle that is of lower weight than~$c^G_{min}$, otherwise~$E$ is not of lowest weight.
  
  Regarding the running time, lines~1 and~2 can be carried out in~$\bigO(n + m)$~time. Extending the weight function in lines~3 to~5 is possible in~$\bigO(n^2)$~time. Lines~6 to~9 take time at most~$\bigO(n)$ using list-implementations of paths. The loop in line~11 is executed at most $n$~times and every iteration takes~$\bigO(n^2)$~time using \autoref{alg:minpath}. Summing up, we get a bound of~$\bigO(n^3)$~time.
\end{proof}
\autoref{lem:findcycle} yields the following theorem:
\begin{theorem} 
  \label{obs:redwmeeatocla}
  \pWMEEA{} and \pWMEECLA{} are equivalent under polynomial-parameter polynomial-time many-one reductions when parameterized by the number of connected components and/or the sum of positive balances of all vertices.
\end{theorem}
\begin{proof}
  Since \pWMEECLAs{} is a subset of \pWMEEAs{} this direction is trivial. To reduce \pWMEEAs{} to \pWMEECLAs{} simply use \autoref{alg:findcycle2} for every cycle-hint in the advice and add the corresponding cycle to the input graph. This is a polynomial-time many-one reduction, because it can be carried out in~$\bigO(|P|n^3)$~time and it is correct because of \autoref{lem:findcycle}. Also, by carrying out the reduction the number of components does not increase and the balance of all vertices stays the same. As a consequence, this is a polynomial-parameter polynomial-time reduction for these parameters.
\end{proof}
\autoref{obs:redwmeeatocla} enables us to simplify reductions and algorithms for \pWMEEAs{} by using the equivalence of \pWMEEAs{} and \pWMEECLAs{} and by considering the simpler problem of \pWMEECLAs{} instead.

\subsection{The Impact of Advice}\label{sec:wmeeslasha}

In this section we investigate the relationship of \pWMEEs{} and \pWMEEAs{}. 
For this, we consider the following restricted form of advice and corresponding problem \pWMEECA{} (\pWMEECAs{}).
\begin{definition}
  Let~$G$ be a directed multigraph and let~$P$ be an advice for~$G$. We call the advice~$P$ \emph{connecting}, if the hints in~$P$ connect every vertex in~$\comp{G}$.
\end{definition}
\decprob{\pWMEECA}{A directed multigraph~$G = (V, A)$ with a weight function~$\wf: V \times V \rightarrow [0, \wf_{max}]\cup \{\infty\}$ and minimal connecting advice~$P$.}{Is there an \EE{}~$E$ of~$G$ that is of weight at most~$\wf_{max}$ and heeds the advice~$P$?} \label{def:pWMEECA}

We show that \pWMEEs{} is parameterized Turing-reducible to \pWMEECAs{} when parameterized by the number~$c$ of components in the input graph or the combined parameter of~$c$ and the sum~$b$ of all positive balances of vertices in the input graph. And we also give a polynomial-time polynomial-parameter many-one reduction from \pWMEEAs{} to \pWMEEs{} with respect to the parameter number of connected components in this section. 

Since in \autoref{sec:incompress} we will show that a polynomial-size problem kernel for \pWMEEs{} would imply~$\clacoNP \subseteq \claNP \slashpoly$ and since in \autoref{sec:releecbm} we will give a polynomial-size problem kernel for \pWMEECAs{}, we cannot hope to replace the Turing reduction with a polynomial-time polynomial-parameter many-one reduction. Otherwise we could derive a polynomial-size problem kernel for \pWMEEs{} using this reduction. 

In terms of classical complexity theory, the parameterized Turing reduction is a very powerful tool, and thus, one could hope for \pWMEECAs{} being polynomial-time solvable. This, however, is unlikely. Although the reductions given in this section do not imply a hardness result for \pWMEECAs{}, we gather \NPhs{} as a simple corollary (\autoref{cor:eecanph}) in \autoref{sec:releecbm}. Nevertheless, the reductions given in this section are of high value to us, because we can use the Turing reduction to derive an efficient algorithm for \pWMEEs{} in \autoref{sec:multivariatealg} and together with the second reduction, we can restate \pWMEEs{} as a matching problem in \autoref{sec:matching}.

\paragraph{Simple Observations Regarding \pWMEECAs{}.}
For running time analysis, we sometimes need to know the maximum number of hints in an advice in \pWMEECAs{}. Here, the following observation is helpful.
\begin{observation} \label{obs:hintsinadvice}
  Let~$G$ be a directed multigraph with~$c$ connected components and let~$P$ be a minimal connecting advice for~$G$. The advice~$P$ contains at most~$c$ hints.
\end{observation}
\begin{proof}
  Since a hint is a path or cycle of length at least one, it connects at least two vertices in~$\comp{G}$. We consider the graph~$(V(\comp{G}), \emptyset)$ and the procedure of successively adding hints~$h_1, \ldots, h_k$ that form a minimal connecting advice. It is clear that every hint~$h_i, 1 \leq i \leq k$, must connect two connected components of the graph~$(V(\comp{G}), \bigcup_{j = 1}^{i - 1}E(h_j))$. Otherwise we could remove~$h_i$ and still connect every vertex in~$\comp{G}$ using the remaining hints. Thus, adding~$c$ hints connects every vertex in~$\comp{G}$ and there are at most~$c$ hints in~$P$.
\end{proof}
It is also easy to see, that we can realize every cycle hint in a minimal connecting advice to obtain a cycle-free minimal connecting advice.
\begin{observation}\label{obs:redwmeecatoccla}
  \pWMEECA{} is equivalent to \pWMEECCLA{} (\pWMEECCLAs{}) under polynomial-parameter polynomial-time many-one reductions with respect to the parameters number of connected components and sum of all positive balances of vertices.
\end{observation}
\begin{proof}
  See \autoref{obs:redwmeeatocla}.
\end{proof}


\subsubsection{Reducing \pWMEEs{} to \pWMEECAs{}} To reduce \pWMEEs{} to \pWMEECAs{} the obvious idea of trying pieces of advice yields a Turing reduction. We make use of the observations in \autoref{sec:trails} to assume certain restrictions on the pieces of advice we have to guess.
\begin{lemma}\label{lem:adstructure}
Let~$G$ be a directed multigraph and let~$E$ be a minimum-weight \EE{} with respect to a weight function~$\wf: V \times V \rightarrow [0, \wf_{max}]\cup \{\infty\}$ for~$G$. There is a minimal connecting advice~$P = \{h_1, \ldots, h_i\}$ such that
\begin{lemenum}
\item $E$ heeds~$P$, and
\item the graph defined by the union of all trails~$h_1, \ldots, h_i$ without their initial vertices does not contain a cycle.
\end{lemenum}
\end{lemma}
\begin{proof}
  This is mainly based on \autoref{the:eestructure}. By the theorem, there is a decomposition of~$E$ into paths and cycles~$t_1, \ldots, t_k$ such that the graph defined by the union of all trails~$\meta{G}{t_1}, \ldots, \meta{G}{t_k}$ without their initial vertices does not contain a cycle. We greedily take paths~$\meta{G}{t_j}$ of length at least one into~$P$ that connect new vertices in~$\comp{G}$.
\end{proof}
Using this restriction, we can guess all forests of~$\comp{G}$ and try all possibilities to extend them to an advice:
\begin{lemma}\label{lem:redwmeetowmeea}
  \pWMEE{} is parameterized Turing-reducible to \pWMEECA{} when parameterized by the number~$c$ of components in the input graph or the combined parameter of~$c$ and the sum of all positive balances of vertices in the input graph. The reduction can be carried out in~$\bigO(16^{c\log(c)}(c + n + m))$~time.
\end{lemma}
\begin{proof}
  Let the directed multigraph~$G=(V,A)$ and the weight function~$\wf: V \times V \rightarrow [0, \wf_{max}]\cup \{\infty\}$ constitute an instance of \pWMEEs{} and let~$c$ be the number of connected components in~$G$. We give an algorithm that decides \pWMEEs{} using an oracle for \pWMEECAs{} in time~$\bigO(2^{c^2\log(c)}(c^3 + n + m))$.

We simply generate all possible pieces of advice and apply the oracle to the resulting instances. If one of the oracle calls accepts the advice-instance, then, clearly, the original instance is a yes-instance. Also, for every yes-instance of \pWMEEs{}, there is an advice derivable from a solution to the instance because of \autoref{lem:adstructure}. Clearly, the number of components and the sum of all positive balances remain the same in the instances passed to the oracle.

Concerning the generation of the pieces of advice, by \autoref{lem:adstructure} we may assume that the hints without their initial vertices form a forest in~$\comp{G}$. Thus, we may simply enumerate all forests contained in~$\comp{G}$, partition their edges into at most~$c$ hints and try all possibilities to add the initial vertex back onto the hints.

To enumerate all forests, we first partition the vertices into at most~$c$ cells (there are~$c^c$ many such partitions), then enumerate all spanning trees in each cell (in each cell there are~$c^{c-2}$ spanning trees~\cite{Cay89}). This is possible in~$\bigO(c^c(c^{c-2} + c^2)) = \bigO(c^{2c-2})$~time~\cite{KR95}. 

We then partition the edges into at most~$c$ hints (there are~$c^c$ partitions), extend every hint by adding an initial vertex (in total, there are~$c^c$ possibilities) and check if this yields a valid advice---that is, whether the hints are paths or cycles and whether the advice is connecting. This procedure can be carried out in~$\bigO(c^{2c}c^3)$~time allowing~$\bigO(c^3)$ for the validity check.

For every guessed advice, we have to pass the instance to the oracle in linear time and, since~$n^n=2^{n\log(n)}$, we can derive the running time bound of~$\bigO(16^{c\log(c)}(c + n + m))$.
\end{proof}
\subsubsection{Reducing \pWMEEAs{} to \pWMEEs{}} Here, we will see that there is only a polynomial number of optimal ways to realize a hint in an advice. Each of these realizations will be modeled by a pair of imbalanced vertices. These pairs will reside in a new component and this component then can only be connected to the rest of the graph by taking arcs into an \EE{} that also connect each component corresponding to inner vertices of the hint. 

For convenience, we give a reduction from \pWMEECLAs{} (see \autoref{sec:remcycle}) instead of \pWMEEAs{}. This is without loss of generality because of \autoref{obs:redwmeeatocla}. We first give an intuitive description, followed by detailed construction and then a correctness proof. The construction uses the $\minpath$ function introduced on page~\pageref{def:minpath} in \autoref{sec:remcycle}.

\paragraph{Intuitive Description.} We look at the hints present in an \pWMEECLAs{} instance and eliminate them one at a time: For every hint~$p_i, 1 \leq i \leq d$, first, a connected component is introduced (vertex set~$W^i_1$, arc sets~$B_1^{i,\pm}, B_1^{i,=}$ in the construction below) and copied for every inner vertex of the hint (vertex sets~$W^i_l$, arc sets~$B_l^{i,\pm}, B_l^{i,=}$ for~$2\leq l \leq k-1$). Each copy is connected to the component corresponding to its vertex in the hint (by the arc-set~$B_l^{i, \gamma}$). The new component and its copies consist of interconnected imbalanced pairs of vertices. In the construction below, these are the vertices~$s^{i,\pm}_{l, u,v},t^{i,\pm}_{l, u,v}$ contained in the~$i$-th component. Each pair corresponds to a pair of vertices~$u,v$ forming the endpoints of a path that realizes the currently considered hint~$p_i$.

 A new weight function gives meaning to the construction and ensures that adding an arc~$(u, t^{i,+}_{1, u, v})$ or an arc~$(s^{i,-}_{1, u, v}, v)$ to an \EE{} has the same weight as a minimum-weight realization of the hint that goes from~$u$ to~$v$ or from~$v$ to~$u$, respectively. Notice that the superscript~``$+$''corresponds to paths in one direction and the superscript~``$-$'' to paths in the opposite direction. The weight function also ensures that if such an arc is present in an \EE{}, then the connected components traversed by the hint are connected to each other.

\begin{construction}\label{cons:redwmeeatowmee}
  Let the directed multigraph~$G_0=(V_0,A_0)$, the weight-function~$\wf_0: V_0 \times V_0 \rightarrow [0, \wf_{max}]\cup \{\infty\}$, and the advice~$P$ constitute an instance~$I_\pWMEECLAs{}$ of \pWMEECLAs. Let~$p_1, \ldots, p_d$ be the elements of~$P$ and let~$C_1, \ldots, C_c$ be the connected components of~$G$.


  For every~$p_i, 1\leq i \leq d$, inductively define~$G_i$ and~$\wf_i$ as follows: Let~$C_{j_1}, \ldots, C_{j_{k}}$ be the components of~$G$ that correspond to the vertices traversed by~$p_i$, ordered according to an arbitrary path orientation of~$p_i$. For every~$1 \leq l \leq k - 1$ introduce the vertex set
\begin{align*}
  W_l^{i,+} & :=\{t^{i,+}_{l,u,v}, s^{i,+}_{l,u,v}:u \in C_{j_1} \cap I^+_G \wedge v \in C_{j_k} \cap I^-_G\} \text{, and}\\
  W_l^{i,-} & := \{s^{i,-}_{l,u,v}, t^{i,-}_{l,u,v}:u \in C_{j_1} \cap I^-_G \wedge v \in C_{j_k} \cap I^+_G\} \text{.}
\end{align*}
Set~$W^i_l := W_l^{i,+} \cup W_l^{i,-}$. Make all these vertices imbalanced via the arc set
\[B_l^{i,\pm} := \{(t^{i,+}_{l, u, v}, s^{i,+}_{l,u,v}), (t^{i,-}_{l,u,v}, s^{i,-}_{l, u, v})\} \text{.}\]
Let~$w^1_l, \ldots, w^h_l$ be the vertices in~$W^i_l$. For each~$1 \leq l \leq k-1$, interconnect these vertices via a cycle, using the following arc set
\[B_l^{i,=} := \{(w^g_{l}, w^{g+1}_{l}) : 1 \leq g < h\} \cup \{(w^h_l, w^1_l)\} \text{.}\]
Furthermore, for each~$2 \leq l \leq k -1$, choose~$c_{j_l} \in C_{j_l}$ and~$w_l \in W^i_l$ arbitrarily and add the following arc set connecting~$W_l^i$ to~$C_{j_l}$:
\[B_l^{i,\gamma} := \{(w_l, c_{j_l}), (c_{j_l}, w_l)\} \text{.}\]
Now set~$G_i = (V_i, A_i) := (V_{i - 1} \cup \bigcup_{l = 1}^{k-1} W^i_l, A_{i -1} \cup \bigcup_{l = 1}^{k-1} (B_l^{i,\pm} \cup B^{i,=}_l)\cup \bigcup_{l = 2}^{k-1} B_l^{i,\gamma}) )$ 
and create a new weight function as follows:
\[\wf_i(u, v) := 
\begin{cases}
  \wf_{i - 1}(u, v), & u, v \in V_{i - 1} \\
  \wf_0(\minpath(G_0, \wf_0, p_i, u, x)), & u \in C_{j_1} \cap I^+_G, v = t^{i,+}_{1, u, x} \\
  \wf_0(\minpath(G_0, \wf_0, p_i, x, v)), & u = s^{i,-}_{1, x, v} , v \in C_{j_1} \cap I^-_G \\
  0, & u = s^{i,+}_{k-1, x, v}, v \in C_{j_k} \cap I^-_G \\
  0, & u \in C_{j_k} \cap I^+_G , v = t^{i,-}_{k-1, u, x}  \\
  0, & u = s^{i,\pm}_{l, x, y}, v = t^{i,\pm}_{l, x, y} \\
  0, & u = s^{i,\pm}_{l, x, y}, v = t^{i,\pm}_{l + 1, x, y} \\
  \infty, & \text{otherwise}\\
\end{cases}
\]

The graph~$G_d$, the weight function~$\wf_d$ and the number~$\wf_{max}$ constitute an instance~$I_\pWMEEs{}$ of \pWMEEs{}.
\end{construction}
\begin{figure}
  \begin{center}
    \subfloat[\pWMEECLAs{} instance]{
      \includegraphics{wmeeslasha.1}
      \label{fig:redwmeeatowmeea}
    } \\
    \subfloat[\pWMEEs{} instance]{
      \includegraphics{wmeeslasha.2}
      \label{fig:redwmeeatowmeeb}}
    \caption{Example for \autoref{cons:redwmeeatowmee} explained in \autoref{ex:redwmeeatowmee}.}
    \label{fig:redwmeeatowmee}
  \end{center}
\end{figure}
\begin{example}\label{ex:redwmeeatowmee}
  Have a look at \autoref{fig:redwmeeatowmee}. At the top, an instance~$I_\pWMEECLAs{}$ of \pWMEECLAs{} is shown. It comprises three connected components and an advice consisting of a single hint~$p_1$ represented by the dashed edges. Below, there is an instance~$I_\pWMEEs{}$ of \pWMEEs{} produced by \autoref{cons:redwmeeatowmee}. The dotted arcs represent the only arcs incident to the new vertices with weight potentially lower than~$\infty$. 

In the new instance the hint~$p_1$ is removed and a new component~$W^1_1$ is introduced. A copy~$W^1_2$ of the vertex set~$W^1_1$ is introduced and connected to the component that corresponds to the inner vertex of~$p_1$. The induced subgraphs of~$W^1_1, W^1_2$ consist of pairs~$t^{i,+}_{l,u,v}, s^{i,+}_{l,u,v}$ of vertices that are made imbalanced via a direct arc and that are connected via a directed cycle. Each of the vertices~$s^{i,+}_{l,u,v}$---the ``sources''---has balance~$1$ and can either be connected to a vertex~$t^{i,+}_{l,u,v}$---the ``targets''---inside the same component or to another component. Analogously, targets can only accept at most one arc from either inside the same component or from outside. 

Consider a solution~$E$ to~$I_\pWMEECLAs{}$ that also contains the arcs~$(1, 3), (3,5)$ as realization of~$p_1$. We may remove these arcs and add the arcs~\[(1, t^{1, +}_{1,1,5}), (s^{1,+}_{1, 1, 5}, t^{1, +}_{2, 1, 5}), (s^{1, +}_{2, 1, 5}, 5)\] to~$E$, and add arcs from all remaining sources to their corresponding targets that reside in the same component to obtain a solution to~$I_\pWMEEs{}$. Also, every solution to~$I_\pWMEEs{}$ has to connect the connected component~$W^1_1$ to the rest of the graph. This is only possible by adding an arc from a source to outside its component, for example at~$s^{1,-}_{1, 6, 2}$. Then the vertex~$t^{1, -}_{1, 6, 2}$ has to fetch an arc from~$s^{1, -}_{2, 6, 2}$ in the \EE{} in order to become balanced. This means that then also the arc~$(6,t^{1, -}_{2, 6, 2})$ has to be included in an \EE{} for~$I_\pWMEECLAs{}$ and thus we can include the path from vertex~$6$ to vertex~$2$ that realizes~$p_1$ computed by the~$\minpath$ function.
\end{example}
\paragraph{Correctness.} We first prove that \autoref{cons:redwmeeatowmee} is polynomial-time computable and that the parameter in the reduced instance is polynomial in the original parameter. We then proceed to show the soundness of the construction.
\begin{observation}\label{obs:redwmeeatowmeepoly}
  \autoref{cons:redwmeeatowmee} is polynomial-time computable. There are~$\bigO(c^2)$ components in~$G_d$.
\end{observation}
\begin{proof}
  We first look at the running time of the construction: The size of~$W^i_l$ and the arc sets~$B^{i,\pm}_l, B^{i,=}_l, B^{i,\gamma}_l$ is at most~$\bigO(n^2)$. It holds that~$l \leq c$ and there are at most~$\bigO(c^2)$ hints in an advice (recall that hints in an advice are edge-disjoint). Hence, at most~$\bigO(c^3n^2)$ vertices and edges are added. This can be done in time linear in the number of  added vertices and edges. Thus, the new weight-function can be computed in~$\bigO(c^6n^4)$~time and this yields a polynomial-time algorithm for \autoref{cons:redwmeeatowmee}.

  Since there are at most~$\bigO(c^2)$ hints in an advice and for every hint, there is exactly one new component (the component with vertex-set~$W^i_1$) in the reduced instance, the new parameter is in~$\bigO(c^2)$.
\end{proof}
\begin{lemma}
  \autoref{cons:redwmeeatowmee} is a polynomial-parameter polynomial-time reduction.
\end{lemma}
\begin{proof}
  By \autoref{obs:redwmeeatowmeepoly} it only remains to show that \autoref{cons:redwmeeatowmee} is correct. For this, first consider an \EE{}~$E$ that is a solution to~$I_\pWMEECLAs{}$. For every hint~$p_i$ the set~$E$ contains a set of paths that realize that hint. Without loss of generality we may assume that among those paths is~$s = \minpath(G_0, \wf_0, p_i, u, v)$ for suitable vertices~$u, v$ in the components that~$p_i$ starts and ends, respectively (see \autoref{obs:minpathinee}). Thus, in order to connect the component~$W^i_l$ to the rest of the graph, we may remove~$s$ from~$E$ and add the arcs
\[(u, t^{i,+}_{1,u,v}), (s^{i,+}_{1,u,v}, t^{i,+}_{2,u,v}), \ldots, (s^{i,+}_{k-2,u,v}, t^{i,+}_{k-1,u,v}),(s^{i,+}_{k-1,u,v}, v) \text{.}\]
This does not increase the weight of~$E$. To balance all vertices~$t^{i,+}_{l, u', v'}, s^{i,+}_{l, u', v'}$ with~$1 \leq l \leq k - 1, u' \neq u, v' \neq v$, we may add the corresponding arcs $(s^{i,+}_{l,u', v'}, t^{i,-}_{l, u', v'})$ and analogously for vertices in~$W^{i,-}_l$, again without increasing the weight. Thus, doing this for every hint yields an \EE{} for~$I_\pWMEEs{}$ of the same weight.

Now consider an \EE{}~$E$ that is a solution to~$I_\pWMEEs{}$. The set~$E$ has to connect the component~$W^i_1$ to the rest of the graph for every hint~$p_i$. Thus, without limitation of generality, there is an arc~$(u, t^{i,+}_{1, u, v})$ for some vertices~$u, v$ in the components that correspond to the endpoints of~$p_i$. For every vertex~$t^{j,\pm}_{l, x, y}$ all arcs with weight lower than~$\infty$ end in it, and since it has balance~$-1$, there is exactly one arc incident to it in~$E$. The same is true for vertices~$s^{j,\pm}_{l, x, y}$ since all arcs with weight lower than~$\infty$ start at them and they have balance~$1$. Hence the arc~$(s^{i, +}_{1, u, v}, t^{i, +}_{2, u, v})$ is present in~$E$, by induction also~$(s^{i, +}_{l, u, v}, t^{i, +}_{l+1, u, v}) \in E, 1 \leq l \leq k-2$, and finally also~$(s^{i, +}_{k-1,u,v}, v) \in E$. Thus we can remove these arcs from~$E$, add~$\minpath(G_0, \wf_0, p_i, u, v)$, and repeat this for all hints to obtain an \EE{} for~$G_0$ that heeds the advice~$P$ and has weight at most~$w_{max}$.
\end{proof}
\begin{theorem}\label{the:redwmeeatowmee}
  \pWMEEA{} is polynomial-time polynomial-parameter many-one reducible to \pWMEE{} when parameterized by the number of components in the input graph.
\end{theorem}

\subsection{An Efficient Multivariate Algorithm for \pWMEECAs{}}\label{sec:multivariatealg}
In this section we consider~\pWMEECAs{} parameterized by both the number of components~$c$ in the input graph and the sum~$b$ of all positive balances of vertices in the input graph. A simple idea is used to obtain an efficient algorithm that solves \pWMEECAs{}. We prove the following theorem:

\begin{theorem}\label{the:kcalgwmeea}
  \pWMEECA{} is solvable in~$\bigO(4^{c\log(b)}n^2(b^2+ n\log(n)) + n^2m)$~time, where~$c$ is the number of components in the input graph and where~$b$ is the sum of all positive balances of vertices in the input graph.
\end{theorem}

In a simple corollary, we also obtain an efficient algorithm for \pWMEEs{}, proving that this problem is fixed-parameter tractable with respect to the combined parameter~$(b, c)$. We deem parameterizing with both~$b$ and~$c$ to be a good choice: 
The reduction we use to show \NPhs{} for \pWMEEs{} in \autoref{sec:relhcwmee} creates instances where~$b = 0$ implying that parameterizing only with~$b$ does not suffice to obtain efficient algorithms. Also, the question whether \pWMEEs{} is fixed-parameter tractable with respect to parameter~$c$ is a long-standing open question dating back to \citet{Fre77}. We reflect on the parameter~$c$ in \autoref{sec:matching} and it seems hard to answer this question. 
 

To obtain an algorithm for \pWMEECAs{}, we use the fact that minimum-weight \EE s for connected multigraphs can be found in~$\bigO(n^3\log(n))$~time~\cite{DMNW10}. To derive a connected instance of \pWMEEs{} from an instance of \pWMEECAs{}, we realize all hints in the given minimal connecting advice. The parameter~$b$ helps to bound the number of possible ways we have to try to realize each hint. An algorithm that achieves the running time given in \autoref{the:kcalgwmeea} can simply try each combination of optimal realizations of each hint in the given advice and then solve the resulting instance comprising a connected multigraph via the polynomial-time algorithm given by \citet{DMNW10}. We denote a call to this algorithm by~$\solveconnected(G,\wf)$, where~$G$ is a connected multigraph and~$\wf: V \times V \rightarrow  [0, \wf_{max}]\cup \{\infty\}$ is a weight function. 

\paragraph{Solution Algorithm.} For convenience, we give an algorithm that solves \pWMEECCLAs{} which we then generalize to an algorithm for \pWMEECAs{} using \autoref{obs:redwmeecatoccla}. 
In \autoref{alg:kcalgwmeecla} a description of the solution algorithm is shown in pseudo code. It is invoked with an instance of \pWMEECCLAs{} and an empty set~$E$. The set~$E$ is then successively extended to a minimum-weight \EE . This is done by iterating over every local-optimal realization of each hint in lines~9 and~10 and recursing for every of them. When each hint is realized, that is~$P = \emptyset$ in line~1, the resulting instance is solved in polynomial time.

\begin{algorithm}
  \LinesNumbered
  
  \KwIn{A directed multigraph~$G=(V, A)$, a weight function~$\wf: V \times V \rightarrow [0, \wf_{max}]\cup \{\infty\}$, a cycle-less advice~$P$ and an arc-set~$E$.}
  \KwOut{A minimum-weight \EE{} for~$G$ that heeds the advice~$P$.}

  \SetKwFunction{solve}{SolveEE$\emptyset$CA}
  \SetKwFunction{detcycle}{DetermineCycle}
  \SetKwData{minee}{MinEE}
  \SetKwData{found}{found\_solution}
  \SetKwData{act}{ActEE}

  \BlankLine
  \eIf{$P = \emptyset$}{\Return{$E \cup \solveconnected(G, \wf)$}\;}{
    $h \leftarrow$ a hint in~$P$\;
    $v_A \leftarrow$ initial vertex of~$h$\;
    $C_A \leftarrow$ connected component of~$G$ corresponding to~$v_A$\;
    $v_\Omega \leftarrow$ terminal vertex of~$h$\;
    $C_\Omega \leftarrow$ connected component of~$G$ corresponding to~$v_\Omega$\;
    \minee $\leftarrow \emptyset$\;
    \found $\leftarrow \false$\;
    \For{$(u, v) \in I^+_G \times I^-_G$ such that $u \in C_A \wedge v \in C_\Omega$ or vice versa}{
      $p \leftarrow \minpath(G, \wf, h, u, v)$\;
      \act $\leftarrow$ \solve{$G + p, \wf, P \setminus \{h\}, E \cup p$}\;
      \If{$\wf(\text{\minee}) > \wf(\text{\act}) \vee \text{\found} = \false$}{
        \found $\leftarrow \true$\;
        \minee $\leftarrow$ \act\;
      }
    }
    \Return{\minee\;}
  }
  
  \SetAlgoRefName{Solve\-EE$\emptyset$CA}
  \caption{Solving \pWMEECCLAs{}.}
  \label{alg:kcalgwmeecla}
\end{algorithm}
\begin{proof}[Proof of \autoref{the:kcalgwmeea}]
  The theorem is mainly based on \autoref{alg:kcalgwmeecla}: Given an instance of \pWMEECAs{} we compute an equivalent instance of \pWMEECCLAs{} using the reduction in \autoref{obs:redwmeeatocla} that uses \autoref{alg:findcycle2}. Then, we apply \autoref{alg:kcalgwmeecla} solving the instance of \pWMEECLAs{}. We first look at the correctness of \autoref{alg:kcalgwmeecla} and then analyze the overall running time.

  Consider the return value~$E'$ of \autoref{alg:kcalgwmeecla} when called with an initially empty arc set~$E$ and an instance of \pWMEECLAs{} consisting of the multigraph~$G$, the weight function~$\wf$, and minimal connecting advice~$P$. For every hint in~$P$ there is realization in~$E'$, that is,~$E'$ connects all connected components of~$G$. Because of the call to~$\solveconnected$ the set~$E'$ also makes every vertex in~$G$ balanced. Hence~$E'$ is an \EE{} for~$G$ that heeds~$P$. Also,~$E'$ is of minimum weight among all \EE s for~$G$ that heed the advice~$P$, because of the weight-minimality of~$\solveconnected$ and because, by \autoref{obs:minpathinee}, we may assume that in a minimum-weight \EE{} all path hints~$h$ are realized by~$\minpath(G, \wf, h, u, v)$ for appropriate vertices~$u, v$ in the components of~$G$ corresponding to the initial and terminal vertices of~$h$.

Concerning the running time of the overall procedure, we have to preprocess the input instance using \autoref{trans:sb} and \autoref{trans:spp} (we have made this preprocessing implicit at the start of the section). By Lemmas~\ref{lem:smallbalancepp} and~\ref{lem:spp} this takes~$\bigO(n^3 + n^2m)$~time. Next, the given instance of \pWMEECAs{} has to be converted to an instance of \pWMEECCLAs{}. By \autoref{lem:findcycle} this is possible in~$\bigO(|P|n^3)$~time. Finally, we apply \autoref{alg:kcalgwmeecla}: Obviously its recursion depth is at most~$|P|$. Because of~$b \geq |I^+_G| = |I^-_G|$, every call of \autoref{alg:kcalgwmeecla} yields at most~$b^2$ recursive calls. This means the sum of all calls is~$b^{2|P|}$. The running-time of one call is dominated by either the computation of~$b^2$ $\minpath$-instances which takes~$\bigO(b^2n^2)$~time (\autoref{lem:minpath}) or the computation of $\solveconnected$ which takes~$\bigO(n^3\log(n))$~time~\cite{DMNW10}. Thus, \autoref{alg:kcalgwmeecla} can be computed in \[\bigO(b^{2|P|}(b^2n^2 + n^3\log(n))) = \bigO(2^{2|P|\log(b)}n^2(b^2 + n\log(n))) \text{ time.}\]
Because of~\autoref{obs:hintsinadvice},~$|P| \leq c$ and thus we can derive that the running-time bound of the overall procedure is in~
\begin{align*}
          & \bigO(2^{2c\log(b)}n^2(b^2+ n\log(n)) + cn^3 + n^2m) \\
\subseteq {} & \bigO(4^{c\log(b)}n^2(b^2+ n\log(n)) + n^2m)\text{.} \tag*{\qedhere}
\end{align*}
\end{proof}

\begin{corollary}\label{cor:wmeefptbc}
  \pWMEE{} is solvable in~\[\bigO(4^{c\log(bc^2)}n^2(b^2 + n\log(n)) + n^2m) \text{ time.}\]
\end{corollary}
\begin{proof}
  By \autoref{lem:redwmeetowmeea} there is a Turing reduction from \pWMEEs{} to \pWMEECAs{} with running time of~$\bigO(16^{c\log(c)}(c + n + m))$ and at most~$16^{c\log(c)}$ oracle calls. Replacing the oracle with the algorithm for \pWMEECAs{} given in \autoref{the:kcalgwmeea} we obtain an algorithm for \pWMEEs{} with~$\bigO(4^{c\log(bc^2)}n^2(b^2 + n\log(n)) + n^2m)$~running time: The algorithm first preprocesses the input using \autoref{trans:sb} and \autoref{trans:spp}, guesses the advice and then, instead of invoking the oracle, reduces the resulting instance of \pWMEECAs{} to an instance of \pWMEECCLAs{}. This instance is then solved using \autoref{alg:kcalgwmeecla}.
\end{proof}


\section{From Eulerian Extension to Matching and back}\label{sec:matching}

The observations in \autoref{sec:trails} suggest the following intuition for making multigraphs Eulerian: To balance every vertex in the given multigraph, we have to add paths from vertices with lower outdegree to vertices with lower indegree. This implies that we have to match these vertices such that adding paths between them leads to a minimum-size \EE. In this section we prove that this intuition is correct and restate \pWMEEs{} as the newly introduced \pCBM{}~(\pCBMs{}).

In previous work by \citet{DMNW10} a similar approach that involves matchings yields polynomial-time algorithms for some restricted \EE{} problems. Of course polynomial-time solvability would be very surprising for \pWMEEs{} because this problem is \NPh ; and we will see that the corresponding matching problem~\pCBMs{} indeed is also \NPh . However, we deem the matching representation to be more accessible in terms of fixed-parameter complexity. In this regard, we show that \pCBMs{} is fixed-parameter tractable on restricted input graphs for a parameter that translates over to the number of components in \pWMEEs{}. Using this we make partial progress to answering the question whether \pWMEEs{} is fixed-parameter tractable with the parameter number of connected components by showing that it indeed is fixed-parameter tractable in a restricted form. 
We also gather a polynomial-size problem kernel for \pWMEECAs{} as a simple corollary using the matching formulation.

In \autoref{sec:cbm} we introduce \pCBMs{}, show that it is \NPh{}, and derive that it is fixed-parameter tractable on special input graphs. In \autoref{sec:releecbm} we investigate the relationship between \pWMEEs{} and \pCBMs{}, and show that they are parameterized equivalent. Using this equivalence, we derive fixed-parameter tractability results for \pWMEEs{} as simple corollaries.

\subsection{\pCBM{}}\label{sec:cbm}

In this section we introduce \pCBM{} (\pCBMs{})---a variant of minimum-weight perfect bipartite matching. We show that this problem is \NPh{} and fixed-parameter tractable on a restricted graph class.

\begin{definition}
  Let $G$ be a bipartite graph,\footnote{Note that~$G$ is undirected.} let~$M$ be a matching of the vertices in~$G$ and let~$P$ be a vertex-partition with the cells~$C_1, \ldots, C_c$. We call an unordered pair~$\{i,j\}$ of integers~$1 \leq i < j \leq c$ a \emph{join} and a set~$J$ a \emph{join set} with respect to~$G$ and~$P$ if~$J \subseteq \{\{i, j\} : 1 \leq i < j \leq c\}$.  
We say that a join~$\{i, j\} \in J$ is \emph{satisfied} by the matching~$M$ of~$G$ if there is at least one edge~$e \in M$ with~$e \cap C_i \neq \emptyset$ and~$e \cap C_j \neq \emptyset$. We say that a matching~$M$ of~$G$ is \emph{$J$-conjoining} with respect to a join set~$J$ if all joins in~$J$ are satisfied by~$M$. If the join set is clear from the context, we simply say that~$M$ is conjoining.
\end{definition}
\decprob{\pCBM{} (\pCBMs{})}{A bipartite graph $G=(V_1 \uplus V_2, E)$ with a weight function~$\wf:E \rightarrow [0, \wf_{max}]\cup \{\infty\}$, a partition~$P = \{C_1, \ldots, C_k\}$ of the vertices in~$G$ and a join set~$J$.}{Is there a matching~$M$ of the vertices of~$G$ such that~$M$ is perfect,~$M$ is conjoining and~$M$ has weight at most~$\wf_{max}$?}
\begin{figure}
  \begin{center}
    \includegraphics{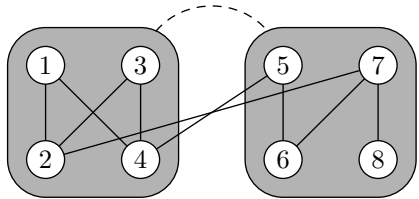}
    \caption{An instance of \pCBMs{} comprising a bipartite graph with the vertices~$1$ through~$8$ and the solid edges, a vertex partition represented by the gray objects, and a join set consisting of a single join that is represented by the dashed line. The weight function is ignored here. The shown instance has a perfect matching, for example~$\{\{1, 2\}, \{3, 4\}, \{5, 6\}, \{7, 8\}\}$. However, it does not have a perfect and conjoining matching: The vertex~$8$ has to be matched to~$7$ in any perfect matching. Thus, the vertices~$2$ and~$7$ cannot be matched. Since~$7$ is already matched, the vertex~$6$ has to be matched to~$5$. This means that the vertices~$4$ and~$5$ cannot be matched. Thus, no edge that satisfies the single join present can be contained in a perfect matching.}
    \label{fig:cbmexample}
  \end{center}
\end{figure}%
For an example of an instance of \pCBMs{}, see \autoref{fig:cbmexample}.
\begin{example}
  \pCBMs{} models a variant of the assignment problem with additional constraints. In this variant, an assignment of workers to tasks is sought such that each worker is busy and each task is being processed. Furthermore, every worker must be qualified for its assigned task. Both the workers and the tasks are grouped and the additional constraints are of the form ``At least one worker from group A must be assigned a task in group B''. An assignment that satisfies such additional constraints may be favorable in the following scenario. 

A company wants to create working groups, each working on a distinct project consisting of multiple tasks. However, every working group shall have a very creative member, a very social and a very methodical member. Here, we assume that extreme creativity, sociality and methodicality are mutually exclusive. 

This scenario can be modeled in \pCBMs{} by defining a bipartite graph that has a vertex for every worker and task, and that has an edge between a worker and a task, if the worker is qualified for the task. The additional constraints can be modeled by first partitioning the tasks into the projects~$C_1, \ldots, C_i$ and partitioning the workers into the creative ones~$C_{i + 1}$, the social ones~$C_{i + 2}$ and the methodical ones~$C_{i+3}$. Then, creating a join set~$\{\{j, i+ 1\}, \{j, i + 2\}, \{j, i + 3\}: 1 \leq j \leq i\}$ ensures that every working group is assigned at least one creative, social, and methodical member. 

The edge weights can be ignored in our scenario. However, as we will see in the forthcoming section, the problem of \pCBMs{} is \NPh{} even in the unweighted case.

\end{example}

\subsubsection{\NPHs}

We reduce from the well-known \ptSATs{} problem~\cite{Kar72}. For this, we briefly recapitulate some related definitions.
\begin{definition}
Consider the boolean variables~$X = \{x_1, \ldots, x_n\}$. \emph{Positive literals} over~$X$ are~$x_i$ and \emph{negative literals} are~$\neg x_i$ with~$x_i \in X$. A \emph{boolean formula~$\phi$ in conjunctive normal form} over the variables~$X$ is of the form~$\bigwedge_{i = 1}^kc_i$, where~$c_i = l_{i_1} \vee \ldots \vee l_{i_{j_i}}$. Here~$l_i$,~$1 \leq i \leq 2n$ are literals over~$X$. The subformulas~$c_i$,~$1\leq i \leq k$, are called \emph{clauses}. If it holds that~$j_1 = \ldots = j_k = d$, then we say that~$\phi$ is in \emph{$d$-conjunctive normal form}. A \emph{truth assignment}~$\nu$ for the variables~$X$ is a function~$\nu: X \rightarrow \{\true, \false\}$. A truth assignment is said to be \emph{satisfying} for a boolean formula~$\phi$ if~$\phi$ evaluates to $\true$ when substituting~$\nu(x_i)$ for every variable~$x_i$ occuring in~$\phi$.
\end{definition}

 In \ptSATs{}, a boolean formula~$\phi$ in 3-conjunctive normal form is given and it is asked whether there is a truth-assignment of the variables in~$\phi$ that satisfies~$\phi$. We use the fact that, in \pCBMs{}, connected components that form cycles have exactly two perfect matchings because every cycle in a bipartite graph has even length. Thus, we model variables as cyclic connected components and the two possible matchings will correspond to the two possible truth values for the variables. Clauses will be modeled by cells in the input partition and a join that forces one of the corresponding variable-cycles into one of the two possible matchings in order to satisfy the clause. 

In the following, we regard clauses of boolean formulas in 3-conjunctive normal form over the variables~$X$ as subsets of~$X \times \{+, -\}$ where~$(x_i, +)$~($(x_i, -)$) in the clause~$c_j$ implies that~$x_i$ is in the clause~$c_j$ as a positive (negative) literal.

First, we give an intuitive description of the reduction, we then go into the details. After that, we give an example and prove the correctness of the reduction.

\paragraph{Intuitive Description.} Let~$\phi$ be a boolean formula in 3-conjunctive normal form with~$n$ variables and~$m$ clauses. For every variable~$x_i$ we introduce a cycle consisting of~$4m$~vertices (vertex set~$V_i$ and edge set~$E_i$ in the below construction). For every such cycle, we fix an ordering of the edges~$e_i^1, \ldots, e_i^{4m}$ according to the order in which they are traversed by the cycle. In a perfect matching of the cycle either all edges~$e_i^k$ with odd~$k$ are matching edges or all edges~$e_i^k$ with even~$k$ are matching edges. These two matchings will correspond to assigning~$x_i$ the value $\false$ or $\true$, respectively.

 Next, for every clause~$c_j$ we define a cell~$C_j$ in order to derive a partition of the vertices in the cycles. For every positive literal~$(x_i,+)$ contained in~$c_j$, we choose an edge~$e^k_i$ such that~$k$ is even and such that its vertices have not been assigned to a cell yet, and put both endpoints of~$e^k_i$ into~$C_j$. Analogously, for every negative literal~$(x_i, -) \in c_i$ we choose an edge~$e^k_i$ such that~$k$ is odd and such that its endpoints have not been assigned yet, and put them into~$C_j$. Finally, all vertices that have not been assigned to a cell yet, are added to the cell~$C_0$ and for every cell~$C_i, i \geq 1$ we add the join~$\{0, i\}$ to the designated join set.
\begin{construction}\label{cons:red3sattocbm}
  Let~$\phi$ be a boolean formula in 3-conjunctive normal form with the variables~$X := \{x_1, \ldots, x_n\}$ and the clauses~$c_1, \ldots, c_m \subseteq X \times \{+, -\}$.

  For every variable~$x_i$, introduce a cycle with~$4m$ edges consisting of the vertex set~$V_i := \{v^j_i: 1 \leq j \leq 4m\}$ and the edge set~\[E_i := \{e^k_i := \{v^k_i, v^{k + 1}_i\} \subseteq V_i\} \cup \{e^{4m}_i := \{v^1_i, v^{4m}_i\}\}\text{.}\]
Define the graph~$G := (\bigcup_{i = 1}^n V_i, \bigcup_{i = 1}^n E_i)$, define the weight function~$\wf$ by~$\wf(e) := 0, e \in E_i$ for any~$1 \leq i \leq n$, and define~$w_{max} := 1$.

Inductively define the vertex partition~$P_m$ of~$V(G)$ and the join set~$J_m$ as follows: Let~$J_0 = \emptyset$ and let~$P_0 := \emptyset$. For every clause~$c_j$ introduce the cell
\begin{align*}
  C_j &:= \{v^{4j-1}_i : (x_i, +) \in c_j \vee (x_i, -) \in c_j\} \cup {} \\
        &\phantom{{}:={}} \{v^{4j-2}_i : (x_i, +) \in c_j\} \cup {} \\
        &\phantom{{}:={}} \{v^{4j}_i : (x_i, -) \in c_j \}\text{.}
\end{align*}
Define~$P_i := P_{i-1} \cup \{C_j\}$ and~$J_i := J_{i- 1} \cup \{\{0, j\}\}$.

Finally, define~$C_0 := V(G) \setminus ( \bigcup_{j=1}^m C_j)$. The graph~$G$, the weight function~$\wf$, the vertex partition~$P_m \cup \{C_0\}$ and the join set~$J_m$ constitute an instance of \pCBMs{}.
  
\end{construction}
\begin{figure}
  \begin{center}
    \includegraphics{cbmexample.1}
    \caption{Example of \autoref{cons:red3sattocbm} explained in \autoref{ex:red3sattocbm}.}
    \label{fig:red3sattocbm}
  \end{center}
\end{figure}
\begin{example}\label{ex:red3sattocbm}
  \autoref{fig:red3sattocbm} shows an instance of \pCBMs{} produced from the formula~$\phi := (\neg x_1 \vee x_2) \wedge (\neg x_1 \vee \neg x_2)$ by \autoref{cons:red3sattocbm}. For simplicity, we chose a formula in 2-conjunctive normal form. The instance comprises the graph~$G$ that consists of two directed cycles (solid edges and dotted edges, respectively), three cells~$C_0, C_1, C_2$ forming a partition of~$V(G)$ (shaded in gray), and a join set with two joins represented by the dashed lines.

  \autoref{cons:red3sattocbm} introduces the solid-edge cycle for variable~$x_1$ and the dotted-edge cycle for variable~$x_2$. The cycle corresponding to~$x_i$ has exactly the two perfect matchings
\begin{align*}
  M^{\true}_i & := \{\{v^k_i, v^{k + 1}_i\} : k \text{ odd}\} \text{ and}\\
  M^{\false}_i & := \{\{v^k_i, v^{k + 1}_i\} : k \text{ even}\} \cup \{\{v^1_i, v^{8}_i\}\}\text{.}
\end{align*}
The cell~$C_1$ models the clause~$\neg x_1 \vee x_2$ and the vertices are chosen such that only edges of~$M^{\false}_1$ and edges of~$M^{\true}_2$ connect the cells~$C_0$ and~$C_1$. Analogously, only edges of~$M^{\false}_1$ and edges of~$M^{\false}_2$ connect the cells~$C_0$ and~$C_2$.

  There is a correspondence between the clauses a variable~$x_i$ satisfies using a particular truth assignment and the joins that are satisfied by matching the cycle that corresponds to~$x_i$ using one of the two available matchings. For example, the variable~$x_1$ satisfies both clauses in~$\phi$ when assigned $\false$ and no clause when assigned $\true$. Accordingly, the matching~$M^{\false}_1$ satisfies both the joins~$\{0, 1\}$, and~$\{0, 2\}$ and the matching~$M^{\true}_1$ satisfies no join. This holds true analogously for~$x_2$ and thus finding a perfect conjoining matching in~$G$ is equivalent to satisfying~$\phi$.
\end{example}
\begin{lemma}
  \pCBMs{} is \NPh{}, even in the unweighted case, even when for every cell~$C_i$ in the given vertex-partition of the input graph~$G=(V \uplus W, E)$ it holds that~$|C_i \cap V| = |C_i \cap W|$ and even when~$G$ has maximum degree two.
\end{lemma}
\begin{proof}
  We prove that \autoref{cons:red3sattocbm} is a polynomial-time many-one reduction from \ptSATs{} to \pCBMs{}. Notice that in instances created by \autoref{cons:red3sattocbm} any matching has weight lower than~$\wf_{max}$ and, thus, the soundness of the reduction implies that \pCBMs{} is hard even without the additional weight constraint. Also, since the cells in the instances of \pCBMs{} are disjoint unions of edges, every cell in the partition~$P_m$ contains the same number of vertices from each cell of the graph bipartition. 

  Concerning \autoref{cons:red3sattocbm}, it is easy to check that it is polynomial-time computable. For the correctness we first need the following definition: For every variable~$x_i \in X$ let
\begin{align*}
  M^{\true}_i & := \{e^k_i  \in E_i: k \text{ odd}\} \text{ and} \\
  M^{\false}_i & := E_i \setminus M^{\true}_i = \{e^k_i \in E_i : k \text{ even}\}\text{.}
\end{align*}
Observe that all perfect matchings in~$G$ are of the form~$\bigcup_{i=1}^nM_i^{\nu(i)}$, where~$\nu$ is some function~$\{1, \ldots, n\} \rightarrow \{\true, \false\}$. We show that the matching~$\bigcup_{i=1}^nM_i^{\nu(i)}$ is a conjoining matching for~$G$ with respect to the join set~$J_m$ if and only if the truth assignment that assigns each~$x_i \in X$ the value~$\nu(i)$ is a satisfying truth assignment for~$\phi$. For this, it suffices to show that for every variable~$x_i \in X$ it holds that
\begin{align}
  \{j : (x_i, +) \in c_j\} &= \{j : M^{\true}_i \text{ satisfies the join } \{0, j\}\}\text{, and} \label{eq:matchvalequiv1}\\
  \{j : (x_i, -) \in c_j\} &= \{j : M^{\false}_i \text{ satisfies the join } \{0, j\}\}\text{.} \label{eq:matchvalequiv2}
\end{align}
We only show that \autoref{eq:matchvalequiv1} holds; \autoref{eq:matchvalequiv2} can be proven analogously. Assume that~$(x_i, +) \in c_j$. By \autoref{cons:red3sattocbm}~$v^{4j-2}_i \in C_j, v^{4j-3}_i \in C_0$ and thus, since \[\{v^{4j-2}_i, v^{4j-3}_i \} = e^{4j-3} \in M^{\true}_i\text{,}\] the matching~$M^{\true}_i$ satisfies the join~$\{0, j\}$. Now assume that~$(x_i, +) \notin c_j$, that is, either 
\begin{inparaenum}
\item both~$(x_i, \pm) \notin c_j$ \label{enu:mnph1} or
\item $(x_i, -) \in c_j$. \label{enu:mnph2}
\end{inparaenum}
In case~\enuref{enu:mnph1} we have that~$V_i$ and~$C_j$ are disjoint and, thus, no matching in~$G[V_i]$ can satisfy the join~$\{0, j\}$. In case~\enuref{enu:mnph2} the only edges in~$E_i$ that can satisfy the join~$\{0, j\}$ are~$e^{4j - 2}_i$ and~$e^{4j}_i$. However, both these edges are not in~$M^{\true}_i$ and, thus, this matching cannot satisfy the join~$\{0, j\}$.
\end{proof}
\begin{observation}\label{obs:cbminnpwp}
  \pCBMs{} is contained in \claNP{} and in \claW{P} when parameterized by the size of the join set.
\end{observation}
\begin{proof}
  Observe that a minimal matching~$M$ that satisfies all joins is a certificate for a yes-instance. Note that~$M$ not necessarily has to be perfect. A minimum-weight perfect conjoining matching~$M' \supseteq M$, if it exists, can then be found in polynomial time by removing the incident vertices of edges in~$M$ from the graph and computing a minimum-weight perfect matching of the remaining vertices. Finding this matching is possible in~$\bigO(mn^2)$~time~\cite{EK72} and it follows that \pCBMs{} is in~\claNP. Also, generating all minimal matchings that satisfy all joins can be done using a polynomial-time Turing machine using at most~$\bigO(c\log(m))$~nondeterministic steps, where~$c$ is the size of the join set: For every join, simply guess an edge that satisfies it. Hence, \pCBMs{} is in~\claW{P}.
\end{proof}
Now we can deduce the following theorem:
\begin{theorem}\label{the:cbmnph}
  \pCBM{} is \NPc .
\end{theorem}

\subsubsection{Tractability on Restricted Graphs}

In this section we use data reduction rules to show that \pCBMs{} is fixed-parameter tractable on some restricted classes of input graphs. In particular, we prove that \pCBMs{} is linear-time decidable on forests (\autoref{cor:cbmlintimeonforests}) and the following theorem:
\begin{theorem}\label{the:cbmmaxdeg2tractable}
  \pCBM{} is solvable in~$\bigO(2^{c(c+1)}n + n^3)$~time, where~$c$ is the size of the join set and when in the bipartite input graph~$G = (V_1 \uplus V_2, E)$ each vertex in~$V_1$ has maximum degree two.
\end{theorem}
Using this theorem and a reduction from \pWMEE{} to \pCBMs{}, we show that \pWMEE{} is tractable on some restricted instances in \autoref{sec:tractislands}. The tractable instances are the preimages of the degree-restricted instances of \pCBMs{} defined in \autoref{the:cbmmaxdeg2tractable}.

To prove the \autoref{the:cbmmaxdeg2tractable}, we use data reduction rules and an observation about matchings in such bipartite graphs as in \autoref{the:cbmmaxdeg2tractable}. We first give some simple reduction rules and then turn our attention to bipartite graphs with maximum degree two. For these graphs we give a slightly more intricate reduction rule restricting the number of cycles they comprise by some function depending only on the join set size~$c$. These reduced instances are then solved via a search-tree procedure which yields fixed-parameter tractability for \pCBMs{} on graphs with maximum degree two. A further observation about matchings in bipartite graphs where each vertex in one cell of the bipartition has maximum degree two is then used to generalize the tractability result to \autoref{the:cbmmaxdeg2tractable}.

In the following, let~$G = (V_1 \uplus V_2, E)$ be a bipartite graph, let~$\wf:E \rightarrow [0, \wf_{max}]\cup \{\infty\}$ be a weight function, let~$P = \{C_1, \ldots, C_d\}$ be a vertex partition of~$G$ and let~$J$ be a join set with respect to~$G$ and~$P$.

\paragraph{Simple Data Reduction Rules.}

\begin{rrule}\label{rr:deg1}
  If there is an edge~$\{v, w\} \in E$ such that~$\deg(v)=1$, then remove both~$v$ and~$w$ from~$G$, and remove any join~$\{i, j\}$ from~$J$, where~$v\in C_i, w\in C_j$. Decrease~$\wf_{max}$ by~$\wf(\{v, w\})$.
\end{rrule}
\begin{observation}\label{obs:deg1runtime}
  \autoref{rr:deg1} is correct and can be applied exhaustively in~$\bigO(n + m)$~time.
\end{observation}
\begin{proof}
  It is clear that \autoref{rr:deg1} is correct because the sought matching is perfect and thus has to match~$v$ with~$w$. It can be applied in linear time by first listing all vertices with degree one in linear time and then applying the rule in a depth-first manner outgoing from the degree-one vertices.
\end{proof}
\begin{corollary}\label{cor:cbmlintimeonforests}
  \pCBMs{} is linear-time solvable on forests.
\end{corollary}
\begin{rrule}\label{rr:compincell}
  If there is a connected component~$C$ of~$G$ such that~$C \subseteq C_j$ for some~$1 \leq j \leq c$, then compute a minimum-weight perfect matching~$M$ in~$G[C]$, remove~$C$ from~$G$ and decrease~$\wf_{max}$ by~$\wf(M)$.
\end{rrule}
\begin{observation}\label{obs:compincellruntime}
  \autoref{rr:compincell} is correct and can be applied exhaustively in~$\bigO(mn^2)$~time.
\end{observation}
\begin{proof}
  The correctness of \autoref{rr:compincell} is easy to prove, since for any perfect conjoining matching~$M'$ for~$G$ we can derive a matching of at most the weight~$\wf(M')$ by matching the vertices in~$G[C]$ according to~$M$. Hence we can derive a matching of weight at most~$ \wf(M') - \wf(M)$ in the graph with~$C$ removed. Obtaining a matching for~$G$ from a matching in the graph~$G$ with~$C$ removed is trivial.

  Applying \autoref{rr:compincell} exhaustively can be done by first finding all connected components~$D_1, \ldots, D_k$ that are contained in one cell in linear time and then computing a minimum-weight perfect matching in the graph~$G[\bigcup_{i = 1}^kD_i]$ in~$\bigO(mn^2)$~time~\cite{EK72}. Then, deleting the affected vertices is possible in linear time.
\end{proof}

\paragraph{Reduction Rule for Maximum Degree Two.}Now let~$G = (V_1 \uplus V_2, E)$ be a bipartite graph with maximum degree two of an instance of \pCBMs{} that is preprocessed with \autoref{rr:deg1} and \autoref{rr:compincell}. In this graph, any degree-one vertices have been deleted and thus each vertex has degree two. It follows that~$G$ consists of connected components each of which is a cycle of even length---because~$G$ is bipartite. Thus every connected component has exactly two perfect matchings. To describe a third reduction rule, we need the following definitions:
\begin{definition}
  For every connected component, that is, every cycle~$D$ contained in~$G$, denote by~$M_1(D)$ a minimum-weight perfect matching of~$D$ with respect to~$\wf$ and denote by~$M_2(D) := E(D) \setminus M_1(D)$, that is, the other perfect matching of~$D$. Furthermore, define
\begin{align*}
  \sigma_1(D) &:= \{j \in J : \exists e \in M_1(D): e \text{ satisfies } j\} \text{,} \\
  \sigma_2(D) &:= \{j \in J : \exists e \in M_2(D): e \text{ satisfies } j\}
\end{align*}
and the \emph{signature}~$\sigma(D)$ of~$D$ as~$(\sigma_1(D), \sigma_2(D))$. We say that two signatures~$\sigma(A), \sigma(B)$ are \emph{equal} and we write~$\sigma(A) \equiv \sigma(B)$, if
\begin{align*}
  &(\sigma_1(A) = \sigma_1(B) \wedge \sigma_2(A) = \sigma_2(B)) \vee {}\\
  &(\sigma_1(A) = \sigma_2(B) \wedge \sigma_2(A) = \sigma_1(B)) \text{.}
\end{align*}
\end{definition}
\begin{rrule}\label{rr:signature}
  Let~$S = \{D_{1}, \ldots, D_{j}\}$ be a maximal set of connected components of~$G$ such that~$\sigma(D_{1}) \equiv \ldots \equiv \sigma(D_{j})$ and~$j \geq 2$. Let~$M^*_1 =\bigcup_{k = 1}^jM_1(D_k)$, let~$D_{l} \in S$ such that~$\wf(M_2(D_{l}))-\wf(M_1(D_{l}))$ is minimum and let~$M^\sim_1 = M^*_1 \setminus M_1(D_l)$.
  \begin{lemenum}
  \item If the matching~$M_1^*$ is conjoining for the join set~$\sigma_1(D_{1}) \cup \sigma_2(D_{1})$, then remove each component in~$S$ from~$G$, remove each join in~$\sigma_1(D_{1}) \cup \sigma_2(D_{1})$ from the join set~$J$, and reduce~$\wf_{max}$ by~$\wf(M_1^*)$.
  \item If the matching~$M_1^*$ is not conjoining for the join set~$\sigma_1(D_{1}) \cup \sigma_2(D_{1})$ remove each component in~$S \setminus \{D_{l}\}$ from~$G$, remove any join in~$\sigma_1(D_{1})$ from the join set~$J$, and reduce~$\wf_{max}$ by~$\wf(M_1^\sim)$. 
  \end{lemenum}

  In either case, update the partition~$P$ accordingly.
\end{rrule}
\begin{lemma}
  \autoref{rr:signature} is correct.
\end{lemma}
\begin{proof}
  Let~$G = (V_1 \uplus V_2, E)$ be a graph with maximum degree two, let ~$\wf:E \rightarrow [0, \wf_{max}]\cup \{\infty\}$ be a weight function, let~$P = \{C_1, \ldots, C_c\}$ be a vertex partition of~$G$ and let~$J$ be a join set with respect to~$G$ and~$P$. The objects~$G$,~$\wf$,~$\wf_{max}$,~$P$, and~$J$ constitute an instance~$I$ of \pCBMs{}. Furthermore, let the graph~$G'$, the weight function~$\wf$, the maximum weight~$\wf_{max}'$, the vertex partition~$P'$, and the join set~$J'$ with respect to~$G'$ and~$P'$ constitute the instance~$I'$ that is obtained from~$I$ by applying \autoref{rr:signature} with the set~$S = \{D_{1}, \ldots, D_{j}\}$ as defined there.

  Let~$M$ be a perfect~$J$-conjoining matching for~$G$ with~$\wf(M) \leq \wf_{max}$ and assume that the matching~$M_1^* = \bigcup_{k = 1}^jM_1(D_{k})$ is conjoining for the join set~$\sigma_1(D_{1}) \cup \sigma_2(D_{1})$. Then either~$M_1^* \subseteq M$, or we can obtain another perfect~$J$-conjoining matching with weight at most~$\wf(M)$ that satisfies this property. Without loss of generality assume that~$M_1^* \subseteq M$. Then~$M \setminus M_1^*$ is a perfect~$J'$-conjoining matching for~$G'$ with weight~$\wf(M) - \wf(M_1^*) \leq \wf_{max}'$. 

Now assume that~$M_1^*$ is not conjoining for the join set~$\sigma_1(D_{1}) \cup \sigma_2(D_{1})$. Then either
\begin{asparaenum}
\item $M_1^* \subseteq M$ or \label{enu:sig11}
\item there is an integer~$n$ such that~$M_2(D_{n}) \subseteq M$. \label{enu:sig12}
\end{asparaenum}
We first show that, in case~\enuref{enu:sig12}, we may assume without loss of generality that~$n$ is unique and that~$n = l$ as in \autoref{rr:signature}. Otherwise we can find another perfect~$J$-conjoining matching with weight at most~$\wf(M)$ that satisfies this property: Since~$M_1^*$ is not conjoining for the join set~$\sigma_1(D_{1}) \cup \sigma_2(D_{1})$, it holds that
\begin{align*}
  \begin{split}
    \sigma_1(D_{1}) = \ldots = \sigma_1(D_{j}) \text{,}
  \end{split}
&\text{and}
&
  \begin{split}
    \sigma_2(D_{1}) = \ldots = \sigma_2(D_{j}) \text{,}
  \end{split}
\end{align*}
because all signatures of the components in~$S$ are equal by prerequisite of \autoref{rr:signature}. If~$n$ is not unique, there are~$n, m$ such that~$M_2(D_{{n}}), M_2(D_{{m}}) \subseteq M$. However, by definition~$\wf(M_1(A)) \leq \wf(M_2(A))$ and if we substitute~$M_1(D_{m})$ for~$M_2(D_{{m}})$ in~$M$, the resulting matching has at most the same weight and is still~$J$-conjoining because~$\sigma_2(D_{{n}}) = \sigma_2(D_{{m}})$. Hence we can assume that~$n$ is unique. We can also assume that~$n = l$ because by definition of~$l$~\[\wf(M_2(D_{l}))-\wf(M_1(D_{l})) \leq \wf(M_2(D_{n}))-\wf(M_1(D_{n}))\] and thus we can substitute~$M_1(D_{n})$ for~$M_2(D_{n})$ and~$M_2(D_{l})$ for~$M_1(D_{l})$ in the matching~$M$ to obtain a perfect~$J$-conjoining matching of at most the same weight. Consider the matching~$M_1^\sim= \bigcup_{1 \leq k \leq j, k \neq l}M_1(D_{k})$. Both in case~\enuref{enu:sig11} and in case~\enuref{enu:sig12}, when assuming that~$n =l$ is unique,~$M \setminus M_1^\sim$ is a perfect~$J'$-conjoining matching for~$G'$ of weight~$\wf(M) - \wf(M_1^\sim) \leq \wf_{max}'$.

  We now have that if~$I$ is a yes instance then~$I'$ is a yes instance. For the other way round, assume that~$M'$ is a perfect~$J'$-conjoining matching for~$G'$ of weight~$\wf(M') \leq \wf_{max}'$. Assume that each component in~$S$ of~$G$ has been removed in~$G'$ by \autoref{rr:signature}. Then the matching~$M' \cup M_1^*$ for~$G$ is perfect,~$J$-conjoining and of weight~$\wf(M) + \wf(M_1^*) \leq \wf_{max}$. Now assume only the component~$D_{l}$ of the components in~$S$ is still present in~$G'$. Then, the matching~$M \cup M_1^\sim$ is a perfect~$J$-conjoining matching for~$G$ of weight~$\wf(M) + \wf(M_1^\sim) \leq \wf_{max}$.
\end{proof}
\begin{lemma}\label{lem:signatureruntime}
  \autoref{rr:signature} can be applied exhaustively in~$\bigO(n^3)$~time.
\end{lemma}
\begin{proof}
  To apply \autoref{rr:signature} once, we can first search for a set of components~$S$ as defined there by first finding all connected components in linear time. Then we find out the signature of each connected component. For this, we first compute a minimum-weight perfect matching for every connected component in overall~$\bigO(m)$~time by simply iterating over the edges in each component, alternatingly summing up the edge weights and choosing the lower one of the two values. We annotate every edge with whether it is contained in the minimum-weight matching or not and which join it satisfies, if any, in~$\bigO(m^2)$~time. We then iterate over every edge and add the information saved in the annotation to the signature of the connected component it is contained in. 

Having computed the signatures, we create a map in~$\bigO(n\log(n))$~time that maps every signature present to the list of connected components that have this signature. We then simply iterate over every list present in the map to obtain a maximal list of components that have the same signature or decide that there is no such list with at least two elements. This is possible in~$\bigO(n)$~time. 

The removal of the connected components and joins, the update of~$\wf_{max}$ and the partition~$P$ is then possible in linear time, because the matchings for each component have already been computed and thus the overall running time is~$\bigO(m^2 + n \log n)$. Observe that in graphs with exactly degree two~$m \in \bigO(n)$ and thus we can derive a running time bound in~$\bigO(n^2)$.
  
In any application either no set~$S$ is found and thus the procedure terminates, or at least~4 vertices are deleted---this is the minimum size of a connected component. Hence the procedure can be applied at most~$n$ times and exhaustively applying \autoref{rr:signature} takes~$\bigO(n^3)$~time.
\end{proof}
\begin{observation}
  When \autoref{rr:signature} cannot be applied anymore, the input graph contains at most~$2^{c + 1}$ components, where~$c$ is the size of the join set.
\end{observation}
\begin{proof}
  When there are~$c$ joins in a join set, then there are at most~$2^{c + 1}$ signatures. For each signature, there is at most one connected component when \autoref{rr:signature} is not applicable.
\end{proof}
\begin{lemma}\label{lem:cbmmaxdeg2tractable}
  \pCBMs{} is solvable in~$\bigO(2^{c(c+1)}n + n^3)$~time on graphs with maximum degree two, where~$c$ is the size of the join set.
\end{lemma}
\begin{proof}
  This follows from exhaustively applying \autoref{rr:deg1}, \autoref{rr:compincell}, and \autoref{rr:signature} and then invoking a search tree algorithm. The algorithm chooses one join, branches into choosing any component that contains an edge that satisfies the join, matches the component accordingly and then recurses until every join is satisfied. Since there are at most~$2^{c+1}$ components in the preprocessed graph, every branching-step invokes at most~$2^{c+1}$ recursive calls. The recursion depth is obviously at most~$c$. In every call at most~$\bigO(n)$~time is spent finding components satisfying the chosen join and thus we can derive a running time bound of~$\bigO(2^{c+1})^cn) = \bigO(2^{c(c+1)}n)$ for the search tree algorithm. The preprocessing rules take~$\bigO(n^3)$~time by \autoref{obs:deg1runtime}, \autoref{obs:compincellruntime}, \autoref{lem:signatureruntime}, and by the fact that~$m \in \bigO(n)$ in graphs with degree at most two. Thus the overall running time bound is~$\bigO(2^{c(c+1)}n + n^3)$.
\end{proof}
\paragraph{Perfect Matchings in Graphs with Maximum Degree Two.} Now let~$G = (V_1 \uplus V_2, E)$ be a bipartite graph where each vertex in~$V_1$ has maximum degree two. We show that if~$G$ has a perfect matching, it will be preprocessed by \autoref{rr:deg1} such that each vertex has degree exactly two.
\begin{lemma}\label{lem:perfmatchatmostonecycle}
  If~$G$ has a perfect matching, every connected component of~$G$ contains at most one cycle as subgraph.
\end{lemma}
\begin{proof}
  We show that if~$G$ contains a connected component that contains two cycles~$c_1, c_2$ as subgraphs, then~$G$ does not have a perfect matching. First assume that~$c_1, c_2$ are vertex-disjoint. Then, there is a path~$p$ from a vertex~$v \in V(c_1)$ to a vertex~$w \in V(c_2)$ such that~$p$ is vertex-disjoint from~$c_1$ and~$c_2$ except for~$v, w$. It is clear that both~$v, w \in V_2$ because they have degree three. Consider the vertices~$V_1^{cp}:=(V(c_1) \cup V(p) \cup V(c_2)) \cap V_1$ and the set~$V_2^{cp}:=(V(c_1) \cup V(p) \cup V(c_2)) \cap V_2$. The set~$V_2^{cp}$ is the set of neighbors of vertices in~$V_1^{cp}$, because they have degree two and thus have neighbors only within~$p$,~$c_1$, and~$c_2$. It is~$|V_1^{cp}| = (|E(c_1)| + |E(p)| + |E(c_2)|) / 2$ since neither of these paths and cycles overlap in a vertex in~$V_1$. However, it is~$|V_2^{cp}| = |V_1^{cp}| - 1$ because~$c_1$ and~$p$ overlap in~$v$ and~$c_2$ and~$p$ overlap in~$w$. This is a violation of Hall's condition---recall the definition of Hall's condition in \autoref{the:hallscondition}---and thus~$G$ does not have a perfect matching.

  The case where~$c_1$ and~$c_2$ share vertices can be proven analogously. (Observe that then there is a subpath of~$c_2$ that is vertex-disjoint from~$c_1$ and contains an even number of edges.)
\end{proof}

\begin{proof}[Proof of \autoref{the:cbmmaxdeg2tractable}]
Consider applying \autoref{rr:deg1} to a graph~$G = (V_1 \uplus V_2, E)$ such that each vertex in~$V_1$ has maximum degree two and such that~$G$ has a perfect matching. This has to yield a graph that is a collection of vertex-disjoint cycles because in every connected component there is at most one cycle as subgraph (\autoref{lem:perfmatchatmostonecycle}). Hence, every component consists of a cycle with a collection of pairwise vertex-disjoint paths incident to it. These paths are completely reduced by \autoref{rr:deg1} and all that remains is either the cycle or nothing. Thus, in order to cope with graphs~$G$ as above, we can modify the algorithm from \autoref{lem:cbmmaxdeg2tractable}: If the application of \autoref{rr:deg1} does not yield a graph that is a collection of vertex-disjoint cycles, we can abort the procedure because it cannot yield a perfect matching. This can be checked in linear time and thus, \autoref{the:cbmmaxdeg2tractable} now directly follows. (Notice that the running time bound of \autoref{lem:cbmmaxdeg2tractable} does not increase, since in graphs~$G$ as above that have a perfect matching it also holds that~$m \in \bigO(n)$.)
\end{proof}


\subsection{The Relationship between Eulerian Extension and Matching} \label{sec:releecbm}

In this section we show that \pCBMs{} parameterized by the size of the join set and \pWMEEs{} parameterized by the number of connected components in the input graph are parameterized equivalent. To this end, we first give a reduction from \pWMEECAs{} to \pCBMs{}. This reduction also yields an efficient algorithm for a restricted variant of \pWMEEs{}. Second, we give a reduction from \pCBMs{} to \pWMEEAs{}. The equivalence of \pWMEEs{} and \pCBMs{} then follows from the reductions given in \autoref{lem:redwmeetowmeea} and \autoref{the:redwmeeatowmee} in \autoref{sec:wmeeslasha}.


\subsubsection{Reducing \pWMEECAs{} to \pCBMs{}}
We first reduce \pWMEECAs{} to \pCBMs{}. In order to simplify our reduction, we reduce from \pWMEECCLAs{} instead (see page~\pageref{def:pWMEECA} in \autoref{sec:wmeeslasha}). We know that \pWMEECCLAs{} and \pWMEECAs{} are equivalent from \autoref{obs:redwmeecatoccla}.

As we have observed in \autoref{obs:walkdichotomy} we have to draw paths between unbalanced vertices in order to make them balanced and to ultimately make the input graph Eulerian. These paths also have to connect all components of the input graph. The basic structure of these paths is made explicit by the advice in \pWMEECCLAs{} and thus we do not have to concern ourselves with finding a suitable order of components for these paths. We simply have to realize every hint to connect the graph and then balance all remaining vertices.

\paragraph{Reduction Outline.}
The basic ideas of our reduction are to use vertices of positive balance and negative balance in an instance of \pWMEECCLAs{} as the two cells of the graph bipartition in a designated instance of \pCBMs{}. Edges between vertices in the new instances represent shortest paths between them that consist of allowed extension arcs in the original instance. Every connected component in the original instance is represented by a cell in the vertex partition in the matching instance and hints are basically modeled by joins.

We proceed with an intuitive description of the reduction and then go into the details in \autoref{con:redwmeeatocbm}. The construction is then followed by a correctness proof. 
For the descriptions, we first need the following definition.
\begin{definition}
  Let~$G$ be a directed multigraph with the connected components~$V_1, \ldots, V_c$ and let~$H$ be a cycle-free advice for~$G$. For every hint~$h \in H$ we define~$\conn(h)= \{i, j\}$, where~$C_i, C_j$ are the components corresponding to the initial and terminal vertices of~$h$.
\end{definition}
\paragraph{Intuitive Description.}
First, consider an instance~$I_\pWMEECCLAs{}$ of \pWMEECCLAs{} that consists of the graph~$G$, the weight function~$\wf: V \times V \rightarrow [1, \wf_{max}]\cup \{\infty\}$ and a cycle-free minimal connecting advice~$H$ that contains only hints of length one. We will deal with longer hints later. We create an instance~$I_\pCBMs{}$ of \pCBMs{} by first defining~$B_0 = (I_G^+ \uplus I_G^-, E_0)$ as a bipartite graph. Here, the set~$E_0$ consists of all edges~$\{u, v\}$ such that~$u \in I_G^+$,~$v \in I_G^-$, and~$\wf(u, v) < \infty$. This serves the purpose of modeling the structure of allowed arcs in the matching instance---we come back to this in \autoref{sec:tractislands}. Second, we derive a vertex partition~$\{V_1', \ldots, V_c'\}$ of~$B_0$ by intersecting the connected components of~$G$ with~$(I_G^+ \uplus I_G^-)$. The vertex-partition obviously models the connected components in the input graph, and the need of connecting them according to the advice~$H$ is modeled by an appropriate join-set~$J_0$, defined as~$\{\conn(h): h \in H\}$. Finally, we make sure that matchings also correspond to \EE s weight-wise, by defining the weight function~$\wf'(\{u, v\})$ for every~$u \in I_G^+, v \in I_G^-$ as~$\wf(u, v)$ with~$\wf_{max}' = \wf_{max}$.

By \autoref{obs:shorttrails} we may assume that every hint in~$H$ of length one is realized by a single arc. Since the advice connects all connected components, by the same observation, we may assume that all other trails in a valid \EE{} have length one (\autoref{obs:shorttrails} also holds for the connected graph obtained by adding the realizations of all hints to the input graph). Finally, by \autoref{lem:smallbalancepp}, we may assume that every vertex has at most one incident incoming or outgoing arc in the extension and, hence, we get an intuitive correspondence between the matchings and \EE s. 

\begin{figure}
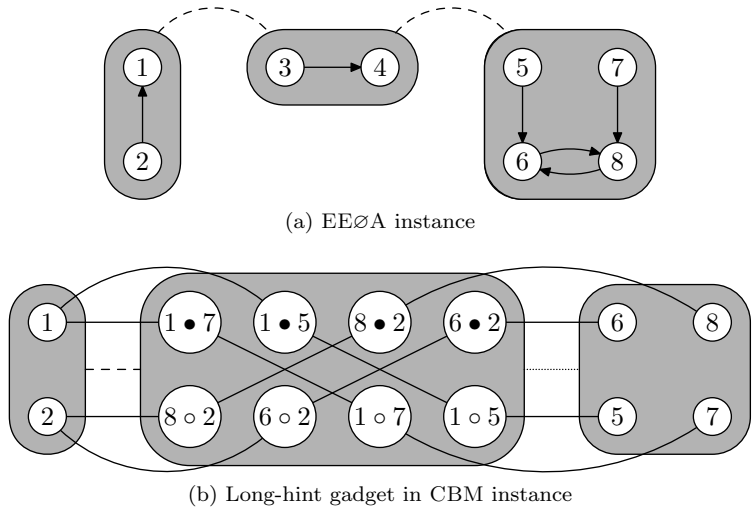

  \begin{center}
    \subfloat[\pWMEECLAs{} instance]{
      \includegraphics{wmeeslasha.1}
      \label{fig:redwmeeatocbma}
    } \\
    \subfloat[Long-hint gadget in \pCBMs{} instance]{
      \includegraphics{wmeecbm.1}
      \label{fig:redwmeeatocbmb}}
    \caption{Example for the long-hint gadget used in \autoref{con:redwmeeatocbm}, explained in the corresponding intuitive description.}
    \label{fig:redwmeeatocbm}
  \end{center}
\end{figure}
To model hints of length at least two, we utilize gadgets similar to the one shown in \autoref{fig:redwmeeatocbm}. On the top, an instance~$I_\pWMEECLAs{}$ is shown, consisting of a graph with three connected components and an advice that contains a single hint~$h$ (dashed lines). Below in \autoref{fig:redwmeeatocbmb} a part of an instance of \pCBMs{} is shown, which comprises the cells that correspond to the initial and terminal vertices of~$h$ and a gadget to model~$h$. The gadget consists of some new vertices which are put into a new cell which is connected by two joins (dashed and dotted lines) to the cells corresponding to the initial and terminal vertices of~$h$.

The gadget comprises two vertices ($u \circ v$ and $ u \bullet v$) for every pair~$(u, v)$ of vertices with one vertex in the component the hint starts and one in the component the hint ends. The vertices~$u \circ v$ and $u \bullet v$ are adjacent and each of these two vertices is connected with one vertex of the pair it represents. The edge~$\{u \bullet v, u\}$ is weighted with the cost it takes to connect~$u, v$ with a path~$p$ such that~$\meta{G}{p}=h$ that is, a path that realizes~$h$. The other edges have weight~$0$. Intuitively these three edges in the gadget represent one concrete realization of~$h$. If~$u \circ v$ and~$u \bullet v$ are matched, this means that this specific path does not occur in a designated \EE{}. However, by adding the vertices of the gadget as cell to the vertex partition and by extending the join set to the gadget, we enforce that there is at least one outgoing edge that is matched. If~$v \bullet u$ is matched with~$v$, then~$v \circ u$ must be matched with~$u$ and vice versa, otherwise the matching could not be perfect. This introduces an edge to the matching that has weight corresponding to a path that realizes~$h$.

\begin{construction}
  \label{con:redwmeeatocbm}
  Let the directed multigraph~$G=(V,A)$, the weight function~$\wf: V \times V \rightarrow [1, \wf_{max}]\cup \{\infty\}$ and the advice~$H$ constitute an instance of \pWMEECCLAs . Let~$V_1, \ldots, V_c$ be the connected components of~$G$.
  
  Let $H^{=1}$ be the set of hints of length one in~$H$ and let~$H^{\geq 2}$ be the set of hints in~$H$ that have length at least two. Define~$J_0$ by the set~$\{\conn(h): h \in H^{=1}\}$. Let~$W^1_0 := I_G^+$,~$W^2_0 := I_G^-$, and let~$B_0 =(W^1_0 \uplus W^2_0, E_0)$ be a bipartite graph where~\[E_0 := \{\{u, v\} : u \in I_G^+ \wedge v \in I_G^- \wedge \wf(u, v) < \infty\} \text{.}\] Define~$V'_i := V_i \cap (I^+_G \cup I^-_G)$,~$1 \leq i \leq c$, and~$\wf'_0(\{u, v\}):=\wf(u, v)$ where~$\{u,v\} \in E, u \in I_G^+$.

  Let $h^{\geq 2}_1, \ldots, h^{\geq 2}_j$ be the hints in~$H^{\geq 2}$. Inductively define~$B_k$,~$V'_{c+k}$,~$\wf'_k$ and~$J_k$,~$1 \leq k \leq j$, as follows: Let~$\conn(h^{\geq 2}_k)=\{o, p\}$. Introduce the vertex sets 
\begin{align*}
  U_1 & := \{v \circ u : v \in I^+_G \cap V_o \wedge u \in I^-_G \cap V_p \wedge \wf(\minpath(G, \wf, h^{\geq 2}_k, v, u)) < \infty\} \cup {} \\
      & \phantom{{}:={}} \{v \circ u : v \in I^-_G \cap V_o \wedge u \in I^+_G \cap V_p \wedge \wf(\minpath(G, \wf, h^{\geq 2}_k, u, v)) < \infty\} \text{,}
\end{align*}
and~$U_2 := \{v \bullet u : v \circ u \in U_1\}$. Introduce the edge sets
\begin{align*}
  E_k^1 &:= \{\{v \circ u, v\} : v \in I^-_G \wedge v \circ u \in U_1\} \text{,} \\
  E_k^2 &:= \{\{v \bullet u, v\} : v \in I^+_G \wedge v \bullet u \in U_2\} \text{, and} \\
  E_k^3 &:= \{\{v \circ u, v \bullet u\}: v \circ u \in U_1 \wedge v \bullet u \in U_2 \} \text{.}
\end{align*}
Set~$E_k := E_k^1 \cup E_k^2 \cup E_k^3$, and set the graph
\begin{align*}
  B_k & := ((W^{1}_{k - 1} \cup U_1) \uplus (W^2_{k - 1} \cup U_2), E_{k-1} \cup E_k) \text{,}
\end{align*}
set~$V'_{c+k} := U_1 \cup U_2$, set~$J_k := J_{k-1} \cup \{\{o, c+k\}, \{p, c+k\}\}$ and the weight-function as follows:  
\[\wf'_k(\{u, v\}) := \begin{cases}
  \wf'_{k-1}(\{u, v\}), & \{u, v\} \in E_{k-1} \\
  0, & \{u, v\} \in E_k^1 \cup E_k^3 \\
  \wf(\minpath(G, \wf, h^{\geq 2}_k, v, w)), & \{u, v\} = \{v \bullet w, v\} \in E_k^2
\end{cases}
\]
Then the graph~$B_j$, the weight function~$\wf'_j$, the vertex partition~$P := \{V_1, ..., V_{c+j}\}$ and the join set~$C_j$ constitute an instance of \pCBMs .
\end{construction}

For the remainder of this section, let the directed multigraph~$G=(V,A)$, the weight function~$\wf: V \times V \rightarrow [1, \wf_{max}]\cup \{\infty\}$ and the cycle-free minimal connecting advice~$H$ constitute an instance of \pWMEECCLAs{} and let the bipartite graph~$B := B_j$, the weight function~$\wf' := \wf'_j$ with the maximum weight~$\wf_{max}$, the vertex partition~$P$ and the join set~$J := J_j$ as in \autoref{con:redwmeeatocbm} constitute an instance of \pCBMs{}.

\begin{lemma}
  \label{lem:redwmeeatocbmtransfer}
  Let~$E$ be an \EE{} for~$G$ that heeds the advice~$H$. Then there is a perfect conjoining matching~$M$ for~$B$ with $\wf'(M) \leq \wf(E)$.
\end{lemma}
\begin{proof}
  We construct the matching successively by first looking at every long-path gadget in~$B$ and then matching the remaining vertices.

  Consider the cell~$V'_{c + k} \in P$ for $k>0$. There are two joins~$\{c+k, o\}$ and~$\{c+k, p\}$ in~$J$. Thus, there is a path hint~$h$ from~$V_o$ to~$V_p$ in~$H$. This means that, there is a path~$s$ in~$E$ that starts in a vertex~$v \in V$ in the component~$V_o$ and ends in a vertex~$u \in V$ in~$V_p$. The weight~$\wf(s)$ is at least~$\wf(\minpath(G,\wf,h,u,v))$ (\autoref{obs:minpathinee}). Thus we may match~$u \bullet v$ with~$v$, $u \circ v$ with~$u$ (this costs weight~$\wf(\minpath(G,\wf,h,u,v))$) and every other pair~$w \bullet x$ and~$w \circ x$ in~$V'_{c + k}$ with each other (this costs weight~$0$). Matching like this, we obtain a matching for the long-hint gadget of~$h$ that fulfills its two joins and is perfect. The weight of the matching is at most the realization of~$h$ in~$E$. 

Because of \SPP{} (\autoref{trans:spp}) and \autoref{obs:metaedgedisjoint} we may assume that there is a set of paths in~$E$ that is edge-disjoint and realizes all hints in~$H$ (otherwise we may obtain an \EE{} of at most the same weight that has this property). Because of this, we may find a matching~$M^{\geq 2}$ for~$B$ that satisfies the joins of every long-hint gadget and is perfect with respect to the vertex set of each long-hint gadget---as in the previous paragraph, iterated for every gadget. Furthermore,~$\wf'(M^{\geq 2})$ is lower than the weight of all paths in~$E$ that realize hints of length at least two in~$H$. 

Now it is easy to extend~$M^{\geq 2}$ to a conjoining matching~$M^{\geq 1}$ for~$B$ and~$J$ just by adding matching edges between vertices that realize hints of length one in~$E$. We may assume by \autoref{obs:shorttrails} that each hint of length one is realized by a single arc in~$E$. The weight of matching edges is exactly the cost of the direct arc between the corresponding vertices. Because of this, we maintain that $\wf'(M^{\geq 1})$~is at most the weight of all paths in~$E$ that realize hints. 

Finally, we have to extend~$M^{\geq 1}$ to a perfect matching~$M$ by matching the remaining non-gadget vertices. We can do this by looking at paths in~$E$ that start and end in the vertices in~$G$, corresponding to still unmatched vertices in~$B$. A set of such paths must exist, because each such vertex has at least one incident arc in~$E$ and because, by \autoref{obs:walkdichotomy}, maximal-length open trails in \EE s start and end in unbalanced vertices. The edges between initial and terminal vertices of those paths in~$B$ have at most the weight of such a path (because of \SPP{} and because they have weight corresponding to the direct arc). Thus, we can add those edges to~$M^{\geq 1}$, obtaining an edge set~$M$. This set is a matching for~$B$ that is perfect, conjoining and~$\wf'(M) \leq \wf(E)$.
\end{proof}
\begin{lemma}
  \label{lem:redwmeeatocbmbacktransfer}
  Let $M$ be a perfect conjoining matching for~$B$. We can construct 
an \EE{}~$E$ for~$G$ that heeds the advice~$H$ such that~$\wf(E) = \wf'(M)$.
\end{lemma}
\begin{proof}
  We simply look at every matching edge that has non-zero weight and add a corresponding path to a designated \EE{}~$E$ of~$G$: For non-gadget matching edges (edges that match vertices in~$V'_1, \ldots, V'_c$) the corresponding path is the direct arc between the two vertices in~$G$. For edges that match a vertex~$v$ in a cell~$V'_o$,~$1 \leq o \leq c$ and a vertex~$u \bullet v \in V'_{c + k}$,~$1 \leq k \leq j$, where~$u \in V'_p, 1 \leq p \leq c$, the corresponding path is~$\minpath(G, \wf, h_k, u, v)$. Here, $h_k$~is the path in~$H$ that lead to the introduction of~$V'_{c+k}$ in \autoref{con:redwmeeatocbm}.


  We immediately see that~$\wf(E) = \wf'(M)$. Also, it is clear that every hint of length one in~$H$ is realized in~$E$ because every hint~$h^1$ of length one leads to the pair~$\conn(p^1)$ in~$J$. Hints~$p^{\geq2}$ of length two are also realized, because every such path leads to a cell~$V'_{c + k}$,~$1 \leq k \leq j$ and also leads to the corresponding joins~$\{o, c+k\}$ and~$\{p, c+k\}$ in~$J$, where~$\{o, p\} = \conn(h^{\geq 2})$. Thus,~$E$ heeds the advice~$H$. Since~$M$ is a perfect matching, every unbalanced vertex in~$G$ is the initial or terminal vertex of exactly one path added to~$E$ in the above paragraph. 
By \autoref{lem:smallbalancepp} we may assume that this suffices to make every vertex in~$G + E$ balanced. Also, $G + E$~is connected, because~$E$ heeds the advice~$H$.
\end{proof}
\begin{lemma}
  \label{con:redwmeeatocbmruntime}
  \autoref{con:redwmeeatocbm} is computable in~$\bigO(|H|n^4 + m)$~time.
\end{lemma}
\begin{proof}
  Computing~$B_0$ takes~$\bigO(n^2)$~time. To compute~$J_0$ one needs~$\bigO(|H|)$~time by iterating over every path in~$H$. Computing the initial partition~$\{V'_1, \ldots,V'_c\}$ takes~$\bigO(n + m)$~time and the initial weight function~$\wf'_0$ can also be computed within this time. Hence, creating the initial instance is possible in~$\bigO(n^2 + m)$~time.
 
  Regarding adding the gadget for one path in~$H$, to compute the sets~$U_1$ and~$U_2$,~$\bigO(n^4)$~time is suffices, because~$n^2$ instances of~$\minpath$ have to be computed, each taking~$\bigO(n^2)$~time (\autoref{lem:minpath}). There are only three edges in the gadget for every vertex~$v \in U_1$, thus computing the edge sets does not increase the running time bound. For the weight function we can reuse the values of~$\minpath$ computed for every pair of vertices~$v \in I_G^+, u \in I_G^-$ and thus we can conclude an overall running time bound of~$\bigO(|H|n^4 + m)$.
\end{proof}
Now the following theorem follows:
\begin{theorem}\label{the:redwmeecatocbm}
  \pWMEECA{} is polynomial-time many-one reducible to \pCBM . The corresponding reduction function is a parameterized reduction with respect to the parameters number of components in the graph of \pWMEECAs{} and join set size in \pCBMs{}.
\end{theorem}
\begin{proof}
  By \autoref{obs:redwmeecatoccla}, there is a polynomial-time many-one reduction from \pWMEECAs{} to \pWMEECCLAs{}. This reduction at most decreases the number of components in the input graph. By \autoref{lem:redwmeeatocbmtransfer} and \autoref{lem:redwmeeatocbmbacktransfer} there is a many-one reduction from \pWMEECCLAs{} to \CBMs{} . Since the construction is polynomial-time computable (\autoref{con:redwmeeatocbmruntime}), since for every hint in the advice there are at most two joins, and since the number of hints is bounded by the number of components in the input graph to \pWMEECAs{} (\autoref{obs:hintsinadvice}) it follows that \autoref{con:redwmeeatocbm} is a parameterized polynomial-time many-one reduction.
\end{proof}
\begin{corollary}
  \pWMEE{} is parameterized Turing reducible to \pCBM{} with respect to the parameters number of components in input graph and join set size.
\end{corollary}
\begin{proof}
  The statement follows from \autoref{lem:redwmeetowmeea} and \autoref{the:redwmeecatocbm}.
\end{proof}

\subsubsection{Islands of Tractability for  \pWMEEs}\label{sec:tractislands}
Using the reduction given in \autoref{con:redwmeeatocbm}, we can gather the fruit of our work in \autoref{sec:cbm} where we showed restricted fixed-parameter tractability of \pCBMs{} with respect to the join set size.
\begin{corollary}\label{cor:wmeetractableforests}
  Let the graph~$G$ and the weight function~$\wf$ constitute an instance~$I_\pWMEEs{}$ of \pWMEEs{}. Let~$c$ be the number of connected components in~$G$. Furthermore, 
\begin{lemenum}
\item let the set~$A_A$ of allowed arcs with respect to~$\wf$ not contain a path or cycle of length at least two,
\item let the underlying graph of the directed graph~$(V, V \times V)\langle A_A \rangle$ be a forest, and
\item let~$G$ contain only vertices with balance between~$-1$ and~$1$.
\end{lemenum}
Then, it is decidable in~$\bigO(16^{c\log(c)}(cn^4 + m))$~time whether~$I_\pWMEEs{}$ is a yes-instance.
\end{corollary}
\begin{proof}
  Observe that such instances are invariant under \autoref{trans:sb} and \autoref{trans:spp}. Thus, we may directly apply the reduction from \pWMEEs{} to \pWMEECAs{} given in \autoref{lem:redwmeetowmeea} that runs in time~$\bigO(16^{c\log(c)}(c + n + m))$. Also, there is no valid advice that contains hints of length two for such graphs. Thus, we can apply \autoref{con:redwmeeatocbm}---running in~$\bigO(cn^4 + m)$~time by \autoref{con:redwmeeatocbmruntime}---to map the instances of \pWMEECAs{} to instances of \pCBMs{} that comprise bipartite graphs that are forests. By \autoref{cor:cbmlintimeonforests}, these instances are solvable in linear time.
\end{proof}
\begin{corollary}\label{cor:wmeetractablecycles}

  Let the graph~$G$ and the weight function~$\wf$ constitute an instance~$I_\pWMEEs{}$ of \pWMEEs{}. Let~$c$ be the number of connected components in~$G$. Furthermore, 
\begin{lemenum}
\item let the set of allowed arcs with respect to~$\wf$ not contain a path or cycle of length at least two,
\item let~$G$ contain only vertices with balance between~$-1$ and~$1$,
\item let every vertex in~$I_G^+$~(every vertex in~$I_G^-$) have only outgoing allowed arcs (incoming allowed arcs),
\item for every connected component~$C$ of~$G$, let either all vertices in~$I_G^+ \cap C$ have at most two incident allowed arcs or let all vertices in~$I_G^- \cap C$ have at most two incident allowed arcs.
\end{lemenum}
Then, it is decidable in~$\bigO(2^{c(c + \log(2c^4))}(n^4 + m))$~time whether~$I_\pWMEEs{}$ is a yes-instance.
\end{corollary}
\begin{proof}
  The proof is analogous to \autoref{cor:wmeetractableforests} by substituting the algorithm we gave in \autoref{the:cbmmaxdeg2tractable} for \autoref{cor:cbmlintimeonforests}. This leads to a running time bound of~$\bigO(16^{c\log(c)}(cn^4+m + 2^{c(c+1)}n)) \subseteq \bigO(2^{c(c + \log(2c^4))}(n^4 + m))$.
\end{proof}

\subsubsection{Reducing \pCBMs{} to \pWMEEAs{}}
To reduce \pCBMs{} to \pWMEEAs{} we first observe that for every instance of \pCBMs{} there is an equivalent instance such that every cell in the input vertex-partition contains equal numbers of vertices from both cells of the graph bipartition. This observation enables us to model cells as connected components and vertices in the bipartite graph as unbalanced vertices in the designated instance of \pWMEEAs{}.

We first need the following auxiliary observations: 
\begin{observation}\label{obs:surplussum}
  Let~$G = (V_1 \uplus V_2, E)$ be a bipartite graph such that~$|V_1| = |V_2|$ and let the set~$P = \{C_1, \ldots, C_k\}$ be a partition of the vertices in~$G$. It holds that \[{\sum_{i: |C_i \cap V_1| > |C_i \cap V_2|}} |C_i \cap V_1| - |C_i \cap V_2| = {\sum_{i: |C_i \cap V_1| < |C_i \cap V_2|}}|C_i \cap V_2| - |C_i \cap V_1| \text{.} \]
\end{observation}
\begin{proof}
  Observe that the equation holds if and only if~$|V_1| = |V_2|$: Without loss of generality we may assume that there are no cells~$C_i$ with~$|C_i \cap V_1| = |C_i \cap V_2|$ because these do contribute summands to the equation. Then we can transpose the equation such that the left-hand side reads as follows \[{\sum_{i: |C_i \cap V_1| > |C_i \cap V_2|}} |C_i \cap V_1| + {\sum_{i: |C_i \cap V_1| < |C_i \cap V_2|}}|C_i \cap V_1|\text{.}\]
This is equal to~$|V_1|$. Analogously, the left-hand side in the transposed formula is equal to~$|V_2|$.
\end{proof}

\begin{lemma}\label{lem:cbmlegalized}
  For every instance of \pCBMs{} there is an equivalent instance comprising the bipartite graph~$G = (V_1 \uplus V_2, E)$, the vertex partition~$P = \{C_1, \ldots, C_{k+1}\}$ and the join set~$J$, such that
\begin{lemenum}
\item for every~$1 \leq i \leq k+ 1$ it holds that~$|V_1 \cap C_i| = |V_2 \cap C_i|$, and\label{enu:cl11}
\item the graph~$(P, \{\{C_i, C_j\}: \{i, j\} \in J\})$ is connected.\label{enu:cl12}
\end{lemenum}
Furthermore, this equivalent instance contains at most one cell more than the original instance.
\end{lemma}
\begin{proof}
  We first prove that there is an equivalent instance corresponding to statement~\enuref{enu:cl11} and then turn to statement~\enuref{enu:cl12}.
  Let the bipartite graph~$G = (V_1 \uplus V_2, E)$, the weight function~$\wf:E \rightarrow [0, \wf_{max}]\cup \{\infty\}$, the vertex partition~$P = \{C_1, \ldots, C_k\}$ and the join set~$J$ constitute an instance~$I_\pCBMs{}$ of \pCBMs{}. First observe that if~$I_\pCBMs{}$ is a yes-instance then~$|V_1| = |V_2|$, otherwise there could not be a perfect matching. Thus, if~$|V_1| \neq |V_2|$ we may simply output a trivial no-instance for which the statement of the lemma holds. Otherwise, by \autoref{obs:surplussum}, the following procedure can be carried out: Add a new cell~$C_{k+1}$ to~$P$ with~\[{\sum_{i: |C_i \cap V_1| > |C_i \cap V_2|}} |C_i \cap V_1| - |C_i \cap V_2|\] vertices in~$V_1$ and the same number of vertices in~$V_2$, and modify the graph~$G$ and each cell~$C_i \in P$ with~$\alpha := |C_i \cap V_1| - |C_i \cap V_2| > 0$ as follows: Add the new vertices~$v_1, \ldots, v_\alpha$ to~$V_2$ and to the cell~$C_i$, and add an edge from~$v_j$ to a vertex in~$C_{k+1} \cap V_1$ for every~$1 \leq j \leq \alpha$ and such that every vertex in~$C_{k+1}$ gets at most one incident edge. Proceed analogously for cells~$C_i$ with~$\alpha := |C_i \cap V_2| - |C_i \cap V_1| > 0$ by adding vertices to~$V_1$ and adding corresponding edges to~$C_{k+1}$. Finally, expand the weight function~$\wf$ to the new edges by giving each of them weight~0.

This construction is obviously correct, since each new vertex can only be matched to its corresponding vertex in~$C_{k+1}$.

Concerning statement~\enuref{enu:cl12}, assume that the statement does not hold for a instance that contains the vertex partition~$P = \{C_1, \ldots, C_k\}$ and a join set~$J$. We greedily choose two cells~$C_i, C_j$ that are in different connected components in the ``cell-join graph''~$(P, \{\{C_i, C_j\}: \{i, j\} \in J\})$, remove them from~$P$, add the cell~$C_k := C_i \cup C_j$ and update~$J$ accordingly---that is, we replace every join~$\{m, l\} \in J$ where~$m \in \{i, j\}$ by the join~$\{k, l\}$. This is correct because all joins satisfied by any solution~$M$ for the new instance are also satisfied by~$M$ in the original instance and vice versa. Iterating the merging of cells in differ ent connected components makes the cell-join graph connected and the statement follows.
\end{proof}
\paragraph{Description of the Reduction.}
To reduce instances of \pCBMs{} that conform to statement~\enuref{enu:cl11} and~\enuref{enu:cl12} of \autoref{lem:cbmlegalized} to instances of \pWMEEAs{} we use the simple idea of modelling every cell as connected component, vertices in~$V_1$ as vertices with balance~$-1$, vertices in~$V_2$ as vertices with balance~$1$, and joins as hints.
\begin{construction}\label{cons:redcbmtoeea}
  Let the bipartite graph~$B = (V_1 \uplus V_2, E)$, the weight function~$\wf:E \rightarrow [0, \wf_{max}]\cup \{\infty\}$, the vertex partition~$P = \{C_1, \ldots, C_k\}$ and the join set~$J$ constitute an instance~$I_\pCBMs{}$ of \pCBMs{} such that~$I_\pCBMs{}$ corresponds to \autoref{lem:cbmlegalized}\enuref{enu:cl11} and~\enuref{enu:cl12}.

  Let~$v^1_1, v^2_1, \ldots, v^1_{n/2}, v^2_{n/2}$ be a sequence of all vertices chosen alternatingly from~$V_1$ and~$V_2$. Define the graph~$G = (V, A) := (V_1 \cup V_2, A_1 \cup A_2)$ where the arc sets~$A_1$ and~$A_2$ are defined as follows: $A_1 := \{(v_i^1, v_i^2) : 1 \leq i \leq n/2\}$. For every~$1 \leq j \leq k$ let~$C_j = \{v_1, \ldots, v_{j_k}\}$, set~\[A_2^j := \{(v_i, v_{i + 1}): 1 \leq i \leq j_k -1\} \cup \{ (v_{j_k}, v_1\}\] and define~$A_2 := \bigcup_{j = 1}^k A_2^j$. Define a new weight function~$\wf'$ for every pair of vertices~$(u, v) \in V \times V$ by
\[\wf'(u, v) := \begin{cases}
  \wf(\{u, v\}), & u \in V_2, v \in V_1, \{u, v\} \in E \\
  \infty, & \text{otherwise.}
\end{cases}
\]
Finally, derive an advice~$H$ for~$G$ by adding a length-one hint~$h$ to~$H$ for every join~$\{o, p\} \in J$ such that~$h$ consists of the edge that connects vertices in~$\comp{G}$ that correspond to the connected components~$C_o$, and~$C_p$.

The graph~$G$, the weight function~$\wf'$, the maximum weight~$\wf_{max}$ and the advice~$H$ constitute an instance~$I_\pWMEEAs{}$ of \pWMEEAs{}.
\end{construction}
\begin{theorem}\label{the:redcbmtowmeea}
  \pCBM{} is polynomial-parameter polynomial-time many-one reducible to \pWMEEA{} with respect to the parameters join set size and connected components in the input graph.
\end{theorem}
\begin{proof}
  We show that the application of \autoref{lem:cbmlegalized} and \autoref{cons:redcbmtoeea} is such a reduction. It can easily be checked that it can be carried out in polynomial time. Also, by \autoref{lem:cbmlegalized} and the definition of~$A_2$ it follows that the instances of \pWMEEAs{} generated in this way have a number of connected components that is at most the size of the join set plus one.

  Assume that there is a perfect conjoining matching~$M$ with weight at most~$\wf_{max}$ for the instance~$I_\pCBMs{}$ as in \autoref{cons:redcbmtoeea}. Then, we derive an \EE{}~$E$ for~$G$ that heeds the advice with the same weight by simply choosing~$E := \{(u, v) : u \in I^-_G \wedge \{u, v\} \in M\}$. By the definition of~$\wf'$,~$\wf'(E) = \wf(M)$. Every hint is realized by~$E$ because for every join there is an edge in~$M$ that satisfies it. Most importantly,~$E$ is an \EE{} for~$G$: Since~$M$ is perfect, every vertex in~$G$ has exactly one arc incident in~$E$. Since every vertex in~$G$ has balance~$-1$ or~$1$ (due to the definition of~$A_1$), this suffices to make all vertices balanced. By \autoref{lem:cbmlegalized}\enuref{enu:cl12}, the advice~$H$ is a connecting advice and thus~$G + E$ is connected.

  Now assume that there is an \EE{}~$E$ for~$G$ that heeds the advice~$H$ and has weight at most~$\wf_{max}$. Choosing~$M := \{\{u, v\} : (u, v) \in E\}$ yields a perfect conjoining matching of the same weight: It holds the~$\wf'(E) = \wf(M)$, because all extension arcs that do not correspond to an edge in~$B$ have weight~$\infty$. The matching~$M$ is perfect, because every vertex in~$I_G^-$ (in~$I_G^+$) has balance~$-1$ (balance~$1$), has only incoming (outgoing) allowed arcs and thus has exactly one arc incident in~$E$. The matching~$M$ is conjoining, because~$E$ heeds the advice~$H$.
\end{proof}
The reduction given above gives rise to the following parameterized equivalence.
\begin{theorem}\label{the:cbmwmeeequiv}
  \pCBM{} and \pWMEE{} are parameterized equivalent with respect to the parameters join set size and connected components in the input graph. 
\end{theorem}
\begin{proof}
  By \autoref{lem:redwmeetowmeea} there is a parameterized reduction from \pWMEEs{} to \pWMEECAs{} with respect to the parameter number of connected components. By \autoref{the:redwmeecatocbm} there is a parameterized reduction from \pWMEECAs{} to \pCBMs{} with respect to the parameters connected components and join set size.

The other direction follows from the reduction from \pCBMs{} to \pWMEEAs{} given above in \autoref{the:redcbmtowmeea} with respect to the parameters join set size and connected components and the reduction from \pWMEEAs{} to \pWMEEs{} given in \autoref{the:redwmeeatowmee}.
\end{proof}
We also can finally prove \NPhs{} for \pWMEECAs{} which we have deferred up to now.
\begin{corollary}\label{cor:eecanph}
  \pWMEECAs{} is \NPh .
\end{corollary}
\begin{proof}
  We have proven in \autoref{the:cbmnph} that \pCBMs{} is \NPh{} via a reduction from \ptSATs{}. Observe that reducing the instances produced by the corresponding \autoref{cons:red3sattocbm} to instances of \pWMEEAs{} by \autoref{cons:redcbmtoeea} yields instances with minimal connecting advice. Thus there is a reduction from \ptSATs{} to \pWMEECAs{}.
\end{proof}
It turns out that reducing from \pWMEECAs{} to \pCBMs{} and back from \pCBMs{} to \pWMEEAs{} can be interpreted as preprocessing procedure for \pWMEECAs{}:
\begin{observation}\label{obs:redbackpreprocess}
  Successively applying \autoref{con:redwmeeatocbm} and \autoref{cons:redcbmtoeea} to an instance of \pWMEECAs{} yields an equivalent instance of \pWMEECAs{}.
\end{observation}
\begin{proof}
  Recall that in \autoref{con:redwmeeatocbm} connected components are directly modeled by cells in the vertex partition, hints of length one are directly modeled by joins and hints of length at least two by a gadget comprising of a new cell and two joins, both involving the new cell and one of the endpoints of the hint. Thus, in the corresponding instance of \pCBMs{} no join can be removed without ``disconnecting'' one of the cells from the others. Since in \autoref{cons:redcbmtoeea} cells are directly modeled by connected components and joins are directly modeled by hints, it follows that the resulting instance has minimal connecting advice.
\end{proof}
This yields the following two results.
\begin{corollary}\label{cor:eecaprobkernel}
  \pWMEECAs{} has a problem kernel with~$\bigO(b^2c)$~vertices, where~$b$ is the sum of all positive balances and~$c$ is the number of connected components.
\end{corollary}
\begin{proof}
  This follows by simply using \autoref{con:redwmeeatocbm} and \autoref{cons:redcbmtoeea} as preprocessing routines. Since both of them have been proven to be polynomial-time reductions in \autoref{the:redwmeecatocbm} and \autoref{the:redcbmtowmeea} they are correct. Observe that \autoref{con:redwmeeatocbm} disposes of all balanced vertices; for every hint of length at least two there are~$2b^2$ new vertices, which gives the bound of~$\bigO(b^2c)$~vertices in the matching instance. \autoref{cons:redcbmtoeea} does not increase the number of vertices and the statement follows.
\end{proof}
\begin{corollary}
  For every instance of \pWMEECAs{} there is an equivalent instance in which every hint has length one.
\end{corollary}
\begin{proof}
  This trivially follows from \autoref{obs:redbackpreprocess}.
\end{proof}

\section{Discussion}
We now briefly recapitulate the results of this chapter, we note what we could not achieve and we give some directions for further research.

Our considerations in this chapter originally were started with the goal in mind to find out whether \pWMEE{}~(\pWMEEs{}) is fixed-parameter tractable with respect to the parameter number~$c$ of connected components.
Unfortunately, this aim has not been achieved yet. However, we have learned much about the structure of \EE s in \autoref{sec:trails}, and could use this knowledge to derive an efficient algorithm for \pWMEEs{} in \autoref{sec:multivariatealg}. 

In further research, a useful tool for the analysis of \pWMEEs{} with respect to the parameter~$c$ could be the parameterized equivalent matching formulation \pCBM{}~(\pCBMs{}) we derived in this chapter, the final theorems of which are proven in \autoref{sec:releecbm}. We deem that in this formulation the sought solution is more concisely defined. This observation is partly justified by the work laid out in Sections~\ref{sec:trails} through~\ref{sec:wmeeslasha} in order to catch the structure of \EE s, which was necessary to finally derive efficient algorithms and arrive at~\pCBMs{}. Also, only considering the structure of the input graph in \pWMEEs{} may be misleading since, for instance, balanced vertices also take part in the combinatorial explosion of possible paths in \EE s. Balanced vertices, however, do not have equivalents in the corresponding matching instance. The matching formulation makes clear that the structure of allowed extension arcs defined by the weight function is of much greater importance. This is also shown in \autoref{sec:tractislands} where we showed that \pWMEEs{} is actually tractable with respect to the parameter~$c$ for some restricted structure in the allowed arcs. There we used the fact that this structure is precisely captured by the bipartite graph in the matching instance.

Of course we did not stop when we arrived at the matching formulation. We tried multiple approaches for either showing that \pCBMs{} is fixed-parameter tractable or likely intractable with respect to the parameter join set size. However, this has not been crowned with success yet. For instance, we tried to show \claW{1}-hardness via parameterized reductions from \textsc{Multicolored Clique}~\cite{FHRV09} where a graph~$G$, an integer~$k$, and a coloring of the vertices is given and it is asked whether there is a clique~$K$ with at least~$k$ vertices in~$G$ such that each vertex of~$K$ has a distinct color. Here it seemed difficult to copy over the information that one vertex is in the clique from one entity representing this vertex to at least~$k$ of the vertices' neighbors. Reductions from the well-known \textsc{Independent Set} problem suffered from a similar flaw since there it is necessary to copy the information that one vertex is in the independent set over to every neighbor. We also tried reductions from several \claM{1}-complete problems (see, for instance, \citet{FG06}) without much success. 

This led us to the assumption that the bipartite graph in \pCBMs{} and matchings in this graph are too weak to model the relationships of entities in presumably fixed-parameter intractable problems. Thus, we tried to apply some of the well-known techniques to show fixed-parameter tractability. However, we were not able to circumvent running times in the order of~$n^j$ where~$j$ is the size of the join set in these approaches. Subsuming, we are not confident with giving a conjecture on whether or not \pCBMs{} is fixed-parameter tractable. 

\chapter{Incompressibility}\label{sec:incompress}

In this chapter we introduce the problem \pKML{} (\pKMLs) for which there are parameterized reductions to two \EE{} problems. We show that polynomial-size kernels for the extension problems would imply polynomial-size kernels for \pKMLs . However, we also show that polynomial-size problem kernels for \pKMLs{} do not exist unless~$\clacoNP \subseteq \claNP \slashpoly$. 

To prove nonexistence of polynomial-size kernels we use the framework introduced by \citet{BDFH09}: An \emph{or-composition algorithm} for a parameterized problem~$(Q, \kappa)$ over the alphabet~$\Sigma$ is an algorithm that
\begin{asparaenum}%
\item receives a number of instances~$I_1,\ldots,I_m \in \Sigma^*$, with~\[\kappa(I_1) = \ldots = \kappa(I_m)=k\text{,}\]
\item runs in time that is polynomial in~$\sum_{i=1}^m|I_i| + k$, and
\item outputs an instance~$I^* \in \Sigma^*$, such that~$\kappa(I^*)$ is bounded by a polynomial in~$k$ and~$I^* \in Q$ if and only if~$I_j \in Q$ for some~$1\leq j \leq m$.
\end{asparaenum}
A parameterized problem is called \emph{or-compositional} if there is an or-composition algorithm for it. Using a result by \citet{FS11}, it can be shown that if an or-compositional parameterized problem admits a polynomial-size problem kernel, then~$\clacoNP \subseteq \claNP \slashpoly$~\cite{BDFH09}.

To prove or-compositionality for \pKMLs{}, we employ a strategy that has been introduced by \citet{DLS09}. The basic idea is as follows: Prove that the problem is fixed-parameter tractable. In the composition algorithm, when there are many input instances, that is, when~$m$ above is at least as large as the fixed-parameter running time, solve all the instances using this algorithm and output a trivial yes or no-instance. Otherwise, if there are less input instances, use this fact to create an identification for every instance. These identifications then can be used to create a composition instance that consists of parts which correspond to the original instance. 
  
\section{\pKML{}}
First we define \pKML{}~(\pKMLs{}) and show that it is \NPc{} and fixed-parameter tractable. For convenience, we use the following notation.
\begin{definition}
  Let~$C$ be a set of colors. A $C$-\emph{\liste{}} is a multiset with the elements drawn from~$C$. A $C$-\emph{\kiste{}} is a multiset with the elements drawn from all $C$-\listen{}. When the color set is clear from the context, we simply speak of \listen{} and \kisten{}.
\end{definition}
\decprob{\pKML{}}{A set~$C$ of~$c$ colors and~$k$ \kisten{} each containing a number of \listen{}.}{Is it possible to choose exactly one \liste{} in each \kiste{} such that each color in~$C$ is contained in at least one of the chosen \listen{}?}
\begin{example}
  Intuitively one may think of \pKMLs{} as the following problem: Given a number of light bulbs, each with a unique color, and a number of \kisten{}. The \kisten{} can be positioned in exactly one of a number of \listen{} specific to the \kiste{}. In each \liste{}, a \kiste{} lights a defined subset of the light bulbs. The question is, given the light bulbs each \kiste{} lights in each \liste{}, is it possible to choose a \liste{} for each \kiste{} such that all light bulbs are turned on?
\end{example}
Note that defining \kisten{} and \listen{} as multisets instead of plain sets does not add depth to this problem and seems to complicate things at first. However, it simplifies constructions and makes them more convenient to read later on.
\begin{lemma}
  \pKML{} is \NPc .
\end{lemma}
\begin{proof}
  We first show membership in \claNP: An example of a certificate for a yes-instance are the chosen \liste{} in each of the \kisten{}. This certificate is of polynomial size in the input length and, thus, \pKMLs{} belongs to~\claNP .

  \NPhs{} of \pKMLs{} can be seen via a simple reduction from the \textsc{Set Cover} problem. \textsc{Set Cover} has been proven to be \NPh{} by \citet{Kar72}. In \textsc{Set Cover} a set~$S$, a family~$F$ of subsets of~$S$, and an integer~$k$ is given. It is asked whether there is a subfamily~$F' \subseteq F$ such that $|F'| \leq k$ and the union of all sets in~$F'$ equals $S$. To solve \textsc{Set Cover} with \pKMLs{}, introduce a color set~$C$ such that there is a bijection between~$S$ and~$C$ and introduce a \kiste{}~$K$. For every set~$f \in F$ add a \liste{}  to~$K$ that contains the colors corresponding to elements of~$f$. Then $k$~copies of~$K$ form our sought instance of \pKMLs . This instance is polynomial-time constructible because $k \leq |F|$. If there is a solution to the \textsc{Set Cover} instance, we may just choose \listen{} accordingly in the \pKMLs{} instance and vice versa. 
\end{proof}
\begin{lemma}\label{lem:kmlfpt}
  \pKML{} can be solved in time~$\bigO^*(2^{ck})$.
\end{lemma}
\begin{proof}
  An algorithm to solve \pKMLs{} may simply try each combination of \listen{} for all the \kisten{}: We may assume that in every \kiste{} there are at most~$2^c$ \listen{} because \listen{} containing the same colors as other \listen{} may be deleted and multiple copies of one color in one \liste{} may also be deleted. Thus, there are at most~$(2^c)^k$ combinations of \listen{}.
\end{proof}

\section{\pKML{} is Or-Compositional}
We now consider \pKMLs{} parameterized by the number of colors~$c$ and the number of \kisten{}~$k$. In order to prove that \pKMLs{} is or-compositional, we have to give an algorithm as described at the beginning of this chapter. However, if such an algorithm receives~$2^{ck}$~instances as input, it may directly solve all of the instances and return a trivial yes- or no-instance: Let~$m$ be the number of input instances. If~$m \geq 2^{ck}$, solving every instance using the algorithm from \autoref{lem:kmlfpt} takes~$\bigO^*(m2^{ck})$~time. This is polynomial in~$m$. Thus, in the following, we may assume the number~$m$ of instances to be smaller than~$2^{ck}$, implying that $\log(m) \leq ck$. This relation allows us to generate an identification for instances.

\paragraph{Construction Outline.}
The basic idea is to create an instance-chooser by introducing new \kisten{} and colors. Every possible way to choose \listen{} in these new \kisten{} shall correspond to exactly one original instance that then has to be solved. In order to achieve this, the input instances are merged by creating new \kisten{} that contain all the \listen{} of exactly one \kiste{} of every instance. The new colors are then distributed among these merged \kisten{} in order to force every solution for the composite instance to solve the chosen original instance.

\paragraph{Composition Algorithm.}
Let $I_i$,~$0 \leq i \leq m - 1$, be instances of \pKMLs{}, each with~$c$ colors and~$k$ \kisten{}~$K^i_1, \ldots , K^i_k$. For convenience and without loss of generality, we assume that each instance uses the same color-set~$C$. Our composition algorithm for \pKMLs{} works as follows: For each~$1 \leq \alpha \leq k$ and each $1 \leq \beta \leq \log(m)$, introduce two colors~$o^0_{\alpha,\beta}$ and~$o^1_{\alpha,\beta}$. Then, for each~$1 \leq \alpha \leq k$, each~$1 \leq \beta \leq \log(m)$, and for each instance~$I_i$, if the binary encoding of~$i$ has a one at the~$\beta$'th binary place,\footnote{Counting the binary places from the right and starting with~1.} add the color~$o^1_{\alpha,\beta}$ to every \liste{} in the \kiste{}~$K^i_\alpha$, otherwise add the color~$o^0_{\alpha,\beta}$ to every \liste{} in \kiste{}~$K^i_\alpha$. Then, create a new instance $I^*$ by creating \kisten{}~$K^*_\alpha, 1 \leq \alpha \leq k$, where~$K^*_\alpha$ contains each of the modified \listen{} of the \kisten{}~$K^i_\alpha, 0 \leq i \leq m -1$. Finally, introduce \kisten{}~$K_\beta$,~$1 \leq \beta \leq \log(m)$, into the instance~$I^*$, where~$K_\beta$ contains one \liste{} with the colors~$o^0_{1,\beta}, \ldots , o^0_{k,\beta}$ and one \liste{} with the colors~$o^1_{1, \beta}, \ldots,  o^1_{k, \beta}$ and return~$I^*$. See also the pseudocode in \autoref{alg:compkml} and an example of a composite instance in \autoref{fig:kmlexample}.
\begin{algorithm}
  \LinesNumbered
  
  \KwIn{Instances~$I_i$, $0 \leq i \leq m - 1$, of \pKMLs{}, each with~$c$ colors from the set~$C$ and~$k$ \kisten{}~$K^i_1, \ldots , K^i_k$.}
  \KwOut{A composite instance $I^*$.}

  \BlankLine
  \lFor{$1 \leq \alpha \leq k, 1 \leq \beta \leq \log(m)$}{generate two new colors $o^0_{\alpha,\beta}$ and~$o^1_{\alpha,\beta}$\;}

  \For{$0\leq i \leq m-1, 1 \leq \alpha \leq k, 1 \leq \beta \leq \log(m)$}{
    \For{each \liste{}~$L$ in $K^i_\alpha$}{
      \uIf{the binary encoding of~$i$ has a one at place $\beta$}{
        add the color $o^1_{\alpha,\beta}$ to $L$\;}
      \lElse{add the color $o^0_{\alpha,\beta}$ to $L$\;}
    }
  }
  $C' \leftarrow C \uplus \{o^1_{\alpha,\beta},o^0_{\alpha,\beta}:  1 \leq \alpha \leq k, 1 \leq \beta \leq \log(m)\}$\;
  $I^* \leftarrow$ empty \pKMLs{} instance with colors~$C'$\;
  \For{$1 \leq \alpha \leq k$}{
    $K^*_\alpha \leftarrow$ \kiste{} with all \listen{} in the \kisten{}~$K^j_\alpha, 0 \leq j \leq m-1$\;
    Add $K^*_\alpha$ to $I^*$\;
  }
  \For{$1 \leq \beta \leq \log(m)$}{
    $K_\beta \leftarrow$ \kiste{} with the \liste{}~$\{o^0_{1,\beta}, \ldots , o^0_{k,\beta}\}$ and the \liste{}~$\{o^1_{1, \beta}, \ldots,  o^1_{k, \beta}\}$\;
    Add $K_\beta$ to $I^*$\;
  }
  \Return{$I^*$\;}
        
  \SetAlgoRefName{Composite\-SSC}
  \caption{Composition algorithm for \pKMLs{}.}
  \label{alg:compkml}
\end{algorithm}
\begin{figure}
  \begin {center}
    \includegraphics{kmlexample.1}
   \caption{Four instances~$I_0, \ldots, I_3$ of \pKMLs{} and a composite instance~$I^*$ produced by \autoref{alg:compkml}. Each of the instances~$I_0, \ldots, I_3$ contains two \kisten{} each with two \listen{}. In the composite instance~$I^*$ the \kisten{}~$K^0_1, \ldots, K^3_1$ and the \kisten{}~$K^0_2, \ldots, K^3_2$ are merged and their \listen{} extended with new colors (\listen{} shaded according to their original instance). Also, in the composite instance, new \kisten{}~$K_1, K_2$ are introduced that contain only \listen{} with new colors. If there is a solution to either of the input instances, then we can choose the corresponding \listen{} in~$I^*$ and cover the remaining new colors via a \liste{} in~$K_1$ and~$K_2$, respectively. Also, if there is a solution to~$I^*$, it has to choose one \liste{} in~$K_1$ and one in~$K_2$. The remaining new colors have to be covered by the \listen{} in~$K^*_1, K^*_2$. The only way to cover the new colors is to choose \listen{} in~$K^*_1, K^*_2$ that correspond to exactly one of the input instances.}
    \label{fig:kmlexample}  \end{center}
\end{figure}
\begin{lemma}\label{lem:kmlcomp}
  The following statements hold for the new instance~$I^*$:
  \begin{lemenum}
  \item $I^*$ has at most $k + ck$ \kisten{} and at most $k + 2ck^2$ colors. \label{enu:lkmlcomp1}
  \item $I^*$ is computable in time polynomial in the sum of the sizes of the input instances. \label{enu:lkmlcomp2}
  \item $I^*$ is a yes-instance if and only if there is a yes-instance~$I_i, 1 \leq i \leq m$. \label{enu:lkmlcomp3}
  \end{lemenum}
\end{lemma}
\begin{proof}
  Concerning statement~\enuref{enu:lkmlcomp1}: There are $k$~\kisten{}~$K^*_\alpha$ and~$\log(m)$ \kisten~$K_\beta$ in~$I^*$. As we observed at the beginning of this section, $\log(m) \leq ck$. The color-set of~$I^*$ consists of~$k$ colors from the input instances plus~$2k \log(m) \leq 2ck^2$ newly introduced colors (line 7 and 8 in \autoref{alg:compkml}).

  Statement~\enuref{enu:lkmlcomp2} can easily be checked by looking at \autoref{alg:compkml}.

  For statement~\enuref{enu:lkmlcomp3}, first assume that there is a yes-instance~$I_j$ among the input instances. Then, all~$c$ colors in~$C$ can be covered by choosing \listen{} in \kisten{} of $I_j$. Since each \liste{} of the \kisten{}~$K^j_\alpha$,~$1\leq \alpha \leq k$, is extended (lines 2 to 6) and added to the \kisten{}~$K^*_\alpha$ (lines 9 to 11), we can choose the corresponding modified \listen{} in each~$K^*_\alpha$ to cover the colors in~$C$ and the colors~$\{o^{\binary(j,\beta)}_{\alpha,\beta}: 1 \leq \alpha \leq k, 1\leq \beta \leq \log(m)\}$, where~$\binary(j,\beta)$ denotes the digit of the binary encoding of~$j$ at the position~$\beta$. It remains to cover the colors~$\{o^{1-\binary(j,\beta)}_{\alpha,\beta}: 1 \leq \alpha \leq k, 1\leq \beta \leq \log(m)\}$. This can be done by choosing the \listen{} of the form~$\{o^{1-\binary(j,\beta)}_{1,\beta}, \ldots , o^{1-\binary(j,\beta)}_{k,\beta}\}$ in the \kisten{}~$K_\beta, 1\leq \beta \leq \log(m)$.

  Now, assume that $I^*$ is a yes-instance, that is, assume that it is possible to choose exactly one \liste{} in each of the \kisten{} of~$I^*$ in order to cover all colors of~$I^*$. This implies that there is an integer~$0 \leq j \leq m -1$ such that the \listen{} chosen in the \kisten{}~$K_\beta$,~$1 \leq \beta \leq \log(m)$, are of the form~$\{o^{1-\binary(j,\beta)}_{1,\beta}, \ldots , o^{1-\binary(j,\beta)}_{k,\beta}\}$. None of these \listen{} cover any color of~$C$ or~$\{o^{\binary(j,\beta)}_{\alpha,\beta}: 1 \leq \alpha \leq k, 1\leq \beta \leq \log(m)\}$. By construction the colors~$o^{\binary(j,\beta)}_{\alpha,\beta}$,~$1\leq \beta\leq \log(m)$, for some fixed~$\alpha$ occur only in the \kiste{}~$K^*_\alpha$. Furthermore, these colors occur together (that is in one \liste{}) in this \kiste{} only in the \listen{} that were taken from the instance~$I_j$. In order to cover these colors, one of the modified \listen{} of the instance~$I_j$ has to be chosen in~$K^*_\alpha$. This holds for all \kisten{}~$K^*_\alpha$,~$1 \leq \alpha \leq k$, and since all colors of~$C$ are covered, $I_j$ must be a yes-instance.
\end{proof}
\autoref{lem:kmlcomp} shows that \autoref{alg:compkml} is a composition algorithm for \pKMLs{}. Thus, the following theorem follows:
\begin{theorem}
  \pKML{} is or-compositional.
\end{theorem}

\section{Lower Bounds for Problem Kernels}
Using the knowledge we have gained about \pKML{}~(\pKMLs{}), we can give lower bounds on kernel sizes for \EE{} problems. We do this by giving a polynomial-parameter polynomial-time reduction from \pKMLs{} (parameterized by the number of colors~$c$ and the number of \kisten{}~$k$) to \pxDEE{2} (parameterized by the maximum number of extension arcs). 
Then, since both problems are \NPc{}, a problem kernel of polynomial size for \pxDEE{2}~(\pxDEEs{2}) would imply a polynomial problem kernel for \pKMLs{}---we could simply transform an \pKMLs{} instance to a \pxDEEs{2} instance via the parameterized reduction, kernelize it, and then back-transform the underlying non-parameterized \pxDEEs{2} instance to an \pKMLs{} instance with polynomial blow-up since the reduction is polynomial-time computable. Furthermore, because \pxDEEs{2} is a special case of \WMEEs{} (see \autoref{sec:constrainedEE}), we also obtain lower bounds on the kernel sizes for this more general problem.

In this section, we use the symbols~$\prec, \preceq, \succ, \succeq$ for pairs of tuples as ``component-wise $<, \leq, >, \geq$'', respectively. We also frequently use the notion of allowed arcs. For their definition, see \autoref{sec:constrainedEE}.

\paragraph{Reduction Outline.}
The reduction uses the fact that the input graph of an \EE{} problem has to be connected by adding extension arcs. Thus, we model colors of an \pKMLs{} instance as connected components that have to be connected by specific paths consisting of allowed extension arcs. These paths will correspond to the \listen{} in the \pKMLs{} instance. The main tool we use for the construction are ``confined regions'' in which vertices can only be connected to one another via extension arcs inside the region and not to vertices outside of the region.

We continue with an intuitive description of our reduction and give a more detailed one in \autoref{cons:redkmlto2dee}. For the detailed description we need some minor problem restrictions. The descriptions are followed up by an example and after this we give the correctness proof.

\paragraph{Intuitive Description.}The idea behind our construction is as follows: It first creates pairs~$v^1_i, v^2_i$ of unbalanced vertices for every \kiste{}~$K_i$ in the given instance~$I_{\pKMLs{}}$ of \pKMLs{}. These pairs are interconnected via arcs from the set~$A_1$ that form a cycle such that all pairs belong to one single component. Next, for every \liste{}~$L^i_j$ in the \kisten{}, vertices~$w_{i,j,m}$ are introduced that correspond to the colors~$1 \leq m \leq c$ in the \liste{}. The vertices are placed such that any incoming extension arcs can only originate from one of the vertices of the same \liste{} or from the unbalanced vertex corresponding to the \kiste{} the \liste{} is contained in. Analogously, outgoing extension arcs can only target vertices of the same \liste{} or the corresponding unbalanced vertex. Finally, all vertices that correspond to a specific color are interconnected via a directed cycle to create one connected component consisting of balanced vertices for every color. Carrying out these steps, we obtain an instance~$I_{\pxDEEs{2}}$ of \pxDEEs{2}.

In a valid \EE{} for~$I_{\pxDEEs{2}}$ all connected components of the input graph are connected to one another via extension arcs. Observe that in~$\pxDEEs{2}$ there are no cycles in any valid \EE{} because a cycle must include at least one arc that points upwards-right. Thus, the connected components in~$I_{\pxDEEs{2}}$ have to be connected via paths. The placement of the vertices ensures that each such path starts and terminates in the unbalanced vertices corresponding to a single \kiste{} and furthermore any such path traverses only vertices corresponding to a single \liste{}. Also, the placement of the vertices~$v^1_i$ will ensure that there is no allowed incoming extension arc and thus every \EE{} can contain at most one path between~$v^1_i$ and~$v^2_i$. This gives a one-to-one correspondence between \pKMLs{} solutions that cover all colors and \EE s in \pxDEEs{2} that connect all components. 

\paragraph{Problem Restrictions.}
For instances of \pKMLs{} with color sets of cardinality~$c$ we assume that there are exactly~$c$ colors in each \liste{}---this is without loss of generality, because if there are more colors, then we can delete a repeated color; if there are less, then we can repeat an arbitrary color already in the list. We also assume that the number of \listen{} in a \kiste{} is the same for all \kisten{}---we can do this because if there is a \kiste{} with less \listen{} than in another \kiste{}, we can just repeat a \liste{} already present.

\begin{construction}\label{cons:redkmlto2dee}
  Let the color set~$\{o_1, \ldots, o_c\}$ and the \kisten{}~$K_i$,~$1 \leq i \leq k$, each with~$l$ \listen{}~$L^i_j$,~$1 \leq j \leq l$, constitute an instance~$I_{\pKMLs{}}$ of \pKMLs{}. Construct an instance of \pxDEEs{2} as follows: 

  Define the following vertices:
\begin{align*}
  \begin{split}
    v^1_i := (8cil,8c(k-i+1)l)
  \end{split}
  \begin{split}
    v^2_i := v^1_i - (4cl, 4cl)
  \end{split}
\end{align*}
Introduce the vertex set~$V := \{v^1_i, v^2_i: 1 \leq i \leq k\} \cup \{v^1_0, v^2_{k+1}\}$. Connect these vertices using the following arc sets:
\begin{align*}
  A_1 &:= \{(v^1_{i-1}, v^1_{i}),(v^2_{i+1},v^2_{i})  : 1 \leq i \leq k\} \cup \{(v^1_k, v^2_{k + 1}), (v^2_1, v^1_0)\}\\
  A_2 &:= \{(v^2_i, v^1_i) : 1 \leq i \leq k\}
\end{align*}
Furthermore, for every $1 \leq i \leq k, 1 \leq j \leq l, 1 \leq m \leq c$, define the following vertices:\[w_{i, j, m} := v^2_i + (0, 4cl) + (2c(2j - 1) - 2(m -1), - 2c(2j - 2) - 2(m - 1))\] 
For every \liste{}~$L^i_j$, let $o^{i,j}_1, \ldots, o^{i, j}_c$ be the colors~$L^i_j$ contains and introduce the vertex set~$\{w_{i, j, m} : 1 \leq m \leq c\}$. Let~$W_n := \{w_{i,j,m} : o^{i, j}_m = o_n\}$ and let~$w^1_n, \ldots , w^p_n$ be a total ordering of~$W_n$. For every~$1 \leq n \leq c$ introduce the following arc set: \[B_n := \{(w^i_n, w^{i + 1}_n):1 \leq i \leq p - 1\} \cup \{(w^p_n, w^1_n)\}\]
The graph~$G := (V \cup \bigcup^c_{n = 1}W_n, A_1 \cup A_2 \cup \bigcup^c_{n = 1}B_n)$ and the integer~$(c+1)k$ constitute an instance~$I_{\pxDEEs{2}}$ of \pxDEEs{2}.
\end{construction}%
\begin{figure}
  \begin {center}
    \subfloat[\pKMLs{} instance]{
      \includegraphics{kmlexample.2}
      \label{fig:redkmlto2deea}
    } \\
    \subfloat[\pxDEEs{2} instance]{
      \includegraphics{kmlexample.3}
      \label{fig:redkmlto2deeb}}
    \caption{Example application of \autoref{cons:redkmlto2dee} explained in \autoref{ex:redkmlto2dee}.}\label{fig:redkmlto2dee}
  \end{center}
\end{figure}%
\begin{example}\label{ex:redkmlto2dee}
  Consider \autoref{fig:redkmlto2dee}. In \autoref{fig:redkmlto2deea}, an instance~$I_\pKMLs{}$ of \pKMLs{} is shown. It contains two \kisten{}~$K_1, K_2$ each with two \listen . Below this, you can see an instance~$I_\pxDEEs{2}$ of \pxDEEs{2} produced from~$I_\pKMLs{}$ by \autoref{cons:redkmlto2dee}---that is, a directed graph embedded in two-dimensional space.\footnote{Not to scale, the coordinates used in \autoref{cons:redkmlto2dee} are simplified for readability.} It comprises a number of vertices represented by circles that may be connected by arcs. The big circles represent vertices that correspond to colors. The number of the color is written in the top half and the vertex name in the bottom half of the circle. Additionally, we see rectangles shaded in gray (``\kiste{} regions'') and white rectangles (``\liste{} regions'').

For the \kiste{}~$K^1_1$ the pair of vertices~$v^1_1, v^2_1$ is introduced in~$I_\pxDEEs{2}$ and these vertices are made unbalanced via the arc~$(v^2_1, v^1_1)$. Analogously this is done for the second \kiste . Additionally, two helper-vertices~$v^1_0, v^2_3$ are introduced. They simply ensure that the built graph does not include multiple arcs and remains a simple directed graph. Next, for every \liste , there are vertices corresponding to their colors. For example, the \liste~$L^1_1 = \{1, 2, 2\}$ corresponds to the vertices~$w_{1,1,1}, w_{1,1,2}, w_{1,1,3}$---by \autoref{cons:redkmlto2dee} this is due to an arbitrary ordering of the colors in the \liste , but we stick to the top-to-bottom ordering given by \autoref{fig:redkmlto2deea} here. All vertices that correspond to one single color are connected by a directed cycle. For instance this is the case for the vertices~$w_{1,1,2}, w_{1,1,3}, w_{1,2,3}$ which correspond to the color~$2$. 

Consider the solution to~$I_\pKMLs{}$ that chooses the \listen{}~$L^1_1, L^2_1$. To solve~$I_\pxDEEs{2}$ we can just add the directed paths from~$v^1_1$ to~$v^2_1$ and from~$v^1_2$ to~$v^2_2$ also traversing the vertices corresponding to the \listen~$L^1_1, L^2_1$, respectively. Also, since the only allowed arcs in~$\pxDEEs{2}$ point downwards-left, any solution to~$I_\pxDEEs{2}$ consists of paths from the \kiste~vertices~$v^1_i$ to the vertices~$v^2_i$ that traverse other vertices that correspond to exactly one \liste{}. Also, in every \kiste-region at most one path can be in an \EE{}, because the vertices~$v^1_1, v^1_2$ have to be balanced and there are no allowed incoming extension arcs. Since the paths connect all connected components of~$I_\pxDEEs{2}$, the corresponding \listen{} contain all colors of~$I_\pKMLs{}$. Hence, these \listen{} form a solution to~$I_\pxDEEs{2}$.
\end{example}%
\paragraph{Correctness.}In order to prove the soundness of the construction, we first make some observations about the placement of the vertices. Then, a useful implication of this placement is observed. We then proceed to show that the above-mentioned paths between the unbalanced vertices~$v^1_i$ and~$v^2_i$ are the only allowed arcs in valid \EE s and, using this, we derive the soundness. The following two observations formulate the notion of regions in the instance.
\begin{observation}\label{obs:lowbound_big_rectangle}
  The vertices~$w_{i,j,m}$ are contained in the rectangle spanned by all points~$p \in \mathbb{Q}^2$ with~$v^2_i \preceq p \preceq v^1_i$.
\end{observation}
\begin{proof}
 Consider the following difference:
\begin{align*}
  {v^1_i} - {w_{i,j,m}} 
  & = 
  \begin{pmatrix}
    4cl - 2c(2j - 1) + 2(m - 1)\\
    4c(j - 1) + 2(m - 1)
  \end{pmatrix}
  ^\top 
\end{align*}
Both coordinates are positive because~$1 \leq j \leq l$. Analogously for~$v^2_i$ and~$w_{i,j,m}$:
\begin{align*}
  w_{i,j,m} - v^2_i 
  & = 
  \begin{pmatrix}
    2c(2j - 1) - 2(m - 1)\\
    4cl - 4c(j - 1) - 2(m - 1)
  \end{pmatrix}
  ^\top
\end{align*}
The first coordinate is positive because~$1 \leq j$ and~$1 \leq m \leq c$. The second coordinate is positive because~$j \leq l$ and~$m \leq c$. 
\end{proof}
\begin{observation}\label{obs:lowbound_small_rectangle}
  The vertices~$w_{i,j,m}$ are contained in the rectangle spanned by all points~$p \in \mathbb{Q}^2$ with~$w_{i, j, c}  \preceq p \preceq w_{i, j, 1}$. Moreover,~$w_{i, j, m} \succeq w_{i, j, m'}$ for~$m' > m$.
\end{observation}
\begin{proof}
  It is clear that~$w_{i, j, m} - w_{i, j, m'}$, where~$m' > m$ is positive, because~$m$ and~$m'$, respectively, have a negative sign in the definition of these vertices. The statement follows because~$1 \leq m \leq c$.
 \end{proof}
We call the corresponding rectangle spanned by~$v^1_i, v^2_i$ the \emph{region of \kiste{}~$K_i$} and the rectangle spanned by~$w_{i, j, 1}, w_{i, j, c}$ the \emph{region of \liste{}~$L^i_j$}.
The intent of this placement is to exploit the following observation.
\begin{observation}\label{obs:lowbound_exclusive_rectangle}
  Let $x,b,d \in \mathbb{Q}^2$ with~$b \preceq d$ and neither~$x \preceq d$ nor~$b \preceq x$. Then, for every~$b \preceq c \preceq d$ it holds that neither~$x \preceq c$ nor~$c \preceq x$.
\end{observation}
\autoref{obs:lowbound_exclusive_rectangle} directly follows from looking at \autoref{fig:region}.
\begin{figure}
  \begin{center}
    \includegraphics{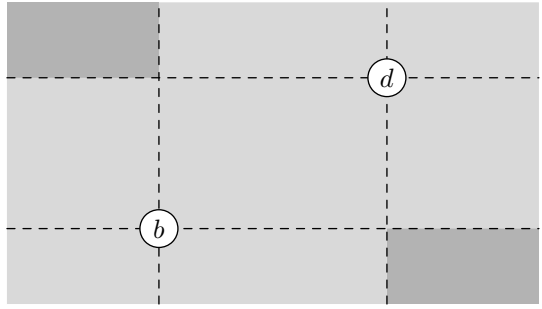}
    \caption{The possible placements for tuples~$y \in \mathbb{Q}^2, y \preceq c \vee c \preceq y$ for any~$c$ such that~$b \preceq c \preceq d$ are colored in light grey. This region does not intersect with the region of possible locations for~$x$ colored in dark grey (not including the dashed lines).}
    \label{fig:region}
  \end{center}
\end{figure}
\begin{lemma}\label{lem:lowbound_allowed_arcs}
  Any allowed extension arc in the instance~$I_{\pKMLs}$ starts in~$v^1_i$ or~$w_{i,j,m}$ for some~$i, j, m$ and ends in either~$v^2_i$ or~$w_{i,j,m'}$ with~$1 \leq m < m' \leq c$.
\end{lemma}
\begin{proof}
  In \pxDEEs{2} arcs~$(u, v)$ are allowed extension arcs if~$v \preceq u$. Because of \autoref{obs:lowbound_big_rectangle} and \autoref{obs:lowbound_small_rectangle} we know that the above mentioned arcs are allowed arcs. It remains to show that they are the only allowed arcs. We first look at arcs between vertices that belong to regions of two different \kisten{}~$K_i, K_{i'}$. For them---by \autoref{obs:lowbound_exclusive_rectangle}---it suffices to show that for~$x \in \{v^1_i, v^2_i\}, y \in \{v^1_{i'}, v^2_{i'}\}$, neither~$x \preceq y$ nor $y \preceq x$. This is the case if in~$x - y$ one coordinate is positive and one negative:
\begin{align*}
  {v^1_i} - {v^1_{i'}} = v^2_i - v^2_{i'} 
  &=
  \begin{pmatrix}
    8c(i - i')l \\
    -8c(i - i')l
  \end{pmatrix}
  ^\top \\
  {v^1_i} - {v^2_{i'}} 
  &=
  \begin{pmatrix}
    8c(i - i' + \nicefrac{1}{2})l \\
    -8c(i - i' - \nicefrac{1}{2})l
  \end{pmatrix}
  ^\top 
\end{align*}%
Concerning arcs between vertices belonging to regions of two \listen~$L^i_j, L^i_{j'}$: By \autoref{obs:lowbound_exclusive_rectangle} it suffices to show that for~$x \in \{w_{i, j, 1},w_{i, j, c}\}, y \in \{w_{i, j', 1}, w_{i, j', c}\}$ neither~$x \preceq y$ nor $y \preceq x$. Again, we look at the differences and observe that one coordinate is positive and one negative:
\begin{align*}
  w_{i,j, 1} - w_{i,j', 1} = w_{i, j, c} - w_{i, j', c} 
   &= 
  \begin{pmatrix}
    4c(j - j') \\
    -4c(j - j')
  \end{pmatrix}
  ^\top \\
  w_{i,j, 1} - w_{i,j', c} 
  &= 
  \begin{pmatrix}
    4c(j - j') + 2(c - 1) \\
    -4c(j - j') + 2(c - 1)
  \end{pmatrix}
  ^\top \tag*{\qedhere}
\end{align*}
\end{proof}
\begin{lemma}
  \autoref{cons:redkmlto2dee} is a polynomial-parameter polynomial-time many-one reduction. 
\end{lemma}
\begin{proof}
  It is clear that the parameter of the target instance is polynomial in the parameter of the original instance. It can also easily be checked that \autoref{cons:redkmlto2dee} can be carried out in polynomial time.

  Concerning the correctness, first assume that the original instance~$I_{\pKMLs{}}$ is a yes-instance. Thus, there is a sequence of~$k$ \listen{}~$L^{1}_{j_1}, \ldots , L^k_{j_k}$ that cover every color in~$C$. For every~$L^\alpha_{j_\alpha}$ take the arcs of the path~$v^1_\alpha, w_{\alpha, j_\alpha, 1}, \ldots , w_{\alpha, j_\alpha, c}, v^2_\alpha$ into the arc set~$E$. The arc set~$E$ is a solution to~$I_{\pxDEEs{2}}$ because of the following: First, $E$ contains exactly~$(c + 1)k$ arcs. Second, every arc of the paths is an allowed arc because of \autoref{obs:lowbound_big_rectangle} and \autoref{obs:lowbound_small_rectangle}. Third, since the \listen{} cover all colors, these paths connect all components of the graph~$G$, that is, $V, W_1, \ldots, W_c$. Fourth, since every vertex~$v^1_i$ (every vertex~$v^2_i$) has exactly one arc starting (ending) in it in~$E$, the graph~$G$ has no unbalanced vertices when adding the arcs of~$E$.

  Now assume that~$I_{\pxDEEs{2}}$ is a yes-instance. In any \EE{}~$E$ for $G$, every vertex~$v^1_i$ (every vertex~$v^2_i$) has exactly one incident outgoing (incoming) arc, because it has balance 1 (-1) and there are no allowed incoming (outgoing) arcs (\autoref{lem:lowbound_allowed_arcs}). Because there are no allowed extension arcs between vertices in different \kiste{} regions (\autoref{lem:lowbound_allowed_arcs}) and thus $E$ consists of a series of~$k$~maximal-length paths~$p_1, \ldots, p_k$. Each path~$p_i$ starts in the vertex~$v^1_i$, ends in the vertex~$v^2_i$ and traverses a subset of vertices in the region of exactly one \liste . This is because in every region of a \kiste{}~$K_i$, the only allowed arcs starting in~$v^1_i$ lead to either~$v^2_i$ or a vertex~$w_{i,j,m}$ in the region of the \liste{}~$L^i_j$ and there are no allowed arcs that lead from one \liste{} region to another (\autoref{lem:lowbound_allowed_arcs}). Let~$L^1_{j_1}, \ldots , L^2_{j_k}$ be the \listen{} corresponding to the \liste{} region the paths~$p_1, \ldots, p_k$ traverse vertices in (if $p_i$ traverses no vertices besides~$v^1_i, v^2_i$, let~$L^i_{j_i}$ be an arbitrary \liste{} in the \kiste{}~$K_i$). Choosing these \listen{} covers all colors in~$I_\pKMLs{}$ because the paths connect all connected components of the graph~$G$.
\end{proof}
Using this reduction, the following theorems now arise:
\begin{theorem}
  \pKML{} parameterized by the number of colors and the number of \kisten{} is polynomial-time polynomial-parameter reducible to \pxDEE{2} parameterized by the number of extension arcs.
\end{theorem}
\begin{theorem}\label{the:nonpoly}
  \pxDEE{2} parameterized by the maximum number of extension edges in a solution does not have a polynomial problem kernel, unless~$\clacoNP \subseteq \claNP \slashpoly $ and thus~\claPH~=~\claS{3}.
\end{theorem}
Using this theorem, we also can easily derive the following corollary.
\begin{corollary}
  \pWMEE{} parameterized by the number of components and/or the sum of all positive balances of vertices does not have a polynomial problem kernel, unless~$\clacoNP \subseteq \claNP \slashpoly $.
\end{corollary}
\begin{proof}
  This is because the number~$c$ of components in the input graph and the sum~$b$ of all positive balances is bounded by the maximum number~$k$ of extension edges in a solution. Observe that any \EE{}~$E$ for the input graph has to connect all connected components. Hence,~$c - 1 \leq | E|$ and thus~$c - 1 \leq k$. Also,~$E$ has to balance every vertex and thus has to include~$d$ arcs for every vertex of balance~$d$. Thus,~$b \leq |E|$ and thus~$b \leq k$.

If there were a polynomial problem kernel for either the parameters~$c$, or~$b$ or the combined parameter~$b, c$, this would imply a polynomial problem kernel for the parameter~$k$ and the statement follows from \autoref{the:nonpoly}.
\end{proof}

\chapter{Conclusion}\label{sec:conclusion}
In this thesis, we have gained insight into the structure of \EE s in \autoref{sec:trails}. We benefited from this knowledge in \autoref{sec:multivariatealg} in that we were able to give an efficient parameterized algorithm for the problem \pWMEE{}~(\pWMEEs{}) with~$\bigO(4^{c\log(bc^2)}n^2(b^2 + n\log(n)) + n^2m)$ running time. Here,~$c$~is the number of components in the input graph and~$b$ is the sum of all positive balances of vertices in the input graph. 

We also gave a reformulation of \pWMEEs{} parameterized by~$c$ in terms of the natural matching problem \pCBM{} in \autoref{sec:matching}. This formulation might help to attack the fixed-parameter tractability of \pWMEEs{} with respect to parameter~$c$ from a different angle, and we already gave some partial tractability results in \autoref{sec:tractislands}. 

Finally, we studied polynomial-time preprocessing routines for \pWMEEs{} with respect to either of the parameters~$k$,~$b$, and~$c$, and showed that such routines cannot yield a polynomial-size problem kernel unless~$\clacoNP \subseteq \claNP \slashpoly$. Thus, polynomial-size problem kernels for \pWMEEs{} would imply that the polynomial hierarchy collapses to the third level~\cite{Yap83}. 

\paragraph{Outlook.} The most interesting open question is whether \pWMEEs{} is fixed-parameter tractable with respect to the parameter~$c$. By the equivalence of \pWMEEs{} and \pCBM{}~(\pCBMs{}) we gave in \autoref{sec:matching}, this question is equivalent to whether \pCBMs{} is fixed-parameter tractable with respect to the parameter ``join set size''. Intuitively, \pCBMs{} is a very natural problem---modelling for example the job assignment problem, where at least one worker of a particular profession must be assigned to a facility of a specific type. This makes work for \pCBMs{} particularly interesting. A way to attack \pCBMs{} could be by delving into the world of hypergraph transversals~\cite{Mir71}, since \pCBMs{} can be seen as a colored variant of the hypergraph transversal problem. Another way of gaining a deeper understanding of \pCBMs{} could be to search for formulations of this problem that show that it is contained in \claW{t} for some constant~$t$.

We observed that the parameters~$b$ and~$c$ are upper bounded by the parameter~$k$ used by \citet{DMNW10} in their algorithm for \pWMEEs{} with running time~$\bigO(4^kn^4)$. Their algorithm uses a dynamic programming approach and likely uses exponential space. We think that our algorithm for \pWMEEs{} with running time~$\bigO(4^{c\log(bc^2)}n^2(b^2 + n\log(n)) + n^2m)$ can be implemented as to use only polynomial space. In this regard, it would be interesting to see how both algorithms perform on sets of practical instances.

In this thesis, we focussed on directed \EE{} problems. However, the undirected variant of \pWMEEs{} is also \NPh{}---the \NPhs{} proof we gave in \autoref{sec:relhcwmee} canonically transfers over to the undirected variants. It would be interesting to see whether our fixed-parameter tractability results also carry over to the undirected problem---we conjecture that this is the case. We also think that the equivalence to a matching problem can be shown in a similar fashion to our observations in \autoref{sec:matching}. However, we think that the corresponding matching formulation will be a non-bipartite version of \pCBMs{}, intuitively making the undirected variant of \pWMEEs{} a harder problem with respect to parameter~$c$.

We also limited ourselves to adding edges in order to make graphs Eulerian in this thesis. However, there are other natural variants, for instance, deleting edges, editing edges---that is deleting or adding edges---, removing vertices, or any combination of those. Mixed graphs, containing both arcs and edges, are also an option.

We did not consider approximation algorithms in this work. However, given the relationships of \pWMEEs{} to \pHC{} and the \pRP{} problems, both of which admit constant-factor polynomial-time approximation algorithms in some special cases, it would be interesting to analyze \pWMEEs{} in this regard. Then again, the reduction from \pKML{} to \pWMEEs{} we gave in \autoref{sec:incompress} might refute attempts in this direction, because \pKML{} is a variant of \textsc{Set Cover} and this problem is likely not constant-factor approximable~\cite{LY95}.

We merely touched the topic of constrained \EE s in \autoref{sec:constrainedEE}. We would like to remark that this topic might make for an interesting field of research: It is likely well-motivated through practical problems and we expect the stronger structural restrictions to be exploitable for efficient algorithms. Problems there could be augmenting transitive graphs to \Eu transitive graphs or the like.

\bibliographystyle{abbrvnat}
\bibliography{steinbruch.bib}

\selectlanguage{ngerman}
\chapter*{Selbstständigkeitserklärung}
 Ich erkläre, dass ich die vorliegende Arbeit selbstständig und nur unter Verwendung der angegebenen Quellen und Hilfsmittel angefertigt habe.

\vspace{2cm}
\noindent Jena, den \today,\hspace{3cm} Manuel Sorge

\end{document}